\documentclass[aps,prb,reprint,superscriptaddress,longbibliography,floatfix]{revtex4-2}

\usepackage[T1]{fontenc}
\usepackage[utf8]{inputenc}
\usepackage{amsmath,amssymb,amsthm}
\usepackage{array}
\usepackage{graphicx}
\usepackage{hyperref}
\usepackage{bibunits}

\hypersetup{hidelinks,
  pdftitle={Asymptotic Pseudospectra in Dissipative Floquet Quantum Systems: Geometric Structures and Observable Dynamics}}
\makeatletter
\let\SMsaved@author\author
\let\SMsaved@affiliation\affiliation
\let\SMsaved@thanks\thanks
\let\SMsaved@maketitle\maketitle
\let\SMsaved@AFjoin\@AF@join
\let\SMorig@addtocontents\addtocontents
\long\def\addtocontents#1#2{%
  \def\@tempa{toc}\def\@tempb{#1}%
  \ifx\@tempa\@tempb\else\SMorig@addtocontents{#1}{#2}\fi
}
\newcommand{\SMtableofcontents}{%
  \let\addtocontents\SMorig@addtocontents
  \tableofcontents
}
\makeatother

\newtheorem{proposition}{Proposition}
\newtheorem{lemma}{Lemma}
\newtheorem{assumption}{Assumption}

\newcommand{\suppref}[1]{SM Sec.~\ref{#1}}
\newcommand{\mainref}[1]{\ref{#1}}
\newcommand{\maineqref}[1]{\eqref{#1}}

\newcommand{\id}{\mathbb I}
\newcommand{\D}{\mathbb D}
\newcommand{\Tr}{\operatorname{Tr}}
\newcommand{\tr}{\operatorname{tr}}
\newcommand{\Ad}{\operatorname{Ad}}
\newcommand{\dist}{\operatorname{dist}}
\newcommand{\spec}{\operatorname{spec}}
\newcommand{\smin}{s_{\min}}
\newcommand{\eps}{\varepsilon}
\newcommand{\cL}{\mathcal L}
\newcommand{\cE}{\mathcal E}
\newcommand{\dd}{\mathrm d}
\newcommand{\ii}{\mathrm i}
\newcommand{\ee}{\mathrm e}
\newcommand{\DG}{\mathrm d\Gamma}
\newcommand{\norm}[1]{\left\lVert#1\right\rVert}
\newcommand{\pclass}{\mathsf P}
\newcommand{\Iclass}{\mathsf I}
\newcommand{\Cclass}{\mathsf C}

\begin{document}

\begin{bibunit}[apsrev4-2]
\title{Asymptotic Pseudospectra in Dissipative Floquet Quantum Systems: \\ Geometric Structures and Observable Dynamics}

\author{Yuncheng Xie}
\thanks{These authors are co-first authors and contributed equally to this work.}
\affiliation{Department of Physics, Fudan University, Shanghai 200433, China}
\affiliation{Key Laboratory of Computational Physical Sciences (Ministry of Education), State Key Laboratory of Surface Physics,
Fudan University, Shanghai 200433, China}

\author{Haozhe Shi}
\thanks{These authors are co-first authors and contributed equally to this work.}
\affiliation{Department of Physics, Fudan University, Shanghai 200433, China}
\affiliation{Key Laboratory of Computational Physical Sciences (Ministry of Education), State Key Laboratory of Surface Physics,
Fudan University, Shanghai 200433, China}

\author{Zhuocheng Ma}
\affiliation{State Key Laboratory of Artificial Microstructure and Mesoscopic Physics, School of Physics, Peking University, Beijing 100871, China}

\author{Weibin Chu}
\affiliation{Department of Physics, Fudan University, Shanghai 200433, China}
\affiliation{Key Laboratory of Computational Physical Sciences (Ministry of Education), State Key Laboratory of Surface Physics,
Fudan University, Shanghai 200433, China}

\author{Xin-Gao Gong}
\affiliation{Department of Physics, Fudan University, Shanghai 200433, China}
\affiliation{Key Laboratory of Computational Physical Sciences (Ministry of Education), State Key Laboratory of Surface Physics,
Fudan University, Shanghai 200433, China}
\date{\today}

\begin{abstract}
In periodically driven open quantum systems, nonnormality renders the Floquet spectrum insufficient as the system approaches the thermodynamic limit, so that pseudospectra are needed to characterize the dynamics accurately.
While conventional approaches mainly focus on the local dynamics of isolated pseudospectra, their global connections and the resulting physical consequences for observables have remained largely unexplored. Here, we uncover this collective behavior by classifying the unit disk into distinct domains of exponential, algebraic, and bounded accuracy according to the asymptotic size-scaling laws of pseudospectral residuals, yielding an underlying geometric structure. We demonstrate the physical implications of this structure through two exactly solvable models. First, in a dissipative shift chain, parameter tuning drives geometric transitions that are detectable via spin-wave observables. Second, in a chiral XY model, we use this geometric structure to explain the origin of a measurable phenomenon: two correlation signals exchange their retention order under continuous parameter tuning. Our findings not only establish a new theoretical paradigm for understanding dissipative Floquet quantum systems, but also predict geometry-driven observables for quantum computing experiments.
\end{abstract}

\maketitle
\raggedbottom

\section{Introduction}
\label{sec:intro}
The study of the open quantum systems~\cite{Lindblad1976,GKS1976}, especially those continuously driven by periodic forces and interacting with their environment~\cite{Bukov2015,OkaKitamura2019,Eckardt2017,Mori2023}, is essential for uncovering the fundamental physics of quantum information and ultimately realizing universal quantum computing~\cite{VerstraeteWolfCirac2009,Wendin2017,Barreiro2011,Kessler2021,Chen2025,Zhang2024,Bloch2012,Mi2022,Jiang2018}. To predict how these complex systems evolve over time, researchers have traditionally relied on exact eigenvalues~\cite{Minganti2018,Prosen2008,Koch2026}. However, this standard paradigm fails to accurately describe the nonnormal system's behavior, especially as the system size increases~\cite{MoriShirai2020,Haga2021,Ughrelidze2024}. To overcome this limitation, researchers introduce pseudospectra, a more general framework that uses approximate states to capture the missing dynamical information~\cite{TrefethenEmbree2005,Trefethen1997,Dencker2004,Davies2003,Kiorpelidis2025}. While it is understood that these approximate modes are essential for the evolutionary dynamics of the systems, the deeper question of how all these pseudospectra are globally organized and connected has remained entirely unexplored.

In this work, we introduce the concept of the \emph{asymptotic pseudospectral atlas} to systematically map the validity regimes of these approximate modes. By evaluating the size-scaling law of the optimal residual at each point in the closed unit disk, we classify the complex plane into distinct extended domains. These regions are characterized by whether the residual decays exponentially, algebraically, or remains bounded as the system size increases. The geometry of this atlas, including its boundaries, gaps, and contacts, dictates the connections of approximate modes and their response families.

To bridge this geometric classification with measurable physical properties, we apply our analysis to two exactly solvable models. First, we consider a dissipative shift chain. We show that parameter tuning drives geometric transitions in the atlas. These transitions involve the merging of adjacent domains and the closing of an enclosed hole, each leaving a measurable signature in the spin-wave observables.

Second, we investigate a chiral XY model subjected to engineered reservoir couplings~\cite{MetelmannClerk2015,Yang2022,McDonald2018,Wanjura2020}. In this model, we use the atlas geometry to explain a counterintuitive phenomenon: under continuous parameter tuning, two physical correlation signals exchange their retention order even though the entire Floquet spectrum remains invariant. Our analysis reveals that this signal inversion is governed by a contact between distinct domains in the atlas, providing a geometric origin for the dynamics.

These results illustrate that the asymptotic pseudospectral atlas is not merely a formal classification, but a physical structure that dictates the observable dynamics. Section~\ref{sec:theory} constructs the atlas and details its geometric properties. Sections~\ref{sec:toy} and~\ref{sec:xy} realize the geometry and its observable consequences in the dissipative shift chain and the chiral XY model, respectively. Finally, Section~\ref{sec:conclusions} turns to the challenges and potential numerical methods for evaluating the atlas in generically interacting driven chains.

The Supplemental Material collects the mathematical details and supporting proofs.

\section{Asymptotic pseudospectral atlas}\label{sec:theory}

\subsection{Construction of The Atlas}
\label{sec:atlas-definition}
Tracing out the environmental reservoir yields an open quantum system characterized by a density operator $\rho_N$, whose reduced dynamics are completely positive and trace-preserving (CPTP)~\cite{Lindblad1976,GKS1976,Pollock2018}. For an $N$-spin system governed by a $T$-periodic Markovian generator~\cite{Schnell2020,Mori2023} such that $\dot{\rho}_N(t) = \mathcal{L}_N(t)\rho_N(t)$, its associated generator and one-period evolution channel are defined as follows:
\begin{equation}
 \begin{gathered}
 \cL_N(t)X=-\ii[H_N(t),X]+\sum_\mu\mathcal D[J_{\mu,N}(t)]X,\\
 \mathcal D[J]X=JXJ^\dagger-\tfrac12\{J^\dagger J,X\},\\
 \Phi_N=\mathcal T\exp\!\left[\int_0^T\cL_N(t)\,\dd t\right].
 \end{gathered}
 \label{eq:main-open-system}
\end{equation}
Here $H_N$ generates the coherent evolution, while the jump operators $J_{\mu,N}$ describe the bath coupling. Furthermore, the overall evolution is best captured by directly defining the one-period channel $\Phi_N$. Throughout this work, we consider unital channels such that $\Phi_N(I)=I$, which ensures contractivity under the normalized Hilbert--Schmidt (HS) inner product, $\langle X,Y \rangle_N=2^{-N}\Tr(X^\dagger Y)$.

A Floquet multiplier $z$ characterizes an approximate mode when a nonzero operator $X$ nearly satisfies $\Phi_N X = zX$. The accuracy of this approximation is rigorously quantified by the residual and its minimum over the entire operator space:
\begin{align}
 \epsilon_N(z;X)&=\frac{\norm{(\Phi_N-zI)X}_{2,N}}{\norm X_{2,N}},\nonumber\\
 \eps_N(z)&=\min_{X\ne0}\epsilon_N(z;X)=\smin(\Phi_N-zI).
 \label{eq:intro-residual}
\end{align}

Conventionally, the tolerance-$\eta$ pseudospectrum $\{z:\eps_N(z)<\eta\}$ defines a specific level set of the residual field at a finite size $N$~\cite{TrefethenEmbree2005,Trefethen1997}. Here, working within the closed unit disk $\overline{\D}=\{z:|z|\le1\}$ where all exact multipliers reside, we instead characterize each point $z$ by the asymptotic law of its residual $\eps_N(z)$ in the thermodynamic limit. Guaranteed by the Lipschitz continuity $|\eps_N(z)-\eps_N(z')|\le|z-z'|$~\cite{TrefethenEmbree2005}, this asymptotic behavior varies controllably across the disk, naturally partitioning it into distinct extended domains. The geometric shapes and relative configurations of these domains form the basis of our subsequent analysis.

To systematically extract the asymptotic law defining these domains, we first introduce the logarithmic residual field:
\begin{equation}
F_N(z)=-\ln\eps_N(z),
\label{eq:main-F}
\end{equation}
which is taken to be $+\infty$ at an exact zero. Based on this field, we then define three pairs of asymptotic indices:
\begin{equation}
 \begin{aligned}
 d_\pm(z)&=\underset{N\to\infty}{\limsup/\liminf}\,\eps_N(z),\\
 \alpha_\pm(z)&=\underset{N\to\infty}{\limsup/\liminf}\,
                  \frac{\ln[1+F_N(z)]}{\ln N},\\
 p_\pm(z)&=\underset{N\to\infty}{\limsup/\liminf}\,
             \frac{F_N(z)}{\ln N}.
 \end{aligned}
 \label{eq:main-six}
\end{equation}
These limits are evaluated at a fixed $z$, where the $+$ and $-$ subscripts denote the $\limsup$ and $\liminf$, respectively. The pair $d_\pm(z)$ tests whether the residual is bounded away from zero or vanishes. When the residual vanishes, $\alpha_\pm(z)$ characterizes the power of $N$ controlling its exponential decay, while $p_\pm(z)$ captures the algebraic decay rate once that exponential power drops to zero. Distinguishing between the upper and lower limits explicitly preserves any oscillatory dependence on specific subsequences of the system size~\cite{Bingham1987}.

In typical regular physical families, size-dependent oscillations vanish as $N \to \infty$, allowing the upper and lower bounds to converge. We thus impose the \emph{common-limit hypothesis}:
\begin{equation}
 \begin{aligned}
 d_-(z)&=d_+(z)=d(z),\\
 \alpha_-(z)&=\alpha_+(z)=\alpha(z),\\
 p_-(z)&=p_+(z)=p(z).
 \end{aligned}
 \label{eq:main-common-limits}
\end{equation}
Under this hypothesis, the system acquires a single index $(d(z),\alpha(z),p(z))$. This index naturally partitions the closed unit disk $\overline{\D}$ into the domains of our atlas. Specifically, we classify the asymptotic behaviors into three regular classes:
\begin{equation}
 \begin{array}{c|ccc}
   &d(z)&\alpha(z)&p(z)\\ \hline
  \Cclass&>0&0&0\\
  \Iclass&0&>0&+\infty\\
  \pclass&0&0&>0\text{ finite}
 \end{array}
 \label{eq:main-classes}
\end{equation}
The $\Cclass$ region is characterized by a bounded residual $\eps_N(z)\to d(z)>0$. In contrast, $\pclass$ and $\Iclass$ regions correspond to algebraically ($\eps_N(z)=N^{-p(z)+o(1)}$) and exponentially ($\eps_N(z)=\exp[-N^{\alpha(z)+o(1)}]$) vanishing residuals, respectively. The $\pclass$ class typically emerges as the boundary separating the $\Iclass$ domains of exponential accuracy from the $\Cclass$ exterior. While this three-class picture captures the primary structure, further technical exceptions and exact-zero conventions are addressed in \suppref{app:foundations}.

\subsection{Physics of Atlas Geometry}
\label{sec:atlas-meaning}
The coordinate $z$ and its asymptotic label characterize independent physical properties within the atlas. Writing $z=|z|\ee^{\ii\phi}$, the radius sets the decay rate $-\ln|z|/T$ of the mode and the angle its Floquet frequency $\phi/T$. The asymptotic label instead sets how long this modal approximation can be trusted. These two aspects naturally define two distinct timescales:
\begin{equation}
 \tau_{\mathrm{dec}}(z)=-\frac{T}{\ln|z|},\qquad
 \tau_{\mathrm{ps},N}(z)=\frac{T}{\eps_N(z)}.
 \label{eq:main-timescales}
\end{equation}
Here $\eps_N(z)$ is the minimum residual of Eq.~\eqref{eq:intro-residual}, so $\tau_{\mathrm{ps},N}(z)$ is the longest validity time that any operator can attain at $z$~\cite{TrefethenEmbree2005,Davies2003}; it is infinite at an exact eigenvalue, where the residual vanishes.
The asymptotic atlas essentially classifies the thermodynamic scaling of this validity time: it diverges exponentially as $T\exp[N^{\alpha(z)+o(1)}]$ in the $\Iclass$ region, grows algebraically as $TN^{p(z)+o(1)}$ in the $\pclass$ region, and remains strictly bounded in the $\Cclass$ region. By contrast, the decay time $\tau_{\mathrm{dec}}(z)$ is governed entirely by the spectral radius $|z|$, dictating a nondecaying envelope on the unit circle ($|z|=1$) and an instantaneous collapse of the reference trajectory after a single Floquet cycle at the origin ($z=0$).

By the contractivity of the channel~\cite{SzNagy2010}, the dynamical error after $n$ cycles is bounded by
\begin{equation}
 \norm{\Phi_N^nX_{N,z}-z^nX_{N,z}}_{2,N}
 \le\eps_N(z)\sum_{j=0}^{n-1}|z|^j.
 \label{eq:main-validity}
\end{equation}
In the interior ($|z|<1$) the geometric sum converges, so the absolute error stays below $\eps_N(z)/(1-|z|)$ at every cycle. Long validity can therefore naturally coexist with the rapid decay of a nonzero signal.
On the unit circle ($|z|=1$), by contrast, the error bound grows linearly as $n\eps_N(z)=nT/\tau_{\mathrm{ps},N}(z)$. The approximation therefore stays within a tolerance $\delta$ for evolution times up to $\delta\,\tau_{\mathrm{ps},N}(z)$.

The geometry of the atlas determines which decay-phase combinations admit asymptotically accurate modes. Under Eq.~\eqref{eq:main-common-limits} the region supporting such vanishing residuals is formally defined as:
\begin{equation}
 K=\{z\in\overline\D:d(z)=0\},\qquad
 \Cclass=\overline\D\setminus K.
 \label{eq:main-vanishing-set}
\end{equation}
The Lipschitz continuity of the residual guarantees that $K$ is a closed set with well-defined boundaries. At a fixed angle, its radial cross-section gives the decay rates available to asymptotically accurate modes oscillating at that frequency. Similarly, at a fixed radius, an angular gap in $K$ marks a band of Floquet frequencies that no asymptotically accurate mode with that decay rate can carry. Geometrically, distant multipliers are linked if they share a connected domain, and dynamically separated if an intervening gap breaks them apart~\cite{ReichelTrefethen1992,TrefethenContediniEmbree2001}. While this geometry is intrinsic to the channel, an observable real-frequency absorption spectrum depends additionally on the specific state preparation and measurement.

The disconnected components of $K$ represent distinct regions of asymptotically accurate responses. Within this geometry, a hole is a finite-residual region completely enclosed by $K$, and a contact is a point where two previously isolated components meet. Continuous tuning of a physical parameter can drive geometric transitions, such as the merging of adjacent components or the abrupt closing of a hole. Consequently, these geometric reorganizations fundamentally reshape the global domains of the asymptotic atlas.

When two components are separated, a fixed closed contour can be drawn through the finite-residual region between them, and the residual along it stays bounded away from zero uniformly in $N$. The Riesz projection along this contour then isolates the exact eigenvalues it encloses, with a norm bound and perturbative stability that do not depend on system size~\cite{Macieszczak2016,Kato1995}. This provides a uniform separation certificate. Once a geometric transition merges the two components, every such contour must cross $K$ and the certificate is lost. In the dissipative shift chain, the susceptibility of the enclosed family to perturbations coupling it to its neighbor indeed grows algebraically at the merger and exponentially beyond it~\cite{Stewart1973,BhatiaRosenthal1997} (Sec.~\ref{sec:riesz-toy}). The explicit construction and bounds are given in \suppref{sec:riesz} and \suppref{app:cross-family}.

\section{Dissipative shift chain}
\label{sec:toy}
To demonstrate the theoretical framework developed above, we construct a dissipative spin chain that allows for an exact analytical evaluation of the atlas. Driven by directed shift and boundary reset, this model features a fundamentally nilpotent structure: nonidentity operators are carried along finite orbits that strictly terminate at the boundary. This mathematical property ensures global solvability of the residual problem, yielding a closed-form atlas tuned by a single parameter. Physically, this exact solution explicitly captures the geometric features anticipated in Sec.~\ref{sec:theory}, an angular gap and an enclosed hole, each closing at a specific threshold and leaving a measurable signature in spin-wave observables.

\subsection{Exact Solution}
The dynamics of the model over a single driving period is governed by a discrete sequence of operations~\cite{Ciccarello2022}: each period comprises a probabilistic dissipative shift followed by a uniform coherent rotation (see Fig.~\ref{fig:toy-atlases}(a)). Governed by a single random variable, the chain shifts to the right with probability $g$, discarding the rightmost spin and inserting a maximally mixed spin at the left, or remains in place with probability $a=1-g$. A global $\pi/2$ pulse about the $x$-axis then completes the cycle:
\begin{equation}
 \Phi_N=\Ad_{R_N}(a\,\mathrm{Id}+g\cE_N)
 \label{eq:main-shift-channel}
\end{equation}
where $\cE_N(X)=\frac{I_2}{2}\otimes\Tr_N X$ and $R_N=\bigotimes_j\ee^{-\ii\pi\sigma_j^x/4}$. This construction yields a unital completely positive trace-preserving (CPTP) channel whose unique stationary state is $2^{-N}I$.

\begin{figure*}[t]
 \centering
 \includegraphics[width=1\textwidth]{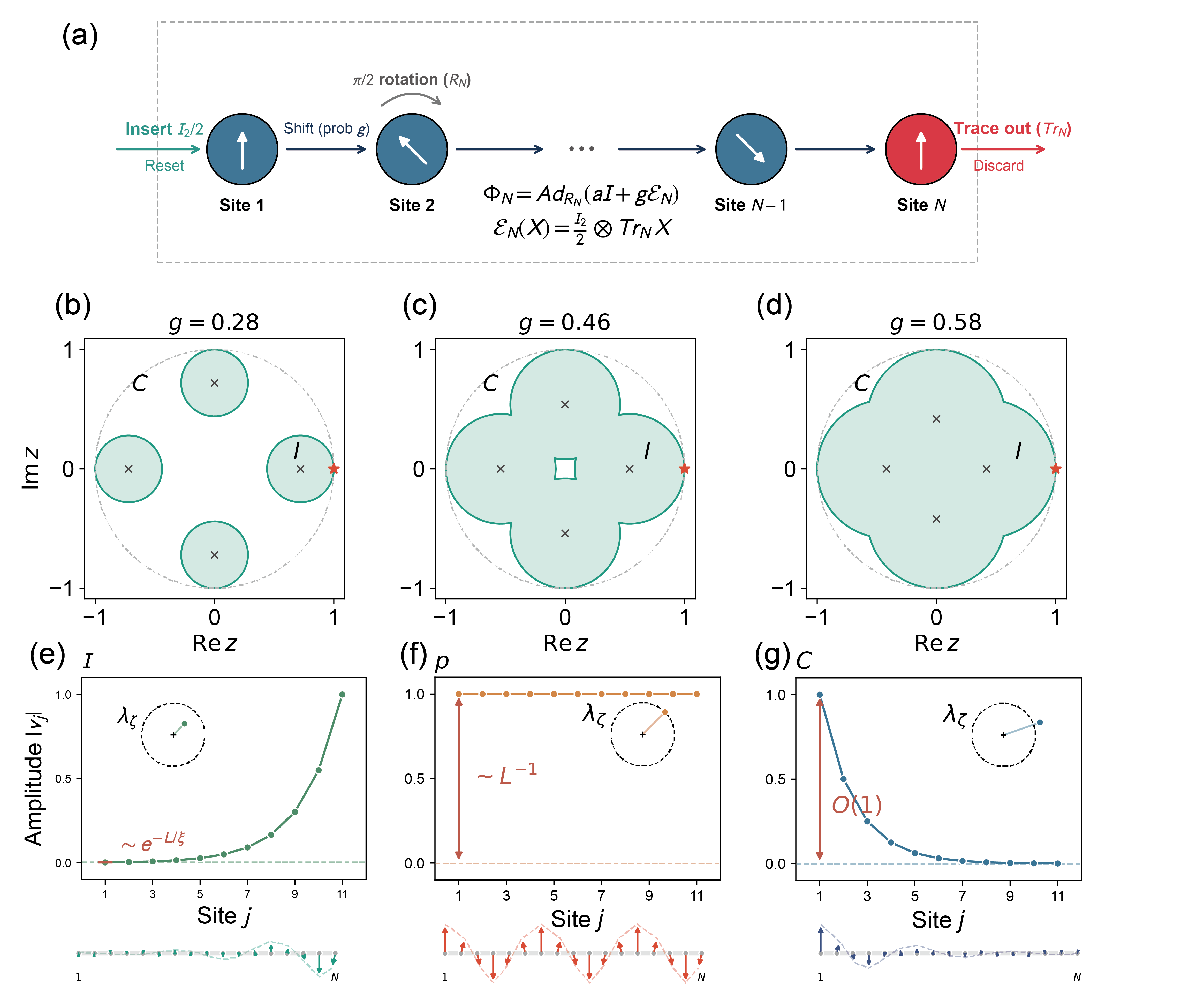}
 \caption{\textbf{Reorganization of the dissipative-shift-chain atlas with the activation probability $g$.}
\textbf{(a)} One period of the channel: insertion of a maximally mixed spin, a shift applied with probability $g$, a $\pi/2$ rotation, and the discarding of the last site.
\textbf{(b)--(d)} Complete asymptotic pseudospectral atlas at $g=0.28$, $0.46$, and $0.58$. The domains pass from disconnected, to a connected union enclosing a hole, to a filled center.
Filled regions carry exponential validity scales, exposed solid boundaries linear scales, and the exterior a positive limiting residual.
Crosses mark the exact nonstationary eigenvalues, the star the stationary multiplier, and the dashed circle is $|z|=1$.
\textbf{(e)--(g)} Spatial amplitude profiles $|v_j|$ of the corresponding spin-wave modes, one for each residual class: exponential ($\Iclass$), algebraic ($\pclass$), and bounded ($\Cclass$).}
 \label{fig:toy-atlases}
\end{figure*}

The solvability of this model relies on the dynamical property: translation preserves the spatial shape of any given spin configuration until it reaches the boundary, causing it to form an orthogonal orbit. In a basis of the coherent pulse, each orbit acquires one of four phases, $\zeta\in\{1,\ii,-1,-\ii\}$. Consequently, the full quantum channel restricted to an orbit of length $L$ reduces simply to $\zeta(aI_L+gS_L)$, where $S_L e_j = e_{j+1}$ and $S_L e_L = 0$. Because these independent orbits exhaustively cover all $4^N$ operator directions, the exact Floquet eigenvalues for $N\ge2$ are strictly limited to $1, a, \ii a, -a,$ and $-\ii a$. By decomposing the exponentially large many-spin problem into a collection of scalar shift problems, this exhaustive reduction guarantees that the global minimum residual over all operator directions can be evaluated exactly at every $z$, fulfilling the theoretical requirement established in Sec.~\ref{sec:atlas-definition}.

Solving the scalar block for each pulse phase reveals that each family contributes a pseudospectral disk centered at $a\zeta$. By defining the minimum distance $\rho(z) = \min_{\zeta^4=1} |z - a\zeta|$, we find that the complete channel obeys three distinct scaling laws, corresponding to three indices defined in Sec.~\ref{sec:atlas-definition}:
\begin{equation}
 \begin{array}{lll}
 \rho(z)<g:& \eps_N(z)=\Theta([\rho(z)/g]^N),& \Iclass\text{ region},\\
 \rho(z)=g:& \eps_N(z)\sim\pi g/(2N),&\pclass\text{ region},\\
 \rho(z)>g:& \eps_N(z)\to\rho(z)-g,&\Cclass\text{ region}.
 \end{array}
 \label{eq:main-shift-atlas}
\end{equation}

A profound physical meaning can be extracted from this geometric classification by applying Hermitian doubling to each residual block~\cite{Herviou2019}. These three scalings elegantly map onto an auxiliary Su--Schrieffer--Heeger (SSH) chain~\cite{SSH1979} with the intra-cell coupling $t_1 = |z - a\zeta|$ and the inter-cell coupling $t_2 = g$. The disk interior corresponds to the nontrivial SSH phase, where opposite-edge singular modes split only exponentially. At the boundary ($t_1 = t_2$), the bulk gap closes, and a critical extended mode produces the $N^{-1}$ algebraic residual. Finally, outside the disk, the auxiliary chain is trivial and gapped. This mapping establishes deep connections between the pseudospectral geometry and topological phases~\cite{Gong2018,YaoWang2018,Okuma2020,Ammari2025,Kawabata2019,Ashida2020,Bergholtz2021,OkumaSato2023}, and the full correspondence is provided in \suppref{sec:SSH}.

In the context of Eq.~\eqref{eq:main-classes}, these scaling laws translate to exponential validity within the four disks, linear validity along the exposed arcs, and bounded validity in the exterior. Geometrically, the atlas is a union of four disks of radius $g$ centered at $(1-g)\zeta$. Since $g$ uniquely determines both the disk sizes and their separation, growing $g$ forces the atlas to reorganize at two critical values:
\begin{equation}
 g_1=\sqrt2-1,\qquad g_2=\tfrac12.
 \label{eq:main-shift-thresholds}
\end{equation}

At the critical threshold $g_1$, adjacent disks first touch, fusing four previously isolated components into a single connected domain. At each contact point $z_\zeta=(1-g_1)\zeta(1+\ii)/2$, the adjacent $\zeta$ and $\ii\zeta$ SSH families simultaneously close their bulk gaps ($t_1=t_2=g_1$). By bridging the angular gap between consecutive pulse phases, this criticality completely eliminates the phase dead zone for long-lived modes. Concurrently, the merger severs the bounded-residual corridor that previously separated distinct response families, stripping them of their isolation protection. Beyond $g_1$, the connected union leaves a finite-residual hole around the origin.

That hole closes at $g_2$. At this point, all four auxiliary chains simultaneously reach criticality ($t_1=a=t_2=1/2$). Unlike the pairwise merger at $g_1$, this second transition concurrently drives all four SSH families from a trivial gapped phase into a topological phase, marked by a jump in their winding numbers from zero to one~\cite{Gong2018,ZhangYangFang2020,Borgnia2020}. Despite their distinct topological characters, both geometric transitions yield the exact same algebraic residual scaling, $\eps_N(z)\sim\pi g/(2N)$. Consequently, local accuracy labels alone cannot differentiate a pairwise merger from a hole closure; this distinction relies entirely on the global geometry. Figures~\ref{fig:toy-atlases}(b)--\ref{fig:toy-atlases}(d) illustrate these three evolving geometries.

\subsection{Spin-wave Observables}
\label{sec:shift-observables}\label{sec:riesz-toy}
To observe the geometric atlas, the initial preparation must actively excite the relevant Floquet multipliers~\cite{Kunjummen2023}. Fortunately, the translation symmetry of the shift chain naturally links spatial wave numbers to these multipliers, allowing spin waves to directly probe the atlas geometry. By tuning the spatial wave number, the bulk multipliers of these spin-wave excitations explicitly sweep the four bounding circles of the atlas. To demonstrate this, we initialize the system in a spatially modulated product state $\rho_0=\bigotimes_j\frac{I+m\cos(qj)\sigma_j^x}{2}$ with $|m|<1$. The resulting local magnetization $f_j(n)=\Tr[\sigma_j^x\Phi_N^n\rho_0]$ strictly obeys the finite-chain recursion
\begin{equation}
 f_j(n+1)=af_j(n)+gf_{j-1}(n),\qquad f_0(n)=0.
 \label{eq:main-shift-recursion}
\end{equation}

By linearly combining the preparations, one can reconstruct a complex spatial wave characterized by the bulk multiplier $a+g\ee^{\ii q}$. As the wave number $q$ varies, this multiplier traces the boundary of the $\zeta=1$ pseudospectral disk. Accessing the remaining three pulse phases, however, requires transverse spin textures and higher-order observables. Specifically, helical states with $\langle\sigma_j^y\rangle=m\cos(qj)$ and $\langle\sigma_j^z\rangle=-m\sin(qj)$ accessed via the observables $\sigma_j^y\pm\ii\sigma_j^z$ capture two of the phases, while a two-spin quadrupole measurement, $(\sigma_j^y\sigma_{j+1}^y-\sigma_j^z\sigma_{j+1}^z)/\sqrt2$, isolates the final one. Together, these experimental configurations yield the complete set of response curves:
\begin{equation}
 \lambda_\zeta(q)=\zeta(a+g\ee^{\ii q}).
 \label{eq:main-response-circles}
\end{equation}
In this unified framework, the initial texture and chosen observable fix the discrete pulse phase $\zeta$, while the continuous wave number $q$ selects a specific point along the corresponding circle. Detailed derivations are provided in \suppref{app:shift-readouts}.

\begin{figure*}[t]
 \centering
 \includegraphics[width=1\textwidth]{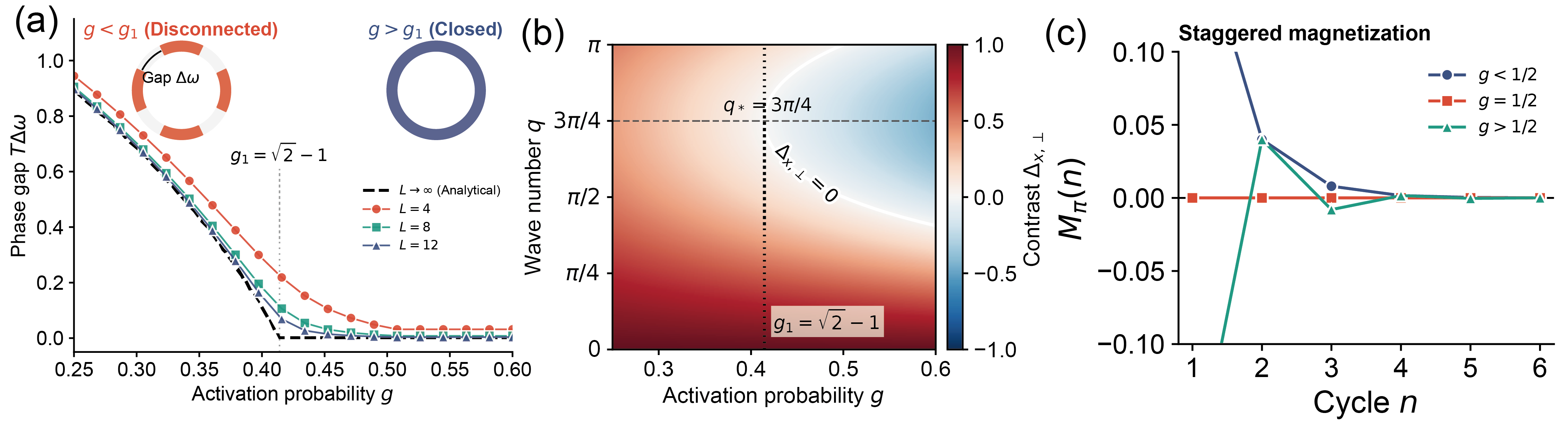}
 \caption{\textbf{Spin-wave observables diagnosing the two geometric transitions of the shift-chain atlas.} All curves are theoretical predictions.
\textbf{(a)} Width $T\Delta\omega$ of the excluded band of Floquet phases, Eq.~\eqref{eq:main-phase-gap}, against the activation probability $g$. It vanishes at the first threshold $g_1=\sqrt2-1$.
\textbf{(b)} One-cycle quadrature contrast against $g$. The contrast changes sign at $g_1$, so it reads out the merger of adjacent components directly.
\textbf{(c)} Staggered magnetization $M_\pi(n)$ over successive cycles. It cancels identically at the second threshold $g_2=1/2$ and alternates in sign while decaying above it. The initial value $M_\pi(0)=1$ is omitted so that the later cycles resolve.}
 \label{fig:physical-shift-readouts}
\end{figure*}

These spin-wave readouts directly uncover the physical origin of the three residual classes, revealing them as three distinct spatial localization behaviors of a single boundary mismatch. Fix a target multiplier $z$ away from the exact spectrum and set
\begin{equation}
\beta_\zeta(z)=\frac{z/\zeta-a}{g},\qquad
\lambda_\zeta(q)-z=\zeta g(\ee^{\ii q}-\beta_\zeta).
\label{eq:main-shift-beta}
\end{equation}

The asymptotic accuracy of this spin-wave mode is therefore entirely dictated by $|\beta_\zeta|$, which governs where the texture localizes relative to the boundary. For $|\beta_\zeta|<1$, the texture localizes at the right edge, exponentially suppressing the normalized boundary mismatch at site $1$ and placing the mode in class $\Iclass$ (Fig.~\ref{fig:toy-atlases}(e)). At $|\beta_\zeta|=1$, the magnitude becomes uniform across the chain, supporting an extended plane wave whose optimal $L^{-1}$ residual is achieved by smoothing the envelope (class $\pclass$; Fig.~\ref{fig:toy-atlases}(f)). For $|\beta_\zeta|>1$, the texture instead localizes at the left edge, so the normalized boundary mismatch remains a finite fraction and the residual stays finite (class $\Cclass$; Fig.~\ref{fig:toy-atlases}(g))~\cite{Wang2024,Kaneshiro2026,Yokomizo2019,YaoWang2018,Weidemann2020}.

The first transition removes an angular gap, corresponding to a forbidden band of Floquet phases that no accurate mode can carry~\cite{Koch2026}. For $g<1/2$, the frequency width of this excluded band is
\begin{equation}
T\Delta\omega(g)=[\pi/2-2\arcsin(g/(1-g))]_+.
\label{eq:main-phase-gap}
\end{equation}
The geometric event that this width vanishes precisely at $g_1$ can be detected within a single evolution cycle by interfering two wave numbers that straddle the gap. Specifically, combining an $x$ texture at $q=3\pi/4$ with a $yz$ helix at $-q$ yields a complex multiplier difference of $[1-(1+\sqrt 2)g](1-\ii)$. The normalized quadrature contrast of this readout changes sign exactly at the critical threshold $g_1$, providing an experimental signature of the domain merger.

Detecting the closure of the central hole requires a different preparation. A staggered initial state, $f_j(0)=m(-1)^j$, directly addresses this central region, yielding a normalized bulk magnetization
\begin{equation}
M_\pi(n)=(1-2g)^n.
\label{eq:main-staggered}
\end{equation}
At the second critical threshold $g_2=1/2$, the shifted and unshifted spin textures cancel exactly. Beyond this threshold ($g>g_2$), the magnetization signal begins to alternate in sign as it decays. This onset of sign alternation physically signifies the exact moment the response circle expands to completely enclose the origin. Consequently, as summarized in Fig.~\ref{fig:physical-shift-readouts}, both geometric transitions in the atlas carry independent experimental diagnostics.

The gaps between the pseudospectral disks correspond to the bounded-residual regions, which safely admit the stabilizing contour construction introduced in Sec.~\ref{sec:atlas-meaning}. This model quantitatively demonstrates how the stability of distinct response families is entirely dictated by the atlas geometry. As the disks transition from strictly separated to finally merged, their susceptibility to inter-family perturbations evolves from bounded, to algebraically growing, and ultimately to exponentially diverging with system size. Detailed perturbation bounds are provided in \suppref{app:cross-family}.

In summary, the dissipative shift chain provides a solvable paradigm that connects abstract pseudospectral geometries to physical observables. Tuning a single activation parameter drives two geometric transitions, the merging of isolated domains and the closing of a central hole, each carrying a measurable spin-wave observable. This model quantitatively demonstrates that the geometry of the atlas governs the dynamical stability of distinct response families.

\section{Chiral XY model}
\label{sec:xy}
While the dissipative shift chain in Sec.~\ref{sec:toy} provides a theoretical paradigm for evaluating the atlas, it remains a mathematical toy model. To ground these geometric concepts in a realistic physical setting, we introduce a chiral XY model governed by standard exchange interactions and engineered reservoirs. Through the Jordan--Wigner transformation~\cite{JordanWigner1928,LiebSchultzMattis1961,Katsura1962}, the XY Hamiltonian maps onto quadratic fermion hopping. This reservoir-engineering scheme recovers the nilpotent shift mechanism explored previously, translating the solvability of the abstract model into an implementable many-body quantum system~\cite{Diehl2008,Kraus2008,Barreiro2011,Diehl2011}.

\begin{figure*}[t]
 \centering
 \includegraphics[width=1\textwidth]{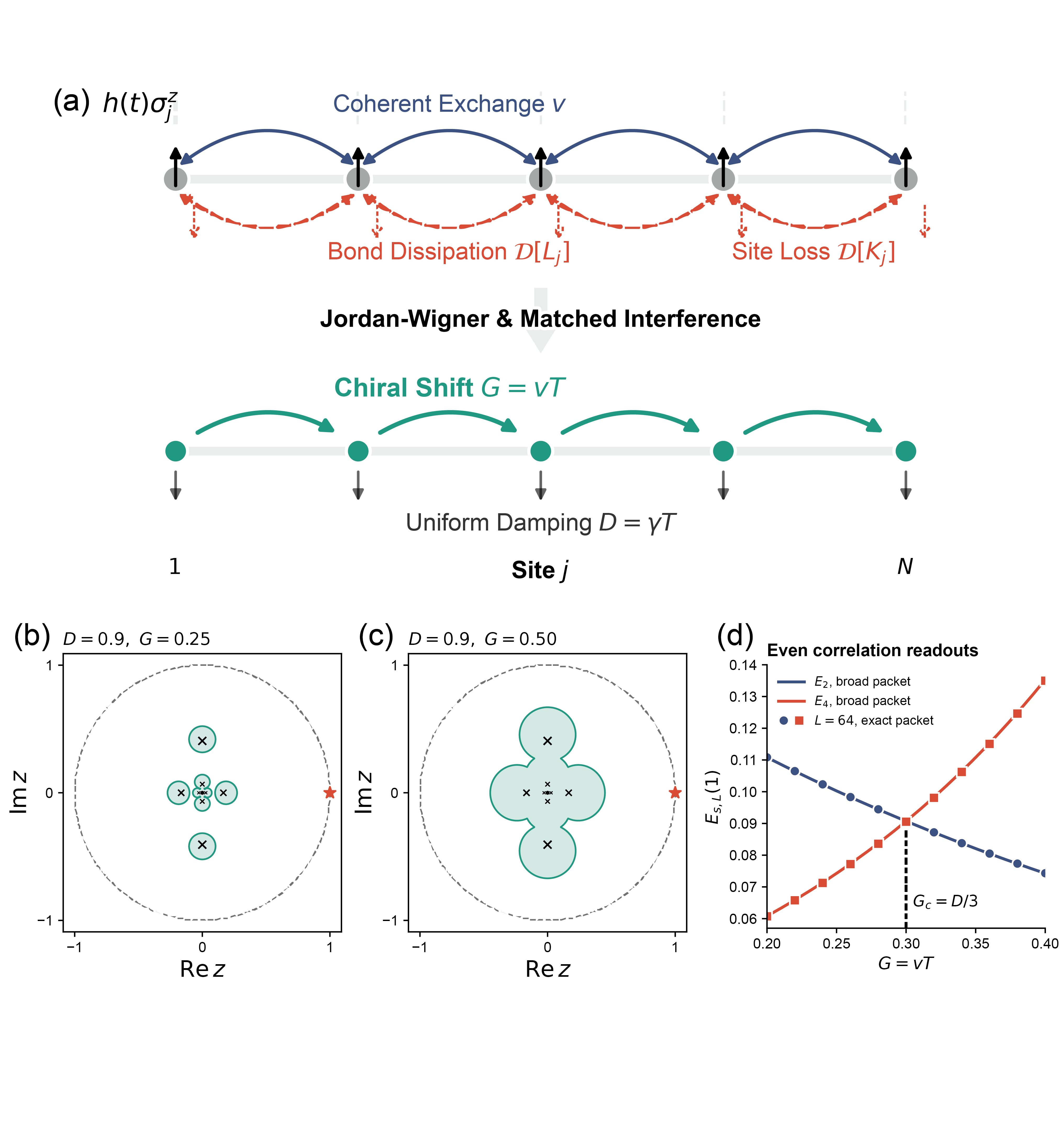}
 \caption{\textbf{Correlation control at fixed stationary state and fixed exact Floquet spectrum in the chiral XY model.}
\textbf{(a)} The effective dynamics. A spin chain with coherent exchange, bond dissipation, and site loss maps under the Jordan--Wigner transformation and matched interference onto a strict directed shift with uniform damping.
\textbf{(b)--(c)} Complete spin-operator pseudospectral atlases at $G=0.25$ and $G=0.50$, both at $D=0.9$. The exact eigenvalues (crosses), the stationary multiplier (star), and the unit circle are identical in the two panels; only the domains deform.
Filled domains carry exponential validity scales, exposed solid interfaces linear scales, and the exterior a positive limiting residual.
\textbf{(d)} The two even correlation readouts of Eq.~\eqref{eq:main-XY-signals} at $n=1$ against $G=vT$. They exchange their magnitude ordering at the exposed contact $G_c=D/3$. Solid lines give the broad-packet limits, symbols the exact $L=64$ packet formulas. Exchange and dissipation rates are varied together.}
 \label{fig:XY-atlases}
\end{figure*}

\subsection{Chiral Dynamics}
We begin with the isotropic XY Hamiltonian subjected to a time-periodic longitudinal field:
\begin{equation}
H_N(t)=\frac v4\sum_{j=1}^{N-1}
(\sigma_j^x\sigma_{j+1}^x+\sigma_j^y\sigma_{j+1}^y)
-\frac{h(t)}2\sum_{j=1}^N\sigma_j^z,
\label{eq:main-XY-H}
\end{equation}
where $h(t)=h_0+h_1\cos(2\pi t/T)$. Through the Jordan--Wigner transformation~\cite{JordanWigner1928,LiebSchultzMattis1961}, defined by $\sigma_j^z=I-2c_j^\dagger c_j$, the exchange interaction naturally maps to nearest-neighbor fermion hopping, while the magnetic field acts as an on-site potential. To preserve the solvability of this quadratic structure in the open-system dynamics, we engineer the reservoir using jump operators that are linear in these same fermions~\cite{Prosen2008,Prosen2011XXZ,Altland2021}. Specifically, the Lindbladian generator takes the form of Eq.~\eqref{eq:main-open-system}, with the following jump operators taken together with their Hermitian adjoints:
\begin{equation}
 \begin{gathered}
 L_j=\sqrt{v/2}(c_j-\ii c_{j+1}),\qquad
 K_j=\sqrt{\kappa_j}\,c_j,\\
 \kappa_j=\gamma-\tfrac v2\deg(j),\qquad \gamma>v>0.
 \end{gathered}
 \label{eq:main-XY-jumps}
\end{equation}
Here, $\deg(j)$ equals two in the bulk, one at either endpoint, and zero for a single isolated site, which inherently ensures exact gain-loss balance across the finite chain.

The relative phase between the coherent exchange and bond dissipation induces an interference that precisely cancels the reverse propagation, leaving only a chiral shift (Fig.~\ref{fig:XY-atlases}(a)). This continuous interference aligns with established reservoir-engineering mechanisms for nonreciprocal transport~\cite{MetelmannClerk2015, Yang2022, McDonald2018, Wanjura2020, LauClerk2018}, while retaining all trace-preserving recycling terms required by Eq.~\eqref{eq:main-open-system}. Explicitly, the damping matrix $M_N=\gamma I+\ii v(S_N-S_N^\dagger)/2$ and the one-particle Hamiltonian matrix $h_N(t)=v(S_N+S_N^\dagger)/2+h(t)I$ combine to yield
\begin{equation}
-M_N-\ii h_N(t)=-\gamma I-\ii vS_N-\ii h(t)I.
\label{eq:main-XY-cancellation}
\end{equation}
With the reverse-shift terms ($S_N^\dagger$) eliminated, integrating over one period with the dimensionless parameters $D=\gamma T$, $G=vT$, and $\theta=h_0T$ yields the elementary particle-hole propagation matrices:
\begin{equation}
B_{N,\pm}=\ee^{-D\pm\ii\theta}\ee^{GS_N},\qquad D>G>0.
\label{eq:main-XY-directed}
\end{equation}
The shift $S_N$ is the same nilpotent matrix utilized in Sec.~\ref{sec:toy}, with the continuous chiral hopping now generating the exponential propagator $\ee^{GS_N}$ rather than the discrete map $aI+gS_N$. While the Jordan--Wigner transformation and the matched bath render this fermion propagation explicitly solvable, evaluating the true pseudospectral residual requires optimizing over the entire $4^N$-dimensional spin-operator space. To bridge this gap, we employ a Hilbert--Schmidt unitary reduction, rigorously elevating the single-particle solution to the full operator algebra. Specifically, antisymmetrized products of $B_{N,+}$ and $B_{N,-}$ span all operator directions. This ensures that even single-fermion observables acquire their precise evolution strictly through this reduction framework~\cite{ProsenIlievski2011,Prosen2008,BarthelZhang2022,Altland2021}. The complete reduction is detailed in \suppref{app:XY-operator-space}.

Exploiting this nilpotency yields the exact spectrum:
\begin{equation}
 \begin{gathered}
 \lambda_{r,q}=\ee^{-rD+\ii q\theta},\\
 r=r_++r_-,\quad q=r_+-r_-,\quad 0\le r_\pm\le N.
 \end{gathered}
 \label{eq:main-XY-isospectral}
\end{equation}
where each $(r_+,r_-)$ sector contributes a multiplicity of $\binom N{r_+}\binom N{r_-}$, and coincident values add their multiplicities. In Eq.~\eqref{eq:main-XY-isospectral}, only $D$ and $\theta$ dictate this spectrum; the coherent exchange rate $G$ is entirely absent. Therefore the entire $4^N$-dimensional eigenvalue multiset and the unique stationary state $\rho_{\mathrm{ss}}=2^{-N}I$ remain frozen as $G$ varies. The exchange rate $G$ does, however, dictate the geometry of the pseudospectral atlas, setting the domain radii in the logarithmic coordinates of Sec.~\ref{sec:xy-atlas-contacts}. Tuning $G$ therefore deforms the atlas along an isospectral family (Figs.~\ref{fig:XY-atlases}(b) and
  \ref{fig:XY-atlases}(c)), guaranteeing that any observable consequence of this deformation is strictly attributable to the underlying geometry alone.

\subsection{Measurable Correlation Signals}
\label{sec:xy-observables}
To demonstrate the geometry-driven dynamics along this isospectral family, we evaluate two specific physical observables chosen for their strategic positions in the atlas: a staggered mode-occupation difference and a joint-occupation correlation of two smooth modes. These configurations are purposefully selected because they react oppositely to the chiral propagation. Absorbing the site phase into $\widetilde c_j=\ii^{j-1}c_j$, we construct two disjoint smooth spatial envelopes, $f_0$ and $f_1$, of length $L$, each proportional to $\sin[\pi j/(L+1)]$, alongside a highly oscillatory staggered envelope $f_{\pi,j}=(-1)^jf_{0,j}$. For a sufficiently long chain ($N\ge2L$), the corresponding encoded modes and target observables are formally defined as:
\begin{equation}
 \begin{gathered}
 d_a=\sum_j f_{a,j}\widetilde c_j,\quad Q_a=I-2d_a^\dagger d_a,
 \quad a=0,1,\pi,\\
 A_2=Q_\pi,\quad A_4=Q_0Q_1,\quad
 \rho_{s,\pm}=2^{-N}(I\pm A_s).
 \end{gathered}
 \label{eq:main-XY-preparation}
\end{equation}
with $s=2,4$. Here, the parity-even operators $A_2$ and $A_4$ rigorously formulate the mode-occupation difference and the joint-occupation correlation, respectively, where the subscript $s$ indicates the number of constituent Majorana factors. Because $A_s^2=I$ and $\Tr A_s=0$, the resulting initial states $\rho_{s,\pm}$ are strictly positive and normalized, establishing perfectly well-defined dynamical readouts for our theoretical comparison.

The normalized dynamical contrast of these preparations,
\begin{equation}
E_{s,L}(n)=\Tr\left[A_s\Phi_N^n\frac{(\rho_{s,+}-\rho_{s,-})}{2}\right]
=2^{-N}\Tr[A_s\Phi_N^n(A_s)],
\label{eq:main-XY-readout}
\end{equation}
starts identically at one. Evaluating this contrast in the broad-packet limit reveals the exact analytical decay rates for both signals at any fixed cycle $n$:
\begin{equation}
\begin{gathered}
E_2(n)=\ee^{-2n(D+G)},\quad E_4(n)=\ee^{-4n(D-G)},\\
\frac{E_4(n)}{E_2(n)}=\ee^{n(6G-2D)}.
\end{gathered}
\label{eq:main-XY-signals}
\end{equation}
This exposes a striking dynamical inversion~\cite{SongYaoWang2019,Haga2021}. For $n\ge1$, the occupation-difference signal $E_2$ is larger when $G<D/3$, whereas $E_4$ dominates the joint-occupation signal when $G>D/3$, with the two perfectly intersecting at the critical threshold $G_c=D/3$. The retention order of these distinct physical correlations is therefore tunable by the coherent exchange $G$, a remarkable dynamical crossing that unfolds while the exact spectrum remains invariant. Exact finite-packet expressions are detailed in \suppref{app:readout-support}.

\subsection{Geometric Explanation}
\label{sec:xy-atlas-contacts}
The preceding analysis revealed a striking dynamical inversion. The physical origin of this unusual phenomenon lies entirely in the geometric deformation of the pseudospectral atlas~\cite{Kiorpelidis2025,ReichelTrefethen1992}. To reveal this mechanism, we classify the response families by their Majorana degree $r\ge1$ and phase charge $q=-r,-r+2,\ldots,r$. Using the logarithmic coordinates $x=-\ln|z|$ and $\phi=\arg z$, the domain for each family is rigorously established as:
\begin{equation}
 \begin{gathered}
 \delta_{r,q}(z)=\sqrt{(x-rD)^2+
       \dist(\phi-q\theta,2\pi\mathbb Z)^2},\\
 \Omega_{r,q}=\{z\ne0:\delta_{r,q}(z)\le rG\},\qquad
 \mathcal K=\{0,1\}\cup\bigcup_{r,q}\Omega_{r,q}.
 \end{gathered}
 \label{eq:main-XY-domains}
\end{equation}
Equation~\eqref{eq:main-XY-domains} dictates that while the domain centers $rD$ are anchored by dissipation, their radii $rG$ expand proportionally with the coherent exchange. Consequently, as $G$ increases, domains belonging to distinct families inflate and advance toward one another. This purely geometric expansion elegantly explains how tuning $G$ continuously reorganizes the atlas, ultimately driving the signal crossing without altering the exact eigenvalues.

For the neutral families ($q=0$), the staggered textures select the outer and inner real tips $\ee^{-r(D-G)}$ and $\ee^{-r(D+G)}$, respectively. Eq.~\eqref{eq:main-XY-signals} probes the inner tip of the $r=2$ domain and the outer tip of the $r=4$ domain. These two boundaries first meet when $\ee^{-2(D+G)}=\ee^{-4(D-G)}$, yielding the critical point:
\begin{equation}
G_c=D/3,\qquad z_c=\ee^{-8D/3}.
\label{eq:main-XY-contact}
\end{equation}
The dynamical threshold at which the two physical signals cross is therefore the threshold at which their geometric domains touch. Our chosen parameters $D=0.9$ and $\theta=\pi/2$ are illustrated in Fig.~\ref{fig:XY-atlases}. This signal crossing marks a genuine geometric transition in the complete atlas.

The contact point $z_c$ naturally carries the algebraic label characteristic of a domain interface. The full-space estimate yields a residual of $\eps_N(z_c)=\Theta(N^{-1})$, which translates to a validity time of $\tau_{\mathrm{ps},N}(z_c)=\Theta(NT)$. Conversely, the modal decay time evaluates to a system-size-independent constant, $\tau_{\mathrm{dec}}(z_c)=3T/(8D)$. This explicitly realizes the theoretical separation of timescales anticipated in Sec.~\ref{sec:atlas-meaning}: the validity time of the dynamical description diverges with the system size, whereas the intrinsic signal decay rate remains fixed.

In summary, the chiral XY model grounds our mathematical framework in a concrete physical setting, demonstrating the profound practical utility of the geometric analysis. The striking dynamical inversion of physical signals, which might otherwise appear as an anomalous experimental phenomenon, finds a natural explanation through the collision of geometric domains within the atlas. This correspondence vividly illustrates how real-world open-system dynamics can be directly decoded and predicted through the lens of underlying geometric structures.

\section{Conclusions and outlook}\label{sec:conclusions}

By classifying the closed unit disk according to the size scaling of pseudospectral residuals, we have organized the accuracy scales of nonnormal Floquet channels into a unified geometric framework~\cite{TrefethenEmbree2005,Kiorpelidis2025}. The resulting asymptotic pseudospectral atlas partitions the complex plane into domains of exponential, algebraic, and bounded scaling. This geometry systematically maps the validity regimes of approximate modes, and the boundaries between these domains determine the stability of different response families.

Our analysis of two exactly solvable models translates this geometric classification into measurable physical phenomena. In the dissipative shift chain, tuning a single parameter alters the radii and separation of the atlas domains. We identified two distinct geometric transitions. Both of these geometric events are directly detectable through spin-wave measurements. In the chiral XY model, we observed two preparable correlations exchanging their retention order by tuning the system along a parameter family. This signal inversion occurs precisely at an exposed contact between two atlas domains, which confirms that the geometry of the residual field can drive observable dynamics.

Constructing the atlas for a generically interacting open quantum system remains an open challenge~\cite{Lee2013,MoriKuwaharaSaito2016}. The theoretical classification developed here applies perfectly to such systems, but the computational cost is the bottleneck. Evaluating the asymptotic residual requires resolving the smallest singular value of a superoperator whose dimension grows exponentially with system size. For instance, the complexity of the optimal operator direction can grow alongside the system size, meaning that truncations to low-order correlations must be carefully justified~\cite{NachtergaeleSims2006,NachtergaeleSimsYoung2019}.

Despite these computational challenges, established numerical tools bring the atlas of interacting systems within reach. Matrix product operators and variational schemes can effectively represent one-dimensional driven dissipative dynamics~\cite{ZwolakVidal2004,Verstraete2004,Werner2016,Cui2015}, while iterative solvers can extract large-scale pseudospectra~\cite{WrightTrefethen2001,Hochstenbach2001}. A driven dissipative Ising chain is a natural first target for these tensor-network approaches. For such future numerical endeavors, the two exact models presented in this work provide essential benchmarks to validate any computational implementations.

\acknowledgments
This work was supported by the National Natural Science Foundation of China (Grants No. 12188101, No. 12274081), National Key Research and Development Program of China (Grant No. 2024YFA1409800).

Y.X. and H.S. contributed equally to this work.

\putbib[references]
\end{bibunit}

\clearpage
\begin{bibunit}[apsrev4-2]
\setcounter{section}{0}
\setcounter{equation}{0}
\setcounter{figure}{0}
\setcounter{table}{0}
\setcounter{proposition}{0}
\setcounter{lemma}{0}
\setcounter{assumption}{0}
\renewcommand{\thesection}{S\arabic{section}}
\renewcommand{\thesubsection}{\thesection.\arabic{subsection}}
\makeatletter\def\p@subsection{}\def\p@subsubsection{}\makeatother
\renewcommand{\theequation}{S\arabic{equation}}
\renewcommand{\thefigure}{S\arabic{figure}}
\renewcommand{\thetable}{S\arabic{table}}
\renewcommand{\theproposition}{S\arabic{proposition}}
\renewcommand{\thelemma}{S\arabic{lemma}}
\renewcommand{\theassumption}{S\arabic{assumption}}
\makeatletter
\let\author\SMsaved@author
\let\affiliation\SMsaved@affiliation
\let\thanks\SMsaved@thanks
\let\maketitle\SMsaved@maketitle
\let\@AF@join\SMsaved@AFjoin
\let\@AAC@list\@empty
\let\@AFF@list\@empty
\let\@AFG@list\@empty
\let\@address\@empty
\c@affil\z@
\c@footnote\z@
\c@mpfootnote\z@
\let\SMsaved@label\label
\let\label\@gobble
\makeatother
\title{Supplemental Material for ``Asymptotic Pseudospectra in Dissipative Floquet Quantum Systems: Geometric Structures and Observable Dynamics''}
\author{Yuncheng Xie}
\thanks{These authors are co-first authors and contributed equally to this work.}
\affiliation{Department of Physics, Fudan University, Shanghai 200433, China}
\affiliation{Key Laboratory of Computational Physical Sciences (Ministry of Education), State Key Laboratory of Surface Physics,
Fudan University, Shanghai 200433, China}

\author{Haozhe Shi}
\thanks{These authors are co-first authors and contributed equally to this work.}
\affiliation{Department of Physics, Fudan University, Shanghai 200433, China}
\affiliation{Key Laboratory of Computational Physical Sciences (Ministry of Education), State Key Laboratory of Surface Physics,
Fudan University, Shanghai 200433, China}

\author{Zhuocheng Ma}
\affiliation{State Key Laboratory of Artificial Microstructure and Mesoscopic Physics, School of Physics, Peking University, Beijing 100871, China}

\author{Weibin Chu}
\affiliation{Department of Physics, Fudan University, Shanghai 200433, China}
\affiliation{Key Laboratory of Computational Physical Sciences (Ministry of Education), State Key Laboratory of Surface Physics,
Fudan University, Shanghai 200433, China}

\author{Xin-Gao Gong}
\affiliation{Department of Physics, Fudan University, Shanghai 200433, China}
\affiliation{Key Laboratory of Computational Physical Sciences (Ministry of Education), State Key Laboratory of Surface Physics,
Fudan University, Shanghai 200433, China}
\date{\today}

\maketitle
\makeatletter
\let\label\SMsaved@label
\makeatother
\setcounter{page}{1}
\renewcommand{\thepage}{S\arabic{page}}
\raggedbottom
This Supplemental Material contains the proofs supporting the paper.
Sections S1--S3 treat the general framework and the two exactly solvable
models in the order of the main text.

\setcounter{tocdepth}{1}
\SMtableofcontents

\section{Residual framework}
\label{app:foundations}
Throughout we use the channel, norm, residual, and indicator definitions of
Sec.~\mainref{sec:theory}, and establish their limiting properties for the
fixed-period unital sequences considered here.

\subsection{Asymptotic Classes}
\label{app:residual-geometry}
Multiplying both input and output Hilbert--Schmidt (HS) inner products by the same
size-dependent scalar leaves superoperator singular values unchanged, whereas
passing to a different norm can change the asymptotic rates~\cite{TrefethenEmbree2005}. For a unital
completely positive trace-preserving (CPTP) map~\cite{Lindblad1976,GKS1976}, the Schwarz inequality gives the
HS contraction used in the main text~\cite{SzNagy2010},
\begin{equation}
 \norm{\Phi_NX}_{2,N}^2
 \le 2^{-N}\Tr\Phi_N(X^\dagger X)=\norm X_{2,N}^2.
 \label{eq:contraction}
\end{equation}

For $z\notin\spec\Phi_N$, singular-value decomposition gives the inverse-norm
identity~\cite{TrefethenEmbree2005,Trefethen1997,vonNeumann1951}
\begin{equation}
 \eps_N(z)=\smin(\Phi_N-z\id)
 =\norm{(\Phi_N-z\id)^{-1}}^{-1}.
 \label{eq:resolvent}
\end{equation}
The same quantity is the smallest unstructured perturbation that creates an
eigenvalue at $z$~\cite{TrefethenEmbree2005}:
\begin{equation}
 \eps_N(z)=\inf\{\norm E:z\in\spec(\Phi_N+E)\}.
 \label{eq:perturbation}
\end{equation}

If $(\Phi_N-z)X=\eps_NY$ with $\norm X=\norm Y=1$, the rank-one superoperator $E=-\eps_N\lvert Y\rangle\langle X\rvert$ makes $X$ an exact eigenoperator, while conversely $z\in\spec(\Phi_N+E)$ implies $\eps_N(z)\le\norm E$. These perturbations are \emph{unstructured}~\cite{TrefethenEmbree2005}: they need not preserve complete positivity, trace preservation, locality, or a fixed bath realization. We work on the closed unit disk, retaining the peripheral information needed for dynamics.

\begin{proposition}[Lipschitz residual geometry]\label{prop:lipschitz}
For every $N$,
\begin{equation}
 |\eps_N(z)-\eps_N(w)|\le|z-w|.
 \label{eq:lipschitz}
\end{equation}
For a uniformly Hilbert--Schmidt-bounded sequence of channels, $d_-(z)$ and $d_+(z)$ are also 1-Lipschitz. Consequently,
\begin{equation}
 K_{\rm all}=\{z:d_+(z)=0\},\qquad
 K_{\rm sub}=\{z:d_-(z)=0\}
 \label{eq:Ksets}
\end{equation}
are closed and $K_{\rm all}\subseteq K_{\rm sub}$.
\end{proposition}
\begin{proof}
For a unit operator vector $X$,
$\norm{(\Phi_N-z)X}\le\norm{(\Phi_N-w)X}+|z-w|$.
Take the infimum, exchange $z,w$, and pass separately to the lower and upper limits, where uniform boundedness keeps these limiting functions finite on compact sets.
\end{proof}
Applying the antiunitary map $X\mapsto X^\dagger$ and using Hermiticity preservation gives $\eps_N(\bar z)=\eps_N(z)$, so the exact atlases are symmetric under reflection about the real axis.

Under the common-limit hypothesis in
Eq.~\maineqref{eq:main-common-limits}, the two sets in Eq.~\eqref{eq:Ksets}
coincide with $K$ in Eq.~\maineqref{eq:main-vanishing-set}, and their relative
complement in the closed disk is the finite-residual class $\{d>0\}$. Both sets
are defined directly from finite systems, and calling either one the spectrum of
a specific infinite-system channel requires in addition a common operator space,
an embedding, and a convergence theorem.

The extended-real limits in Eq.~\maineqref{eq:main-six} are logarithmic growth
orders, weaker than the dilation-uniform indices of regular-variation
theory~\cite{Bingham1987}:
\begin{align}
 d_-(z)>0 &: \text{ eventual separation from zero},\nonumber\\
 d_-(z)=0<d_+(z) &: \text{ vanishing on some subsequences},\nonumber\\
 d_+(z)=0 &: \eps_N(z)\longrightarrow0.
 \label{eq:first-tier}
\end{align}

The condition $d_+(z)>0$ alone does not classify a point as finite-residual, since a sequence alternating between $1/2$ and $\ee^{-N}$ satisfies it. If $0<\alpha_-(z) =\alpha_+(z)=\alpha(z)<\infty$, then directly from the definition
\begin{equation}
 \eps_N(z)=\exp[-N^{\alpha(z)+o(1)}].
 \label{eq:alpha-class}
\end{equation}
Finite positive upper and lower orders give, for each sufficiently small $\delta>0$ and all sufficiently large $N$,
\begin{equation}
 \ee^{-N^{\alpha_+(z)+\delta}}
 \le\eps_N(z)\le
 \ee^{-N^{\alpha_-(z)-\delta}}.
 \label{eq:alpha-bounds}
\end{equation}
When $0<p_-(z)=p_+(z)=p(z)<\infty$, the corresponding statement is
\begin{equation}
 \eps_N(z)=N^{-p(z)+o(1)},
 \label{eq:p-class}
\end{equation}
with analogous bounds $N^{-(p_+(z)+\delta)}\le\eps_N(z)\le N^{-(p_-(z)-\delta)}$ when both orders are finite and positive. The reduction of Eq.~\maineqref{eq:main-six} to common orders uses
the exact inversion $\eps_N(z)=\ee^{-F_N(z)}$. If
$\ln(1+F_N(z))/\ln N\to\alpha(z)>0$, then
$1+F_N(z)=N^{\alpha(z)+o(1)}$ and $F_N(z)\to\infty$, whereupon the inversion gives
Eq.~\eqref{eq:alpha-class}. If $F_N(z)/\ln N\to p(z)\in(0,\infty)$,
then $F_N(z)=p(z)\ln N+o(\ln N)$, giving Eq.~\eqref{eq:p-class}.
The common $d(z)$ limit follows directly from the first pair. Existence of a
finite nonzero limit of $F_N(z)/N^{\alpha(z)}$ is a separate requirement, which
the common growth order leaves open. Exact zeros are recorded separately through
a discrete flag specifying whether exact eigenvalue coincidences occur only
finitely often, infinitely often but not eventually always, or at every
sufficiently large size.

Coincidence of all three lower/upper pairs still falls short of an exhaustive three-class picture: both $\eps_N(z)=\ee^{-(\ln N)^2}$ and $\eps_N(z)=1/\ln N$ have coincident paired indicators, yet neither has a finite positive exponent in Eq.~\eqref{eq:alpha-class} or \eqref{eq:p-class}. Reducing the descriptor therefore uses one further assumption.

\begin{assumption}[Regular-scale trichotomy]\label{ass:trichotomy}
Away from separately marked exact spectral points, every $z$ under consideration shows exactly one of the following behaviors:
\begin{enumerate}
\item $\eps_N(z)\to d(z)>0$;
\item $\eps_N(z)=\exp[-\kappa(z)N^{\alpha(z)}+o(N^{\alpha(z)})]$, with $0<\alpha(z),\kappa(z)<\infty$;
\item $\eps_N(z)=N^{-p(z)+o(1)}$, with $0<p(z)<\infty$.
\end{enumerate}
When the algebraic set is used as an \emph{interface}, we further assume it has no two-dimensional interior and separates the exponential and finite-residual domains apart from isolated critical contacts.
\end{assumption}

Under Assumption~\ref{ass:trichotomy}, the classes are those of
Eq.~\maineqref{eq:main-classes}. An algebraic locus becomes an interface only
through the geometric condition above, since an isolated algebraic closure point
need not be a boundary point of the closed vanishing set.

The three pairs fix the atlas classes while retaining only selected asymptotic
scales, since the residuals $N^{-2}/\ln N$ and $N^{-2}/\ln\ln N$ share
one descriptor and $\alpha(z)=1$ does not distinguish $\ee^{-N}$ from
$\ee^{-N\ln N}$. A quantitative exponential rate is therefore extra
information beyond the classification indicators, which we state model by
model. At an extensive exponential point we write
\begin{equation}
 I(z)=\lim_{N\to\infty}\frac{-\ln\eps_N(z)}N,
 \qquad \eps_N(z)=\ee^{-NI(z)+o(N)},
 \label{eq:extensive-rate}
\end{equation}
a limit the dissipative shift chain below attains with a closed expression for
$I(z)$. Both models satisfy the common-limit hypothesis at every multiplier, so
a single descriptor suffices throughout.

\subsection{Spectral Separation}
\label{sec:riesz}
Isolation of a true spectral cluster uses a contour projection that groups
generalized eigenspaces~\cite{Kato1995,Stewart1973}. For a contour $\Gamma\subset\Cclass$ enclosing a true
spectral cluster of $\Phi_N$, the Riesz projection is
\begin{equation}
 P_{A,N}=\frac1{2\pi\ii}\oint_\Gamma(zI-\Phi_N)^{-1}\,\dd z,
 \label{eq:riesz-projection}
\end{equation}
and the certificate of Sec.~\mainref{sec:atlas-meaning} is the following bound,
uniform in $N$ because the margin is.

\begin{proposition}[Uniform contour separation]\label{prop:riesz}
Let $\Gamma$ be a fixed, positively oriented, rectifiable simple closed contour with length $\ell_\Gamma$, and suppose that, for all sufficiently large $N$,
\[\inf_{z\in\Gamma}\eps_N(z)\ge\delta>0.\]
Then Eq.~\eqref{eq:riesz-projection} obeys
\begin{equation}
 \norm{P_{A,N}}\le\frac{\ell_\Gamma}{2\pi\delta}.
 \label{eq:riesz-bound}
\end{equation}
If $\widetilde\Phi_N=\Phi_N+E_N$, with $\norm{E_N}\le\eta<\delta$, the same contour defines $\widetilde P_{A,N}$ and
\begin{equation}
 \norm{\widetilde P_{A,N}-P_{A,N}}
 \le\frac{\ell_\Gamma}{2\pi}
       \frac{\eta}{\delta(\delta-\eta)}.
 \label{eq:riesz-perturb}
\end{equation}
\end{proposition}

\begin{proof}
The resolvent norm on $\Gamma$ is at most $\delta^{-1}$, and integrating it gives Eq.~\eqref{eq:riesz-bound}. A Neumann series bounds the perturbed resolvent by $(\delta-\eta)^{-1}$, whereupon the resolvent identity gives Eq.~\eqref{eq:riesz-perturb}. If $d_-(z)>0$ at every point of a fixed compact $\Gamma$, then Proposition~\ref{prop:lipschitz}, a finite cover, and the definition of $\liminf$ supply a common $\delta$.
\end{proof}

The projection collects the generalized eigenspaces of the \emph{true} eigenvalues inside $\Gamma$, so a pseudomode is not an additional Riesz eigenstate and an empty spectral interior gives $P_{A,N}=0$. Disjoint clusters satisfy $P_AP_B=0$, though generally $P_A^\dagger\ne P_A$. If the cluster excludes $1$, trace preservation implies
\begin{equation}
 \Tr(P_AX)=0\quad\text{for every }X,
 \label{eq:riesz-tracezero}
\end{equation}
since $\Tr[(z-\Phi)^{-1}X]=\Tr X/(z-1)$, whose contour integral vanishes. A nonzero $P_A\rho$ is therefore not a density matrix, and $P_A$ need not be positive or completely positive. Its contribution to a measured signal, $\Tr[O\Phi_N^nP_{A,N}\delta\rho]$, requires suitable overlaps with both the preparation $\delta\rho$ and the observable $O$ \cite{Macieszczak2016}.

Merging vanishing-residual domains withdraws this uniform separation certificate. Section~\mainref{sec:riesz-toy} identifies the quantity that then becomes singular: the susceptibility of a cluster projection to cross-cluster perturbations.

\section{Dissipative shift chain}
\label{app:shift-proof}
This section proves the complete-channel claims for
Eq.~\maineqref{eq:main-shift-channel}.

\subsection{Exact Reduction}
\label{app:shift-reduction}\label{sec:SSH}
A collision realization~\cite{Ciccarello2022} uses one maximally mixed ancillary spin per site. Conditional on activation, first swap system spin $j$ with ancilla $j$, then swap ancilla $j$ with system spin $j+1$ for $j<N$, and finally discard all ancillas. Both steps are layers of disjoint local swaps on the enlarged system, and exchange interactions generate each swap:
\begin{equation}
 \ee^{-\ii\frac\pi4(\sigma^x\otimes\sigma^x+
 \sigma^y\otimes\sigma^y+\sigma^z\otimes\sigma^z)}
 =\ee^{-\ii\pi/4}\,\mathrm{SWAP}.
 \label{eq:swap}
\end{equation}

The uniform rotation follows every period, and one shared Bernoulli variable
selects activation, giving the stated repeated-interaction
channel~\cite{Ciccarello2022}. Iterating the reset $N$ times gives
\begin{equation}
 \cE_N^N(X)=2^{-N}\id\,\Tr X.
 \label{eq:reset-nilpotent}
\end{equation}

An orthonormal one-spin operator basis is
\begin{equation}
 \id,\quad X=\sigma^x,\quad
 Q_\pm=\frac{\sigma^y\pm\ii\sigma^z}{\sqrt2},
 \label{eq:toy-basis}
\end{equation}
with rotation phases $1,1,-\ii,+\ii$ respectively. In any nonidentity tensor-product word, remove the leading and trailing identities, after which the remaining core has width $\ell$ and $L=N-\ell+1$ allowed positions in the chain. Its translates are mutually orthogonal, and on their span $\cE_N$ acts as the nilpotent shift
\begin{equation}
 S_Le_j=e_{j+1}\ (j<L),\qquad S_Le_L=0.
 \label{eq:shift}
\end{equation}
All translates acquire the same rotation phase $\zeta$ with $\zeta^4=1$, so the corresponding block is
\begin{equation}
 \Phi_{L,\zeta}=\zeta(a\id_L+gS_L).
 \label{eq:toy-block}
\end{equation}
Every nonidentity word belongs to exactly one such orbit, and the identity forms a separate one-dimensional block with eigenvalue one. The decomposition is therefore orthogonal and exhausts the full $4^N$-dimensional space.

Define
\begin{equation}
 \delta_L(\rho)=\smin(\rho\id_L-gS_L),\qquad \rho\ge0.
 \label{eq:delta-scalar}
\end{equation}
Diagonal phase unitaries remove the phases of the diagonal and subdiagonal entries of $\Phi_{L,\zeta}-z$, and $\delta_{L+1}(\rho)\le\delta_L(\rho)$ because the last $L$ coordinates form an invariant subspace of $S_{L+1}$ carrying $S_L$, so only the longest orbit of each phase need be retained. Cores of width one realize $\zeta=1,\ii,-\ii$ and width two is needed for $\zeta=-1$, so for $N\ge2$
\begin{align}
 \eps_N(z)=\min\bigl\{&|1-z|,\delta_N(|z-a|),\delta_N(|z-\ii a|),\nonumber\\
 &\delta_{N-1}(|z+a|),\delta_N(|z+\ii a|)\bigr\}.
 \label{eq:toy-full}
\end{align}
In particular,
\begin{equation}
 \spec\Phi_N=\{1,a,\ii a,-a,-\ii a\}.
 \label{eq:toy-spectrum}
\end{equation}
Five exact eigenvalues therefore generate the extended regions derived below,
and the full minimum reduces the atlas problem to a scalar estimate.

\begin{lemma}[Scalar shift residual]\label{lem:scalar}
For fixed $g>0$, the following hold as $L\to\infty$:
\begin{align}
 (g-\rho)(\rho/g)^L&\le\delta_L(\rho)
 \le g(\rho/g)^L,&&0<\rho<g,\label{eq:scalar-exp}\\
 \delta_L(g)&=2g\sin\frac{\pi}{4L+2}
 \sim\frac{\pi g}{2L},\label{eq:scalar-boundary}\\
 \delta_L(\rho)&\longrightarrow\rho-g,&&\rho>g.
 \label{eq:scalar-outside}
\end{align}
At $\rho=0$, the residual is exactly zero.
\end{lemma}
\begin{proof}
For $\rho>0$, nilpotency gives
\begin{equation}
 (\rho\id-gS_L)^{-1}
 =\rho^{-1}\sum_{j=0}^{L-1}(g/\rho)^jS_L^j.
 \label{eq:scalar-inverse}
\end{equation}
For $\rho<g$ the triangle inequality bounds its norm above by $(g-\rho)^{-1}(g/\rho)^L$, while its bottom-left entry has modulus $g^{L-1}/\rho^L$, yielding the opposite bound and hence Eq.~\eqref{eq:scalar-exp}. At $\rho=g$ the squared singular-value problem is a second-difference equation with one modified endpoint, whose eigenvalues $4g^2\sin^2[(2j-1)\pi/(4L+2)]$, $1\le j\le L$, prove Eq.~\eqref{eq:scalar-boundary}. For $\rho>g$, $\norm{(\rho\id-gS_L)v}\ge(\rho-g)\norm v$, and the normalized constant vector gives an upper bound whose square is $(\rho-g)^2+[\rho^2-(\rho-g)^2]/L$, proving Eq.~\eqref{eq:scalar-outside}.
\end{proof}
Combining the full-space minimum with Lemma~\ref{lem:scalar} proves
Eq.~\maineqref{eq:main-shift-atlas}. In the exponential interior the largest
orbit lengths $N$ or $N-1$ give an extensive residual rate,
\begin{equation}
 -\ln\eps_N(z)=NI(z)+O(1),\qquad
 I(z)=\ln\frac{g}{\rho(z)},
 \label{eq:toy-rate}
\end{equation}
while on an exposed arc all relevant longest blocks
share the leading prefactor $\pi g/(2N)$. The
stationary block cannot produce a smaller limit there, since
$\rho(z)-g\le|z-a|-g\le|z-1|$.

The exponential region is the union of the four open disks of radius $g$ centered at $a\zeta$, and only the exposed circle arcs carry the algebraic label, since an arc covered by another open disk stays exponentially small. Writing $z=x+\ii y$, the exposed candidate equation is
\begin{equation}
 x^2+y^2+(1-g)^2-g^2
 =2(1-g)\max(|x|,|y|).
 \label{eq:toy-boundary}
\end{equation}
Its left side minus right side is negative in $\Iclass$ and positive in $\Cclass$.

Adjacent centers are separated by $\sqrt2(1-g)$, so first contact occurs at
\begin{equation}
 g_{\rm merge}=\sqrt2-1.
 \label{eq:toy-merge}
\end{equation}
For $\sqrt2-1<g<1/2$ neighboring disks overlap while opposite disks do not, and the union encloses a central finite-residual hole. At
\begin{equation}
 g_{\rm hole}=\frac12,
 \label{eq:toy-hole}
\end{equation}
the hole closes, and for $g>1/2$ every disk contains the origin, making the union star-shaped about it. At the origin,
\begin{equation}
 \begin{cases}
 d(0)=1-2g,&g<1/2,\\
 \eps_N(0)\sim\pi/(4N),&g=1/2,\\
 I(0)=\ln[g/(1-g)],&g>1/2.
 \end{cases}
 \label{eq:toy-origin}
\end{equation}
The isolated algebraic point at $g=1/2$ is not a boundary point of the closed disk union, which illustrates the geometric qualification in Assumption~\ref{ass:trichotomy}.

The two contacts correspond, as derived next, to simultaneous criticality in
different numbers of auxiliary Su--Schrieffer--Heeger (SSH) families. For a fixed
shift sector, let
\begin{equation}
 A_{L,\zeta}(z)=\zeta(a\id_L+gS_L)-z\id_L,
 \qquad \rho_\zeta=|z-a\zeta|.
 \label{eq:SSH-A}
\end{equation}
Its Hermitian doubling is
\begin{equation}
 \mathbb H_{L,\zeta}(z)=
 \begin{pmatrix}0&A_{L,\zeta}(z)\\A_{L,\zeta}(z)^\dagger&0\end{pmatrix},
 \qquad
 \Gamma=\begin{pmatrix}\id&0\\0&-\id\end{pmatrix}.
 \label{eq:SSH-double}
\end{equation}
It obeys $\Gamma\mathbb H\Gamma=-\mathbb H$ with eigenvalues $\pm s_j(A_{L,\zeta})$, so a left/right singular-vector pair of the original residual is an eigenvector pair of a chiral Hermitian Hamiltonian, a doubling with precedents in the non-Hermitian and directional-amplification literature \cite{Porras2019,Gong2018}.

The correspondence is explicit: phase unitaries make the diagonal and subdiagonal
couplings real, and interleaving the two sublattices identifies
Eq.~\eqref{eq:SSH-double} with a Su--Schrieffer--Heeger chain \cite{SSH1979} of
intracell and intercell couplings
\begin{equation}
 t_1=\rho_\zeta,\qquad t_2=g.
 \label{eq:SSH-couplings}
\end{equation}
A shift of Bloch momentum gives the bands
\begin{equation}
 E_\pm(k)=\pm\sqrt{\rho_\zeta^2+g^2+2\rho_\zeta g\cos k}.
 \label{eq:SSH-bands}
\end{equation}
The distance of the bulk bands from zero is
\begin{equation}
 \Delta_{\rm aux}(z,g)=|\rho_\zeta-g|.
 \label{eq:SSH-gap}
\end{equation}
This vanishes exactly on the residual circle $|z-a\zeta|=g$. With the convention $S\leftrightarrow\beta$, the point-gap winding is
\begin{equation}
 \nu_\zeta(z)=\frac1{2\pi\ii}\int_0^{2\pi}
 \partial_k\ln[\zeta(a+g\ee^{\ii k})-z]\,\dd k
 =\begin{cases}1,&\rho_\zeta<g,\\0,&\rho_\zeta>g,\end{cases}
 \label{eq:SSH-winding}
\end{equation}
and is undefined when the circle passes through $z$.

For $0<\rho_\zeta<g$ the semi-infinite zero-mode recurrence has amplitude ratio $-\rho_\zeta/g$ per cell, so
\begin{equation}
 \xi_\zeta^{-1}(z)=\ln(g/\rho_\zeta),
 \qquad
 \eps_L^{(\zeta)}(z)\asymp\ee^{-L/\xi_\zeta}.
 \label{eq:SSH-xi}
\end{equation}
In a finite chain the two opposite-edge singular modes hybridize and split by an exponentially small energy, becoming exact zero modes only at the dimerized spectral center~\cite{SSH1979,Kunst2018}. At $\rho_\zeta=g$ their localization length diverges and the exact finite-section result is $2g\sin[\pi/(4L+2)]$. The $\Iclass/\pclass/\Cclass$ distinction thus acquires an exact auxiliary reading: edge-mode splitting in the nontrivial phase, a critical extended mode with $p=1$, and a trivial gapped singular-value problem.

The first atlas contact is located by following the midpoint between adjacent
centers, $z_m(g)=a\zeta(1+\ii)/2$. Direct substitution gives
\begin{align}
 \rho_\zeta(z_m)=\rho_{\ii\zeta}(z_m)&=a/\sqrt2,\nonumber\\
 \rho_{-\zeta}(z_m)=\rho_{-\ii\zeta}(z_m)&=a\sqrt{5/2}.
 \label{eq:SSH-adjacent-contact}
\end{align}
The two adjacent-family bulk gaps are $|a/\sqrt2-g|$ and close at
$a/\sqrt2=g$, that is at $g=g_1=\sqrt2-1$, where the other two gaps remain
positive. Along this midpoint path the adjacent windings are both zero below
$g_1$ and both one above it, with inverse localization lengths
$\ln(\sqrt2g/a)$ on the nontrivial side. The four choices of $\zeta$ give four
simultaneous pairwise contacts, with no third family covering their algebraic
residual, and at the origin the auxiliary gap is still $a-g=3-2\sqrt2>0$ when
$g=g_1$.

The second contact occurs at $z=0$, where every family has $\rho_\zeta=a$, and
Eqs.~\eqref{eq:SSH-gap} and \eqref{eq:SSH-winding} give
\begin{equation}
 \begin{gathered}
  \Delta_{\rm aux}(0,g)=|1-2g|,\qquad
  \nu_\zeta(0)=\begin{cases}0,&g<g_2,\\1,&g>g_2,\end{cases}\\
  \xi_\zeta^{-1}(0)=\ln(g/a),\qquad g>g_2=\tfrac12.
 \end{gathered}
 \label{eq:SSH-origin-contact}
\end{equation}
All four windings therefore change together at the hole closure. At either type
of contact the critical blocks have lengths $N$ or $N-1$, at least one of them
$N$, so
Eq.~\eqref{eq:scalar-boundary} and the full-space minimum give
$\eps_N\sim\pi g/(2N)$. At $g_2$ in particular
$\eps_N(0)=\sin[\pi/(4N+2)]$, consistent with Eq.~\eqref{eq:toy-origin}.

The same three classes have an equivalent reading through the root moduli
$\beta_\zeta(z)=(z/\zeta-a)/g$ of the four-branch symbol
$f_z(\beta)=z^4-(a+g\beta)^4$. The conditions
$|\beta_\zeta|\lessgtr1$ that classify the boundary mismatch of the measurable
texture in Eq.~\maineqref{eq:main-shift-beta} are exactly the root-modulus
conditions of $f_z$, and the atlas boundaries are the corresponding amoeba
walls \cite{Wang2024,Kaneshiro2026}. The rate of Eq.~\eqref{eq:toy-rate} is
$I(z)=-\min_\zeta\ln|\beta_\zeta(z)|$ wherever this is positive, so the root
count carries the integer winding while the radial root positions carry the
continuous rate.

\subsection{Spin-wave Readouts}
\label{app:shift-readouts}
The orthogonal shift orbits of Eq.~\eqref{eq:toy-block} yield the finite-chain
texture propagation and the four response families measured in
Sec.~\mainref{sec:shift-observables}.

Prepare
\begin{equation}
 \rho_0=\bigotimes_{j=1}^N\frac{I+t_j\sigma_j^x}{2},
 \qquad t_j=m\cos(k_0j),\quad |m|<1.
 \label{eq:prepared-texture}
\end{equation}
For $f_j(n)=\Tr[\sigma_j^x\Phi_N^n\rho_0]$, partial tracing and reinjection give the exact recursion
\begin{equation}
 f_j(n+1)=af_j(n)+gf_{j-1}(n),\qquad f_0(n)=0.
 \label{eq:magnetization-recursion}
\end{equation}
This is exact at finite $m$, with no linearization of the density operator. Explicitly,
\begin{equation}
 f_j(n)=\sum_{s=0}^{\min(n,j-1)}
 \binom{n}{s}a^{n-s}g^s f_{j-s}(0).
 \label{eq:magnetization-binomial}
\end{equation}
The propagation distance before truncation has mean $gn$ and variance $ng(1-g)$, and a spatially periodic bulk input $f_j\propto\ee^{-\ii qj}$ has response multiplier $\lambda_1(q)=a+g\ee^{\ii q}$. For a local $n$-cycle readout the bulk formula is exact at $j>n$, and cosine and sine preparations reconstruct the complex amplitude from two real experiments.

Single-site $yz$ helices and a two-site quadrupolar observable complete the four response families
\begin{equation}
 \lambda_\zeta(q)=\zeta(a+g\ee^{\ii q}),
 \qquad \zeta\in\{1,\ii,-1,-\ii\}.
 \label{eq:physical-response-circles}
\end{equation}
For instance, $\langle\sigma_j^y\rangle=m\cos(qj)$ and $\langle\sigma_j^z\rangle=-m\sin(qj)$ give the $+\ii$ multiplier for the measured combination $\langle\sigma_j^y\rangle+\ii\langle\sigma_j^z\rangle$, while the Hermitian two-site operator $(\sigma_j^y\sigma_{j+1}^y-\sigma_j^z\sigma_{j+1}^z)/\sqrt2$ changes sign under the pulse and supplies the $-1$ family.

For $g<1/2$ the phase range of the first circle is $[-\alpha,\alpha]$ with $\alpha=\arcsin[g/(1-g)]$, so the four ranges leave a phase-coverage gap
\begin{equation}
 T\Delta\omega(g)=
 [\pi/2-2\arcsin(g/(1-g))]_+,
 \label{eq:response-phase-gap}
\end{equation}
which closes exactly at $g_1=\sqrt2-1$, where the four atlas lobes first connect. The excluded band is a gap in the allowed \emph{phases of bulk response multipliers}, not in real-frequency absorption or the relaxation gap. The phase winding of one branch,
\begin{equation}
 \nu(g)=\frac1{2\pi\ii}\int_{-\pi}^{\pi}
          \partial_q\ln[a+g\ee^{\ii q}]\,\dd q,
 \label{eq:response-winding}
\end{equation}
changes from zero to one at $g_2=1/2$ and is undefined there because $\lambda_1(\pi)=0$. It is a point-gap invariant of the response family, and the finite-chain eigenvalue set stays fixed \cite{Gong2018}. The identity $|\lambda_\zeta(q;g)|^2=1-4g(1-g)\sin^2(q/2)$ shows in addition that $g$ and $1-g$ share modal envelopes while differing in response winding.

An economical test uses $q_*=3\pi/4$: reconstruct $\lambda_x(q_*)$ from phase-cycled $x$ textures and $\lambda_\perp(-q_*)=\ii\lambda_x(-q_*)$ from a $yz$ helix. Their difference is
\begin{equation}
 \lambda_x(q_*)-\lambda_\perp(-q_*)
 =[1-(1+\sqrt2)g](1-\ii),
 \label{eq:shift-response-difference}
\end{equation}
so the normalized phase-resolved one-cycle contrast
\begin{equation}
 \mathcal D_{x\perp}(g)=1-(1+\sqrt2)g
 \label{eq:shift-crossing-readout}
\end{equation}
changes sign at the first merger. At the crossing the two responses coincide at $(2-\sqrt2)(1+\ii)/2$ with the finite decay time $-T/\ln(\sqrt2-1)$, so the protocol identifies the coalescence of two \emph{response-family endpoints}.

For the second threshold, prepare a staggered texture $f_j(0)=m(-1)^j$. The bulk normalized staggered magnetization is
\begin{equation}
 M_\pi(n)=(1-2g)^n.
 \label{eq:staggered-observable}
\end{equation}
For $g<1/2$ the pattern decays without alternating sign, at $g=1/2$ its bulk component disappears in one step, and for $g>1/2$ it alternates while decaying, in each case through cancellation between the unshifted and one-site-shifted textures. The bulk average excludes the leftmost site, whose inclusion would add $g/N$.

\subsection{Cross-family Conditioning}
\label{app:cross-family}
This section proves the perturbation scales quoted in
Sec.~\mainref{sec:riesz-toy}, where the merger of two disks costs the response
families their isolation. On the transient space, put $\mathcal U=\Ad_{R_N}$.
The four unperturbed phase projectors are independent of $g$ and have unit norm,
\begin{equation}
 P_\zeta=\frac14\sum_{k=0}^3\zeta^{-k}\mathcal U^k,
 \qquad\norm{P_\zeta}=1,
 \label{eq:fixed-phase-projectors}
\end{equation}
so what becomes singular at a merger is not the norm of these projections.

A block construction quantifies how a perturbation coupling two response families affects their classification~\cite{Stewart1973,Kato1995}. For blocks $A,B$ with disjoint spectra, take
\begin{align}
 T(u)&=\begin{pmatrix}A&uE\\0&B\end{pmatrix},\nonumber\\
 P_A(u)&=\begin{pmatrix}I&uX\\0&0\end{pmatrix},\qquad AX-XB=E.
 \label{eq:triangular-riesz}
\end{align}
The inverse Sylvester operator, standard in invariant-subspace perturbation theory \cite{Stewart1973,BhatiaRosenthal1997}, controls the susceptibility:
\begin{equation}
 \chi_{A,B}^{(\nu)}
 =\sup_{E\ne0}\frac{\norm{(X\mapsto AX-XB)^{-1}E}_\nu}
                        {\norm E_\nu},
 \label{eq:sylvester-susceptibility}
\end{equation}
where $\nu$ denotes the matrix spectral or Frobenius norm. The two differ only in finite-size prefactors, so the leading classifications below hold for either.

\begin{proposition}[Cross-family conditioning at a disk contact]\label{prop:sylvester}
Take equal-length translation blocks
$A=\zeta_1(aI_L+gS_L)$ and $B=\zeta_2(aI_L+gS_L)$, and let $d=a|\zeta_1-\zeta_2|>0$. Then
\begin{equation}
\begin{array}{ll}
 d>2g:&\displaystyle\chi_{A,B}^{(\nu)}\longrightarrow(d-2g)^{-1},\\[1mm]
 d=2g:&\displaystyle\chi_{A,B}^{(\nu)}=\Theta(L/g),\\[1mm]
 d<2g:&\displaystyle\ln\chi_{A,B}^{(\nu)}
       =2L\ln(2g/d)+O(\ln L).
\end{array}
\label{eq:sylvester-scales}
\end{equation}
\end{proposition}

\begin{proof}
For the triangular family in Eq.~\eqref{eq:triangular-riesz}, the Sylvester solution is unique whenever $\sigma(A)\cap\sigma(B)=\varnothing$. Independent diagonal phase unitaries on the two translation chains, together with a scalar phase, reduce the Sylvester map to $dI-g\mathcal Q$ with $\mathcal Q(X)=S_LX+XS_L$ and $d>0$.

In either norm $\norm{\mathcal Q}\le2$, and left and right multiplication commute, so $\mathcal Q^{2L-1}=0$ and the inverse is exactly
\begin{equation}
 \mathcal S^{-1}=\frac1d\sum_{k=0}^{2L-2}(g/d)^k\mathcal Q^k.
 \label{eq:sylvester-series}
\end{equation}

For $d>2g$ this gives
\begin{equation}
 \chi_{A,B}^{(\nu)}\le(d-2g)^{-1}.
 \label{eq:sylvester-outside-upper}
\end{equation}
Take a smooth normalized packet $f$ for which both
$\norm{(I-S_L)f}=O(L^{-1})$ and
$\norm{(I-S_L^\dagger)f}=O(L^{-1})$.
For $X=ff^\dagger$, $\norm X_\nu=1$ and
\begin{equation}
 \norm{[dI-g\mathcal Q]X-(d-2g)X}_\nu=O(g/L),
 \label{eq:sylvester-trial}
\end{equation}
which matches the limit outside. At $d=2g$ the same trial gives $\chi\ge cL/g$, while the $2L-1$ terms of Eq.~\eqref{eq:sylvester-series} give $\chi\le(2L-1)/(2g)$.

For $d<2g$, summing the geometric norm bound gives
\begin{equation}
 \chi\le\frac1d\sum_{k=0}^{2L-2}(2g/d)^k.
 \label{eq:sylvester-inside-upper}
\end{equation}
For $E=|1\rangle\langle L|$, only the term with exactly $L-1$ left and $L-1$ right shifts contributes to $X_{L1}$, giving
\begin{equation}
 X_{L1}=\frac1d\binom{2L-2}{L-1}(g/d)^{2L-2}.
 \label{eq:sylvester-corner}
\end{equation}

Its binomial coefficient has logarithm $(2L-2)\ln2+O(\ln L)$, and since the absolute
value of an entry is bounded by either matrix norm, this gives the lower exponential
bound and Eq.~\eqref{eq:sylvester-scales}.
\end{proof}

For adjacent rotation sectors $d=\sqrt2(1-g)$, so the condition $d=2g$ is exactly
$g_1=\sqrt2-1$, while for opposite sectors it is $g_2=1/2$. Equal-length
representative trajectories can be chosen with $L=N-O(1)$ in the full algebra,
making these genuine size-dependent cross-family susceptibilities.

The susceptibility maximizes over unrestricted cross-block matrices, so
symmetry-preserving perturbations can leave the projectors exactly fixed. Taken
with Eq.~\eqref{eq:fixed-phase-projectors}, this identifies what the merger
costs: the projections themselves keep unit norm, while their susceptibility to
cross-family perturbations passes from bounded to algebraic to exponential in
the system size.

\section{Chiral XY model}
\label{app:XY-proof}
The Hamiltonian, balanced jumps, and matching parameters throughout this
section are those of
Eqs.~\maineqref{eq:main-XY-H}--\maineqref{eq:main-XY-directed}.

\subsection{Operator-space Reduction}
\label{app:XY-operator-space}
Jordan--Wigner transformation~\cite{JordanWigner1928,LiebSchultzMattis1961,Katsura1962} gives, up to a scalar,
\begin{equation}
 H_N(t)=\frac v2\sum_{j=1}^{N-1}
 (c_j^\dagger c_{j+1}+c_{j+1}^\dagger c_j)
 +h(t)\sum_jc_j^\dagger c_j,
 \label{eq:XY-fermion}
\end{equation}
and the balanced bond and site jumps yield the complex damping matrix and
one-particle Hamiltonian matrix
\begin{align}
 M_N&=\gamma\id+\frac{\ii v}{2}(S_N-S_N^\dagger),\nonumber\\
 h_N(t)&=\frac v2(S_N+S_N^\dagger)+h(t)\id.
 \label{eq:matched-matrices}
\end{align}

Their reverse-shift terms cancel as in Eq.~\maineqref{eq:main-XY-cancellation}.
A diagonal phase unitary and a circular particle/hole basis give
\begin{equation}
 A_\pm(t)=-\gamma\id+vS_N\pm\ii h(t)\id.
 \label{eq:Apm}
\end{equation}

These are structure matrices, and they determine the full spin problem only
after the HS-unitary treatment of both physical parities carried out next.
Ordered products of $2N$ Majoranas with fixed phases form an orthonormal HS basis~\cite{Prosen2008,BarthelZhang2022}. For a Hermitian linear Majorana jump $L=\ell^{\mathsf T}w$ with real $\ell$, put $M=\sum_\mu\ell_\mu\ell_\mu^{\mathsf T}$. A real orthogonal change of Majorana basis diagonalizes $M$ without changing Hilbert--Schmidt singular values, and in that basis the jumps $\sqrt{m_a}\,w_a$ suffice.

Let $w_A$ be an ordered product of $s$ distinct Majoranas, so that $w_aw_Aw_a=(-1)^s w_A$ for $a\notin A$ and the sign is $(-1)^{s-1}$ for $a\in A$. Since
\begin{equation}
 \mathcal D[\sqrt{m_a}w_a]w_A
 =m_a(w_aw_Aw_a-w_A),
\end{equation}
the dissipative eigenvalue is $-2\sum_{a\in A}m_a$ on even $s$ and $-2\sum_{a\notin A}m_a$ on odd $s$. Restoring basis covariance gives respectively
\begin{align}
 \mathcal D^{(s)}&=-2\DG_s(M),&&s\text{ even},\nonumber\\
 \mathcal D^{(s)}&=2\DG_s(M)-2\tr M\,\id,&&s\text{ odd}.
 \label{eq:noise-degrees}
\end{align}
The Hamiltonian is quadratic in Majoranas and its commutator is a derivation, so on every degree it acts as $\DG_s(K)$ with $K^{\mathsf T}=-K$. For odd $s$ the combined generator is
\begin{equation}
 \mathcal L^{(s)}=\DG_s(K)+2\DG_s(M)-2\tr M\,\id.
 \label{eq:odd-sector}
\end{equation}

The Hodge map sends an ordered subset to its ordered complement with the permutation sign, and is complex linear and unitary on the coefficient Hilbert spaces. The identity
\begin{equation}
 *\DG_s(A)*^{-1}=\tr A\,\id-\DG_{2N-s}(A^{\mathsf T})
 \label{eq:hodge}
\end{equation}
holds for matrix units and then by linearity. Since $\tr K=0$, $M^{\mathsf T}=M$, and $K^{\mathsf T}=-K$, it transforms the odd-degree generator into
\begin{align}
 *[\DG_s(K)+2\DG_s(M)-2\tr M]*^{-1}
 &=\DG_{2N-s}(K-2M).
 \label{eq:hodge-generator}
\end{align}
Because $2N$ is even, complementing the odd degrees permutes the odd part of the full exterior algebra while leaving the even degrees unchanged, so no operator-space dimension is lost.

Let $B(t)$ solve $\dot B=(K-2M)B$ with $B(0)=\id$. The exterior-power identity
\begin{equation}
 \frac\dd{\dd t}\bigwedge^s B(t)
 =\DG_s(K-2M)\bigwedge^s B(t)
\end{equation}
holds even when the matrices at different times fail to commute, making the full channel unitarily equivalent to $\Gamma(B)=\bigoplus_s\bigwedge^s B$. For the number-conserving matched model a constant circular basis reduces $K-2M$ to the complex particle/hole matrices of Eq.~\eqref{eq:Apm}, and the exterior identity
\begin{equation}
 \bigwedge^r(V_+\oplus V_-)
 \cong\bigoplus_{r_++r_-=r}
 \bigwedge^{r_+}V_+\otimes\bigwedge^{r_-}V_-
\end{equation}
then gives Eq.~\eqref{eq:full-exterior}. Singular values are transferable here
because of the special real-noise structure together with unitarity of the
reductions, neither of which follows from spectral third
quantization~\cite{BarthelZhang2022} alone: that reduction keeps the distinction
between even and odd sectors, whereas the one here retains the odd physical
Majorana degrees and complements them, covering the full spin algebra.

Since the time-dependent scalar field in Eq.~\eqref{eq:Apm} commutes with the rest, the exact full decomposition is
\begin{equation}
 \Phi_N\simeq_U
 \bigoplus_{r_+=0}^{N}\bigoplus_{r_-=0}^{N}
 a_{r,q}\,\ee^{GQ_{N;r_+,r_-}},
 \label{eq:full-exterior}
\end{equation}
where
\begin{align}
 r&=r_++r_-,\qquad q=r_+-r_-,\nonumber\\
 a_{r,q}&=\ee^{-rD+\ii q\theta},\nonumber\\
 Q_{N;r_+,r_-}
 &=\DG_{r_+}(S_N)\otimes\id
 +\id\otimes\DG_{r_-}(S_N).
 \label{eq:Q}
\end{align}
The block dimensions $\binom N{r_+}\binom N{r_-}$ sum to $4^N$, so
\begin{equation}
 \eps_N(z)=\min_{0\le r_\pm\le N}
 \smin\bigl(a_{r,q}\ee^{GQ_{N;r_+,r_-}}-z\id\bigr).
 \label{eq:XY-full-residual}
\end{equation}
The $r=0$ block is the identity eigenoperator, and nilpotency of $Q$ makes every other eigenvalue $a_{r,q}$. At fixed $D,\theta$ the entire finite-size eigenvalue multiset is independent of $G$, while the pseudospectrum is not.

From $\norm{\DG_s(S_N)}\le s$,
\begin{equation}
 \norm{Q_{N;r_+,r_-}}\le r,\qquad
 \norm{a_{r,q}\ee^{GQ}}\le\ee^{-r(D-G)},
 \label{eq:grade-bound}
\end{equation}
and therefore
\begin{equation}
 \eps_N^{(r,q)}(z)
 \ge\bigl[|z|-\ee^{-r(D-G)}\bigr]_+.
 \label{eq:grade-exclusion}
\end{equation}
At each fixed $z\ne0$, sufficiently large $r$ are thus uniformly unable to produce a vanishing residual, which carries the fixed-degree estimates into an exact large-size conclusion about the full many-body space.

At fixed degree, increasing $N$ embeds the previous block in the trailing-coordinate invariant subspace of the next, an embedding preserved by exterior powers, so the fixed-sector residuals and their full minimum are nonincreasing in $N$. In this model the ordinary residual limit $d(z)$ therefore exists before any growth exponent is analyzed.

\subsection{Physical Readouts}
\label{app:readout-support}
The proofs below treat the even occupation and joint-occupation protocol of
Sec.~\mainref{sec:xy-observables} and the finite-packet corrections to its
broad-packet limits.

For the encoded modes of Eq.~\maineqref{eq:main-XY-preparation}, orthogonality of
the disjoint $0,1$ envelopes gives $\{d_a,d_b^\dagger\}=\delta_{ab}$ for
$a,b\in\{0,1\}$, and the canonical staggered mode enters the separate
degree-two preparation. Hence $n_a=d_a^\dagger d_a$ is a projection, $Q_a=I-2n_a$
a Hermitian involution, and $Q_0,Q_1$ commute. In an orthonormal encoded Majorana
basis $Q_a=-\ii w_{a1}w_{a2}$, so $A_2$ and $A_4$ are traceless monomials of
physical degrees two and four, whence $A_s^2=I$, $\norm{A_s}_{2,N}=1$, and
$2^{-N}(I\pm A_s)$ is positive with unit trace. Substituting the difference of
these states in the readout proves Eq.~\maineqref{eq:main-XY-readout}, including
its normalization at $n=0$.

For the normalized sine envelope, define
\begin{align}
 f_j&=\sqrt{\frac2{L+1}}\sin\frac{\pi j}{L+1},\quad 1\le j\le L,
 \nonumber\\
 w_s&=\sum_{j=1}^{L-s}f_jf_{j+s},
 \label{eq:sine-packet}
\end{align}
together with the associated envelope series
\begin{equation}
 K_L(t)=\sum_{s=0}^{L-1}w_s\frac{t^s}{s!}.
 \label{eq:packet-K}
\end{equation}

In either circular sector the propagator is
$\ee^{-nD\pm\ii n\theta}\ee^{nGS_N}$ after the known site gauge.
Its diagonal packet matrix element is $K_L(nG)$ for a smooth envelope and
$K_L(-nG)$ for the staggered version. Exterior propagation gives more generally
the determinant identity
\begin{equation}
 \left\langle\bigwedge_{a=1}^s f_a,
 \left(\bigwedge^s B\right)\bigwedge_{b=1}^s f_b\right\rangle
 =\det[\langle f_a,Bf_b\rangle]_{a,b=1}^s,
 \label{eq:packet-determinant}
\end{equation}
which follows by expanding both wedges in their ordered basis. For the two
ordered intervals a right packet cannot propagate leftwards, so the overlap
matrix is triangular with both diagonal entries $K_L(nG)$. The two circular
factors of $Q_\pi$ give the square
and the two two-packet determinants of $Q_0Q_1$ the fourth power:
\begin{align}
 E_{2,L}(n)&=\ee^{-2nD}K_L(-nG)^2,\nonumber\\
 E_{4,L}(n)&=\ee^{-4nD}K_L(nG)^4.
 \label{eq:finite-XY24}
\end{align}

These identities are exact for $N\ge2L$ and integer $n\ge0$, open boundaries
included, and they fix the comparison for a finite experiment. The
finite-envelope errors are bounded next at fixed time.

For the normalized sine packet, set $\alpha_L=\pi/(L+1)$. Direct trigonometric summation gives, for $0\le s<L$,
\begin{equation}
 w_s=\frac{(L-s)\cos(s\alpha_L)
            +\sin[(s+1)\alpha_L]/\sin\alpha_L}{L+1},
 \label{eq:sine-overlap-exact}
\end{equation}
with $w_s=0$ for $s\ge L$. Zero extension of the packet to a bilateral lattice with unitary shift $U$ gives
\begin{equation}
 1-w_s=\tfrac12\norm{f-U^sf}^2
       \le s^2(1-\cos\alpha_L),
 \label{eq:sine-overlap-bound}
\end{equation}
so that, for every fixed real $t$,
\begin{equation}
 |K_L(t)-\ee^t|
 \le (1-\cos\alpha_L)(|t|+t^2)\ee^{|t|}
 =O(L^{-2}).
 \label{eq:packet-uniform-bound}
\end{equation}
The same bound controls the finite-time correlation corrections. Near the even $2/4$ crossing the derivative of the logarithmic ratio is nonzero, so the implicit-function argument gives a crossing displacement $O(L^{-2})$. Equation~\maineqref{eq:main-XY-signals} thus holds at fixed cycle number, with an approximation that is not uniform on the boundary-emptying timescale.

Define $\lambda_{2,\star}=\ee^{-2(D+G)}$ and $\lambda_{4,\star}=\ee^{-4(D-G)}$.
For their operator residuals, zero extension to the bilateral shift $U$ gives
$\norm{(U-I)f}=\sqrt{2(1-\cos\alpha_L)}=:b_L$, an error that compression to the
open chain cannot increase, and staggering changes the corresponding approximate
shift eigenvalue from $+1$ to $-1$. In the orthogonal exterior representation the
normalized $A_2$ is a product of two staggered packet factors and $A_4$ of four
smooth ones, so the Leibniz rule for $Q$ and the norm bound for wedge products
give
\begin{equation}
 \norm{(Q-\mu_r)A_r}_{2,N}\le rb_L,
 \qquad \mu_2=-2,\quad \mu_4=4.
 \label{eq:packet-generator-residual}
\end{equation}
Both blocks have $q=0$, and the identity
\begin{equation}
 (\ee^{GQ}-\ee^{G\mu_r})A_r
 =\int_0^G\ee^{(G-t)Q}\ee^{t\mu_r}(Q-\mu_r)A_r\,\dd t
 \label{eq:packet-exponential-residual}
\end{equation}
with $\norm Q\le r$ bounds the right-hand side by $G\ee^{Gr}rb_L$, whereupon multiplying by $\ee^{-rD}$ proves Eq.~\eqref{eq:packet-operator-residual}.

\begin{equation}
 \begin{gathered}
 \epsilon_N(\lambda_{r,\star};A_r)
 \le rG\ee^{-r(D-G)}
       \sqrt{2\left(1-\cos\frac{\pi}{L+1}\right)},\\
 r=2,4.
 \end{gathered}
 \label{eq:packet-operator-residual}
\end{equation}

At the exposed even contact $z_c=\ee^{-8D/3}$ with $N=2L$, the packet bound gives
$\eps_N(z_c)\le\epsilon_N(z_c;A_r)\le C/N$ for $r=2,4$ with a constant
independent of $N$, which realizes from above the interface law
$\eps_N(z_c)=\Theta(N^{-1})$ quoted in Sec.~\mainref{sec:xy-atlas-contacts}.
The two scales separate accordingly: the certified validity time at $z_c$ grows
at least linearly in $N$, whereas the decay time stays size independent. The
$O(L^{-2})$ scalar correlation correction and the $O(L^{-1})$ full-operator
residual control different errors.

The balanced generator preserves physical Majorana degree and distinct degrees
are HS-orthogonal, so the two readouts never mix:
\begin{equation}
 \Tr[A_4\Phi_N^n(A_2)]=0.
 \label{eq:even-cross-grade-null}
\end{equation}

\putbib[references]
\end{bibunit}


\begin{thebibliography}{80}%
\makeatletter
\providecommand \@ifxundefined [1]{%
 \@ifx{#1\undefined}
}%
\providecommand \@ifnum [1]{%
 \ifnum #1\expandafter \@firstoftwo
 \else \expandafter \@secondoftwo
 \fi
}%
\providecommand \@ifx [1]{%
 \ifx #1\expandafter \@firstoftwo
 \else \expandafter \@secondoftwo
 \fi
}%
\providecommand \natexlab [1]{#1}%
\providecommand \enquote  [1]{``#1''}%
\providecommand \bibnamefont  [1]{#1}%
\providecommand \bibfnamefont [1]{#1}%
\providecommand \citenamefont [1]{#1}%
\providecommand \href@noop [0]{\@secondoftwo}%
\providecommand \href [0]{\begingroup \@sanitize@url \@href}%
\providecommand \@href[1]{\@@startlink{#1}\@@href}%
\providecommand \@@href[1]{\endgroup#1\@@endlink}%
\providecommand \@sanitize@url [0]{\catcode `\\12\catcode `\$12\catcode
  `\&12\catcode `\#12\catcode `\^12\catcode `\_12\catcode `\%12\relax}%
\providecommand \@@startlink[1]{}%
\providecommand \@@endlink[0]{}%
\providecommand \url  [0]{\begingroup\@sanitize@url \@url }%
\providecommand \@url [1]{\endgroup\@href {#1}{\urlprefix }}%
\providecommand \urlprefix  [0]{URL }%
\providecommand \Eprint [0]{\href }%
\providecommand \doibase [0]{https://doi.org/}%
\providecommand \selectlanguage [0]{\@gobble}%
\providecommand \bibinfo  [0]{\@secondoftwo}%
\providecommand \bibfield  [0]{\@secondoftwo}%
\providecommand \translation [1]{[#1]}%
\providecommand \BibitemOpen [0]{}%
\providecommand \bibitemStop [0]{}%
\providecommand \bibitemNoStop [0]{.\EOS\space}%
\providecommand \EOS [0]{\spacefactor3000\relax}%
\providecommand \BibitemShut  [1]{\csname bibitem#1\endcsname}%
\let\auto@bib@innerbib\@empty
\bibitem [{\citenamefont {Lindblad}(1976)}]{Lindblad1976}%
  \BibitemOpen
  \bibfield  {author} {\bibinfo {author} {\bibfnamefont {G.}~\bibnamefont
  {Lindblad}},\ }\bibfield  {title} {\bibinfo {title} {On the generators of
  quantum dynamical semigroups},\ }\href {https://doi.org/10.1007/BF01608499}
  {\bibfield  {journal} {\bibinfo  {journal} {Commun. Math. Phys.}\ }\textbf
  {\bibinfo {volume} {48}},\ \bibinfo {pages} {119} (\bibinfo {year}
  {1976})}\BibitemShut {NoStop}%
\bibitem [{\citenamefont {Gorini}\ \emph {et~al.}(1976)\citenamefont {Gorini},
  \citenamefont {Kossakowski},\ and\ \citenamefont {Sudarshan}}]{GKS1976}%
  \BibitemOpen
  \bibfield  {author} {\bibinfo {author} {\bibfnamefont {V.}~\bibnamefont
  {Gorini}}, \bibinfo {author} {\bibfnamefont {A.}~\bibnamefont
  {Kossakowski}},\ and\ \bibinfo {author} {\bibfnamefont {E.~C.~G.}\
  \bibnamefont {Sudarshan}},\ }\bibfield  {title} {\bibinfo {title} {Completely
  positive dynamical semigroups of {N}-level systems},\ }\href
  {https://doi.org/10.1063/1.522979} {\bibfield  {journal} {\bibinfo  {journal}
  {J. Math. Phys.}\ }\textbf {\bibinfo {volume} {17}},\ \bibinfo {pages} {821}
  (\bibinfo {year} {1976})}\BibitemShut {NoStop}%
\bibitem [{\citenamefont {Bukov}\ \emph {et~al.}(2015)\citenamefont {Bukov},
  \citenamefont {D'Alessio},\ and\ \citenamefont {Polkovnikov}}]{Bukov2015}%
  \BibitemOpen
  \bibfield  {author} {\bibinfo {author} {\bibfnamefont {M.}~\bibnamefont
  {Bukov}}, \bibinfo {author} {\bibfnamefont {L.}~\bibnamefont {D'Alessio}},\
  and\ \bibinfo {author} {\bibfnamefont {A.}~\bibnamefont {Polkovnikov}},\
  }\bibfield  {title} {\bibinfo {title} {Universal high-frequency behavior of
  periodically driven systems: from dynamical stabilization to {Floquet}
  engineering},\ }\href {https://doi.org/10.1080/00018732.2015.1055918}
  {\bibfield  {journal} {\bibinfo  {journal} {Adv. Phys.}\ }\textbf {\bibinfo
  {volume} {64}},\ \bibinfo {pages} {139} (\bibinfo {year} {2015})}\BibitemShut
  {NoStop}%
\bibitem [{\citenamefont {Oka}\ and\ \citenamefont
  {Kitamura}(2019)}]{OkaKitamura2019}%
  \BibitemOpen
  \bibfield  {author} {\bibinfo {author} {\bibfnamefont {T.}~\bibnamefont
  {Oka}}\ and\ \bibinfo {author} {\bibfnamefont {S.}~\bibnamefont {Kitamura}},\
  }\bibfield  {title} {\bibinfo {title} {{Floquet} engineering of quantum
  materials},\ }\href
  {https://doi.org/10.1146/annurev-conmatphys-031218-013423} {\bibfield
  {journal} {\bibinfo  {journal} {Annu. Rev. Condens. Matter Phys.}\ }\textbf
  {\bibinfo {volume} {10}},\ \bibinfo {pages} {387} (\bibinfo {year}
  {2019})}\BibitemShut {NoStop}%
\bibitem [{\citenamefont {Eckardt}(2017)}]{Eckardt2017}%
  \BibitemOpen
  \bibfield  {author} {\bibinfo {author} {\bibfnamefont {A.}~\bibnamefont
  {Eckardt}},\ }\bibfield  {title} {\bibinfo {title} {Colloquium: Atomic
  quantum gases in periodically driven optical lattices},\ }\href
  {https://doi.org/10.1103/RevModPhys.89.011004} {\bibfield  {journal}
  {\bibinfo  {journal} {Rev. Mod. Phys.}\ }\textbf {\bibinfo {volume} {89}},\
  \bibinfo {pages} {011004} (\bibinfo {year} {2017})}\BibitemShut {NoStop}%
\bibitem [{\citenamefont {Mori}(2023)}]{Mori2023}%
  \BibitemOpen
  \bibfield  {author} {\bibinfo {author} {\bibfnamefont {T.}~\bibnamefont
  {Mori}},\ }\bibfield  {title} {\bibinfo {title} {Floquet states in open
  quantum systems},\ }\href
  {https://doi.org/10.1146/annurev-conmatphys-040721-015537} {\bibfield
  {journal} {\bibinfo  {journal} {Annu. Rev. Condens. Matter Phys.}\ }\textbf
  {\bibinfo {volume} {14}},\ \bibinfo {pages} {35} (\bibinfo {year}
  {2023})}\BibitemShut {NoStop}%
\bibitem [{\citenamefont {Verstraete}\ \emph {et~al.}(2009)\citenamefont
  {Verstraete}, \citenamefont {Wolf},\ and\ \citenamefont
  {Cirac}}]{VerstraeteWolfCirac2009}%
  \BibitemOpen
  \bibfield  {author} {\bibinfo {author} {\bibfnamefont {F.}~\bibnamefont
  {Verstraete}}, \bibinfo {author} {\bibfnamefont {M.~M.}\ \bibnamefont
  {Wolf}},\ and\ \bibinfo {author} {\bibfnamefont {J.~I.}\ \bibnamefont
  {Cirac}},\ }\bibfield  {title} {\bibinfo {title} {Quantum computation and
  quantum-state engineering driven by dissipation},\ }\href
  {https://doi.org/10.1038/nphys1342} {\bibfield  {journal} {\bibinfo
  {journal} {Nat. Phys.}\ }\textbf {\bibinfo {volume} {5}},\ \bibinfo {pages}
  {633} (\bibinfo {year} {2009})}\BibitemShut {NoStop}%
\bibitem [{\citenamefont {Wendin}(2017)}]{Wendin2017}%
  \BibitemOpen
  \bibfield  {author} {\bibinfo {author} {\bibfnamefont {G.}~\bibnamefont
  {Wendin}},\ }\bibfield  {title} {\bibinfo {title} {Quantum information
  processing with superconducting circuits: a review},\ }\href
  {https://doi.org/10.1088/1361-6633/aa7e1a} {\bibfield  {journal} {\bibinfo
  {journal} {Rep. Prog. Phys.}\ }\textbf {\bibinfo {volume} {80}},\ \bibinfo
  {pages} {106001} (\bibinfo {year} {2017})}\BibitemShut {NoStop}%
\bibitem [{\citenamefont {Barreiro}\ \emph {et~al.}(2011)\citenamefont
  {Barreiro}, \citenamefont {M{\"u}ller}, \citenamefont {Schindler},
  \citenamefont {Nigg}, \citenamefont {Monz}, \citenamefont {Chwalla},
  \citenamefont {Hennrich}, \citenamefont {Roos}, \citenamefont {Zoller},\ and\
  \citenamefont {Blatt}}]{Barreiro2011}%
  \BibitemOpen
  \bibfield  {author} {\bibinfo {author} {\bibfnamefont {J.~T.}\ \bibnamefont
  {Barreiro}}, \bibinfo {author} {\bibfnamefont {M.}~\bibnamefont
  {M{\"u}ller}}, \bibinfo {author} {\bibfnamefont {P.}~\bibnamefont
  {Schindler}}, \bibinfo {author} {\bibfnamefont {D.}~\bibnamefont {Nigg}},
  \bibinfo {author} {\bibfnamefont {T.}~\bibnamefont {Monz}}, \bibinfo {author}
  {\bibfnamefont {M.}~\bibnamefont {Chwalla}}, \bibinfo {author} {\bibfnamefont
  {M.}~\bibnamefont {Hennrich}}, \bibinfo {author} {\bibfnamefont {C.~F.}\
  \bibnamefont {Roos}}, \bibinfo {author} {\bibfnamefont {P.}~\bibnamefont
  {Zoller}},\ and\ \bibinfo {author} {\bibfnamefont {R.}~\bibnamefont
  {Blatt}},\ }\bibfield  {title} {\bibinfo {title} {An open-system quantum
  simulator with trapped ions},\ }\href {https://doi.org/10.1038/nature09801}
  {\bibfield  {journal} {\bibinfo  {journal} {Nature}\ }\textbf {\bibinfo
  {volume} {470}},\ \bibinfo {pages} {486} (\bibinfo {year}
  {2011})}\BibitemShut {NoStop}%
\bibitem [{\citenamefont {Ke{\ss}ler}\ \emph {et~al.}(2021)\citenamefont
  {Ke{\ss}ler}, \citenamefont {Kongkhambut}, \citenamefont {Georges},
  \citenamefont {Mathey}, \citenamefont {Cosme},\ and\ \citenamefont
  {Hemmerich}}]{Kessler2021}%
  \BibitemOpen
  \bibfield  {author} {\bibinfo {author} {\bibfnamefont {H.}~\bibnamefont
  {Ke{\ss}ler}}, \bibinfo {author} {\bibfnamefont {P.}~\bibnamefont
  {Kongkhambut}}, \bibinfo {author} {\bibfnamefont {C.}~\bibnamefont
  {Georges}}, \bibinfo {author} {\bibfnamefont {L.}~\bibnamefont {Mathey}},
  \bibinfo {author} {\bibfnamefont {J.~G.}\ \bibnamefont {Cosme}},\ and\
  \bibinfo {author} {\bibfnamefont {A.}~\bibnamefont {Hemmerich}},\ }\bibfield
  {title} {\bibinfo {title} {Observation of a dissipative time crystal},\
  }\href {https://doi.org/10.1103/PhysRevLett.127.043602} {\bibfield  {journal}
  {\bibinfo  {journal} {Phys. Rev. Lett.}\ }\textbf {\bibinfo {volume} {127}},\
  \bibinfo {pages} {043602} (\bibinfo {year} {2021})}\BibitemShut {NoStop}%
\bibitem [{\citenamefont {Chen}\ \emph {et~al.}(2025)\citenamefont {Chen},
  \citenamefont {Abbasi}, \citenamefont {Erdamar}, \citenamefont {Muldoon},
  \citenamefont {Joglekar},\ and\ \citenamefont {Murch}}]{Chen2025}%
  \BibitemOpen
  \bibfield  {author} {\bibinfo {author} {\bibfnamefont {W.}~\bibnamefont
  {Chen}}, \bibinfo {author} {\bibfnamefont {M.}~\bibnamefont {Abbasi}},
  \bibinfo {author} {\bibfnamefont {S.}~\bibnamefont {Erdamar}}, \bibinfo
  {author} {\bibfnamefont {J.}~\bibnamefont {Muldoon}}, \bibinfo {author}
  {\bibfnamefont {Y.~N.}\ \bibnamefont {Joglekar}},\ and\ \bibinfo {author}
  {\bibfnamefont {K.~W.}\ \bibnamefont {Murch}},\ }\bibfield  {title} {\bibinfo
  {title} {Engineering nonequilibrium steady states through {Floquet
  Liouvillians}},\ }\href {https://doi.org/10.1103/PhysRevLett.134.090402}
  {\bibfield  {journal} {\bibinfo  {journal} {Phys. Rev. Lett.}\ }\textbf
  {\bibinfo {volume} {134}},\ \bibinfo {pages} {090402} (\bibinfo {year}
  {2025})}\BibitemShut {NoStop}%
\bibitem [{\citenamefont {Zhang}\ \emph {et~al.}(2024)\citenamefont {Zhang},
  \citenamefont {Zhang}, \citenamefont {Xu}, \citenamefont {Hu}, \citenamefont
  {Bao},\ and\ \citenamefont {Shen}}]{Zhang2024}%
  \BibitemOpen
  \bibfield  {author} {\bibinfo {author} {\bibfnamefont {Z.}~\bibnamefont
  {Zhang}}, \bibinfo {author} {\bibfnamefont {F.}~\bibnamefont {Zhang}},
  \bibinfo {author} {\bibfnamefont {Z.}~\bibnamefont {Xu}}, \bibinfo {author}
  {\bibfnamefont {Y.}~\bibnamefont {Hu}}, \bibinfo {author} {\bibfnamefont
  {H.}~\bibnamefont {Bao}},\ and\ \bibinfo {author} {\bibfnamefont
  {H.}~\bibnamefont {Shen}},\ }\bibfield  {title} {\bibinfo {title} {Realizing
  exceptional points by {Floquet} dissipative couplings in thermal atoms},\
  }\href {https://doi.org/10.1103/PhysRevLett.133.133601} {\bibfield  {journal}
  {\bibinfo  {journal} {Phys. Rev. Lett.}\ }\textbf {\bibinfo {volume} {133}},\
  \bibinfo {pages} {133601} (\bibinfo {year} {2024})}\BibitemShut {NoStop}%
\bibitem [{\citenamefont {Bloch}\ \emph {et~al.}(2012)\citenamefont {Bloch},
  \citenamefont {Dalibard},\ and\ \citenamefont {Nascimb{\`e}ne}}]{Bloch2012}%
  \BibitemOpen
  \bibfield  {author} {\bibinfo {author} {\bibfnamefont {I.}~\bibnamefont
  {Bloch}}, \bibinfo {author} {\bibfnamefont {J.}~\bibnamefont {Dalibard}},\
  and\ \bibinfo {author} {\bibfnamefont {S.}~\bibnamefont {Nascimb{\`e}ne}},\
  }\bibfield  {title} {\bibinfo {title} {Quantum simulations with ultracold
  quantum gases},\ }\href {https://doi.org/10.1038/nphys2259} {\bibfield
  {journal} {\bibinfo  {journal} {Nat. Phys.}\ }\textbf {\bibinfo {volume}
  {8}},\ \bibinfo {pages} {267} (\bibinfo {year} {2012})}\BibitemShut {NoStop}%
\bibitem [{\citenamefont {Mi}\ \emph {et~al.}(2022)\citenamefont {Mi},
  \citenamefont {Ippoliti}, \citenamefont {Quintana}, \citenamefont {Greene},
  \citenamefont {Chen}, \citenamefont {Gross}, \citenamefont {Arute},
  \citenamefont {Arya}, \citenamefont {Atalaya}, \citenamefont {Babbush},
  \citenamefont {Moessner}, \citenamefont {Kechedzhi}, \citenamefont
  {Khemani},\ and\ \citenamefont {Roushan}}]{Mi2022}%
  \BibitemOpen
  \bibfield  {author} {\bibinfo {author} {\bibfnamefont {X.}~\bibnamefont
  {Mi}}, \bibinfo {author} {\bibfnamefont {M.}~\bibnamefont {Ippoliti}},
  \bibinfo {author} {\bibfnamefont {C.}~\bibnamefont {Quintana}}, \bibinfo
  {author} {\bibfnamefont {A.}~\bibnamefont {Greene}}, \bibinfo {author}
  {\bibfnamefont {Z.}~\bibnamefont {Chen}}, \bibinfo {author} {\bibfnamefont
  {J.}~\bibnamefont {Gross}}, \bibinfo {author} {\bibfnamefont
  {F.}~\bibnamefont {Arute}}, \bibinfo {author} {\bibfnamefont
  {K.}~\bibnamefont {Arya}}, \bibinfo {author} {\bibfnamefont {J.}~\bibnamefont
  {Atalaya}}, \bibinfo {author} {\bibfnamefont {R.}~\bibnamefont {Babbush}},
  \bibinfo {author} {\bibfnamefont {R.}~\bibnamefont {Moessner}}, \bibinfo
  {author} {\bibfnamefont {K.}~\bibnamefont {Kechedzhi}}, \bibinfo {author}
  {\bibfnamefont {V.}~\bibnamefont {Khemani}},\ and\ \bibinfo {author}
  {\bibfnamefont {P.}~\bibnamefont {Roushan}},\ }\bibfield  {title} {\bibinfo
  {title} {Time-crystalline eigenstate order on a quantum processor},\ }\href
  {https://doi.org/10.1038/s41586-021-04257-w} {\bibfield  {journal} {\bibinfo
  {journal} {Nature}\ }\textbf {\bibinfo {volume} {601}},\ \bibinfo {pages}
  {531} (\bibinfo {year} {2022})},\ \bibinfo {note} {author list abridged; see
  the published record for the full Google Quantum AI
  collaboration}\BibitemShut {NoStop}%
\bibitem [{\citenamefont {Jiang}\ \emph {et~al.}(2018)\citenamefont {Jiang},
  \citenamefont {Sung}, \citenamefont {Kechedzhi}, \citenamefont
  {Smelyanskiy},\ and\ \citenamefont {Boixo}}]{Jiang2018}%
  \BibitemOpen
  \bibfield  {author} {\bibinfo {author} {\bibfnamefont {Z.}~\bibnamefont
  {Jiang}}, \bibinfo {author} {\bibfnamefont {K.~J.}\ \bibnamefont {Sung}},
  \bibinfo {author} {\bibfnamefont {K.}~\bibnamefont {Kechedzhi}}, \bibinfo
  {author} {\bibfnamefont {V.~N.}\ \bibnamefont {Smelyanskiy}},\ and\ \bibinfo
  {author} {\bibfnamefont {S.}~\bibnamefont {Boixo}},\ }\bibfield  {title}
  {\bibinfo {title} {Quantum algorithms to simulate many-body physics of
  correlated fermions},\ }\href
  {https://doi.org/10.1103/PhysRevApplied.9.044036} {\bibfield  {journal}
  {\bibinfo  {journal} {Phys. Rev. Applied}\ }\textbf {\bibinfo {volume} {9}},\
  \bibinfo {pages} {044036} (\bibinfo {year} {2018})}\BibitemShut {NoStop}%
\bibitem [{\citenamefont {Minganti}\ \emph {et~al.}(2018)\citenamefont
  {Minganti}, \citenamefont {Biella}, \citenamefont {Bartolo},\ and\
  \citenamefont {Ciuti}}]{Minganti2018}%
  \BibitemOpen
  \bibfield  {author} {\bibinfo {author} {\bibfnamefont {F.}~\bibnamefont
  {Minganti}}, \bibinfo {author} {\bibfnamefont {A.}~\bibnamefont {Biella}},
  \bibinfo {author} {\bibfnamefont {N.}~\bibnamefont {Bartolo}},\ and\ \bibinfo
  {author} {\bibfnamefont {C.}~\bibnamefont {Ciuti}},\ }\bibfield  {title}
  {\bibinfo {title} {Spectral theory of {Liouvillians} for dissipative phase
  transitions},\ }\href {https://doi.org/10.1103/PhysRevA.98.042118} {\bibfield
   {journal} {\bibinfo  {journal} {Phys. Rev. A}\ }\textbf {\bibinfo {volume}
  {98}},\ \bibinfo {pages} {042118} (\bibinfo {year} {2018})}\BibitemShut
  {NoStop}%
\bibitem [{\citenamefont {Prosen}(2008)}]{Prosen2008}%
  \BibitemOpen
  \bibfield  {author} {\bibinfo {author} {\bibfnamefont {T.}~\bibnamefont
  {Prosen}},\ }\bibfield  {title} {\bibinfo {title} {Third quantization: a
  general method to solve master equations for quadratic open {Fermi}
  systems},\ }\href {https://doi.org/10.1088/1367-2630/10/4/043026} {\bibfield
  {journal} {\bibinfo  {journal} {New J. Phys.}\ }\textbf {\bibinfo {volume}
  {10}},\ \bibinfo {pages} {043026} (\bibinfo {year} {2008})}\BibitemShut
  {NoStop}%
\bibitem [{\citenamefont {Koch}\ \emph {et~al.}(2026)\citenamefont {Koch},
  \citenamefont {Hu},\ and\ \citenamefont {Budich}}]{Koch2026}%
  \BibitemOpen
  \bibfield  {author} {\bibinfo {author} {\bibfnamefont {F.}~\bibnamefont
  {Koch}}, \bibinfo {author} {\bibfnamefont {Y.-M.}\ \bibnamefont {Hu}},\ and\
  \bibinfo {author} {\bibfnamefont {J.~C.}\ \bibnamefont {Budich}},\ }\bibfield
   {title} {\bibinfo {title} {Liouvillian topology and nonreciprocal dynamics
  in open {Floquet} chains},\ }\href {https://doi.org/10.1103/284g-926d}
  {\bibfield  {journal} {\bibinfo  {journal} {Phys. Rev. Research}\ }\textbf
  {\bibinfo {volume} {8}},\ \bibinfo {pages} {023270} (\bibinfo {year}
  {2026})}\BibitemShut {NoStop}%
\bibitem [{\citenamefont {Mori}\ and\ \citenamefont
  {Shirai}(2020)}]{MoriShirai2020}%
  \BibitemOpen
  \bibfield  {author} {\bibinfo {author} {\bibfnamefont {T.}~\bibnamefont
  {Mori}}\ and\ \bibinfo {author} {\bibfnamefont {T.}~\bibnamefont {Shirai}},\
  }\bibfield  {title} {\bibinfo {title} {Resolving a discrepancy between
  {Liouvillian} gap and relaxation time in boundary-dissipated quantum
  many-body systems},\ }\href {https://doi.org/10.1103/PhysRevLett.125.230604}
  {\bibfield  {journal} {\bibinfo  {journal} {Phys. Rev. Lett.}\ }\textbf
  {\bibinfo {volume} {125}},\ \bibinfo {pages} {230604} (\bibinfo {year}
  {2020})}\BibitemShut {NoStop}%
\bibitem [{\citenamefont {Haga}\ \emph {et~al.}(2021)\citenamefont {Haga},
  \citenamefont {Nakagawa}, \citenamefont {Hamazaki},\ and\ \citenamefont
  {Ueda}}]{Haga2021}%
  \BibitemOpen
  \bibfield  {author} {\bibinfo {author} {\bibfnamefont {T.}~\bibnamefont
  {Haga}}, \bibinfo {author} {\bibfnamefont {M.}~\bibnamefont {Nakagawa}},
  \bibinfo {author} {\bibfnamefont {R.}~\bibnamefont {Hamazaki}},\ and\
  \bibinfo {author} {\bibfnamefont {M.}~\bibnamefont {Ueda}},\ }\bibfield
  {title} {\bibinfo {title} {Liouvillian skin effect: Slowing down of
  relaxation processes without gap closing},\ }\href
  {https://doi.org/10.1103/PhysRevLett.127.070402} {\bibfield  {journal}
  {\bibinfo  {journal} {Phys. Rev. Lett.}\ }\textbf {\bibinfo {volume} {127}},\
  \bibinfo {pages} {070402} (\bibinfo {year} {2021})}\BibitemShut {NoStop}%
\bibitem [{\citenamefont {Ughrelidze}\ \emph {et~al.}(2024)\citenamefont
  {Ughrelidze}, \citenamefont {Flynn}, \citenamefont {Cobanera},\ and\
  \citenamefont {Viola}}]{Ughrelidze2024}%
  \BibitemOpen
  \bibfield  {author} {\bibinfo {author} {\bibfnamefont {M.}~\bibnamefont
  {Ughrelidze}}, \bibinfo {author} {\bibfnamefont {V.~P.}\ \bibnamefont
  {Flynn}}, \bibinfo {author} {\bibfnamefont {E.}~\bibnamefont {Cobanera}},\
  and\ \bibinfo {author} {\bibfnamefont {L.}~\bibnamefont {Viola}},\ }\bibfield
   {title} {\bibinfo {title} {Interplay of finite- and infinite-size stability
  in quadratic bosonic {Lindbladians}},\ }\href
  {https://doi.org/10.1103/PhysRevA.110.032207} {\bibfield  {journal} {\bibinfo
   {journal} {Phys. Rev. A}\ }\textbf {\bibinfo {volume} {110}},\ \bibinfo
  {pages} {032207} (\bibinfo {year} {2024})}\BibitemShut {NoStop}%
\bibitem [{\citenamefont {Trefethen}\ and\ \citenamefont
  {Embree}(2005)}]{TrefethenEmbree2005}%
  \BibitemOpen
  \bibfield  {author} {\bibinfo {author} {\bibfnamefont {L.~N.}\ \bibnamefont
  {Trefethen}}\ and\ \bibinfo {author} {\bibfnamefont {M.}~\bibnamefont
  {Embree}},\ }\href {https://doi.org/10.1515/9780691213101} {\emph {\bibinfo
  {title} {Spectra and Pseudospectra: The Behavior of Nonnormal Matrices and
  Operators}}}\ (\bibinfo  {publisher} {Princeton University Press},\ \bibinfo
  {address} {Princeton, NJ},\ \bibinfo {year} {2005})\BibitemShut {NoStop}%
\bibitem [{\citenamefont {Trefethen}(1997)}]{Trefethen1997}%
  \BibitemOpen
  \bibfield  {author} {\bibinfo {author} {\bibfnamefont {L.~N.}\ \bibnamefont
  {Trefethen}},\ }\bibfield  {title} {\bibinfo {title} {Pseudospectra of linear
  operators},\ }\href {https://doi.org/10.1137/S0036144595295284} {\bibfield
  {journal} {\bibinfo  {journal} {SIAM Rev.}\ }\textbf {\bibinfo {volume}
  {39}},\ \bibinfo {pages} {383} (\bibinfo {year} {1997})}\BibitemShut
  {NoStop}%
\bibitem [{\citenamefont {Dencker}\ \emph {et~al.}(2004)\citenamefont
  {Dencker}, \citenamefont {Sj{\"o}strand},\ and\ \citenamefont
  {Zworski}}]{Dencker2004}%
  \BibitemOpen
  \bibfield  {author} {\bibinfo {author} {\bibfnamefont {N.}~\bibnamefont
  {Dencker}}, \bibinfo {author} {\bibfnamefont {J.}~\bibnamefont
  {Sj{\"o}strand}},\ and\ \bibinfo {author} {\bibfnamefont {M.}~\bibnamefont
  {Zworski}},\ }\bibfield  {title} {\bibinfo {title} {Pseudospectra of
  semiclassical (pseudo-) differential operators},\ }\href
  {https://doi.org/10.1002/cpa.20004} {\bibfield  {journal} {\bibinfo
  {journal} {Commun. Pure Appl. Math.}\ }\textbf {\bibinfo {volume} {57}},\
  \bibinfo {pages} {384} (\bibinfo {year} {2004})}\BibitemShut {NoStop}%
\bibitem [{\citenamefont {Davies}(2003)}]{Davies2003}%
  \BibitemOpen
  \bibfield  {author} {\bibinfo {author} {\bibfnamefont {E.~B.}\ \bibnamefont
  {Davies}},\ }\href {https://doi.org/10.48550/arXiv.math/0303235} {\bibinfo
  {title} {Approximating semigroups by using pseudospectra}} (\bibinfo {year}
  {2003}),\ \Eprint {https://arxiv.org/abs/math/0303235} {arXiv:math/0303235
  [math.SP]} \BibitemShut {NoStop}%
\bibitem [{\citenamefont {Kiorpelidis}\ and\ \citenamefont
  {Makris}(2025)}]{Kiorpelidis2025}%
  \BibitemOpen
  \bibfield  {author} {\bibinfo {author} {\bibfnamefont {I.}~\bibnamefont
  {Kiorpelidis}}\ and\ \bibinfo {author} {\bibfnamefont {K.~G.}\ \bibnamefont
  {Makris}},\ }\bibfield  {title} {\bibinfo {title} {Scaling of pseudospectra
  in exponentially sensitive lattices},\ }\href
  {https://doi.org/10.1103/PhysRevResearch.7.L032043} {\bibfield  {journal}
  {\bibinfo  {journal} {Phys. Rev. Research}\ }\textbf {\bibinfo {volume}
  {7}},\ \bibinfo {pages} {L032043} (\bibinfo {year} {2025})}\BibitemShut
  {NoStop}%
\bibitem [{\citenamefont {Metelmann}\ and\ \citenamefont
  {Clerk}(2015)}]{MetelmannClerk2015}%
  \BibitemOpen
  \bibfield  {author} {\bibinfo {author} {\bibfnamefont {A.}~\bibnamefont
  {Metelmann}}\ and\ \bibinfo {author} {\bibfnamefont {A.~A.}\ \bibnamefont
  {Clerk}},\ }\bibfield  {title} {\bibinfo {title} {Nonreciprocal photon
  transmission and amplification via reservoir engineering},\ }\href
  {https://doi.org/10.1103/PhysRevX.5.021025} {\bibfield  {journal} {\bibinfo
  {journal} {Phys. Rev. X}\ }\textbf {\bibinfo {volume} {5}},\ \bibinfo {pages}
  {021025} (\bibinfo {year} {2015})}\BibitemShut {NoStop}%
\bibitem [{\citenamefont {Yang}\ \emph {et~al.}(2022)\citenamefont {Yang},
  \citenamefont {Jiang},\ and\ \citenamefont {Bergholtz}}]{Yang2022}%
  \BibitemOpen
  \bibfield  {author} {\bibinfo {author} {\bibfnamefont {F.}~\bibnamefont
  {Yang}}, \bibinfo {author} {\bibfnamefont {Q.-D.}\ \bibnamefont {Jiang}},\
  and\ \bibinfo {author} {\bibfnamefont {E.~J.}\ \bibnamefont {Bergholtz}},\
  }\bibfield  {title} {\bibinfo {title} {Liouvillian skin effect in an exactly
  solvable model},\ }\href {https://doi.org/10.1103/PhysRevResearch.4.023160}
  {\bibfield  {journal} {\bibinfo  {journal} {Phys. Rev. Research}\ }\textbf
  {\bibinfo {volume} {4}},\ \bibinfo {pages} {023160} (\bibinfo {year}
  {2022})}\BibitemShut {NoStop}%
\bibitem [{\citenamefont {McDonald}\ \emph {et~al.}(2018)\citenamefont
  {McDonald}, \citenamefont {Pereg-Barnea},\ and\ \citenamefont
  {Clerk}}]{McDonald2018}%
  \BibitemOpen
  \bibfield  {author} {\bibinfo {author} {\bibfnamefont {A.}~\bibnamefont
  {McDonald}}, \bibinfo {author} {\bibfnamefont {T.}~\bibnamefont
  {Pereg-Barnea}},\ and\ \bibinfo {author} {\bibfnamefont {A.~A.}\ \bibnamefont
  {Clerk}},\ }\bibfield  {title} {\bibinfo {title} {Phase-dependent chiral
  transport and effective non-{Hermitian} dynamics in a bosonic
  {Kitaev-Majorana} chain},\ }\href {https://doi.org/10.1103/PhysRevX.8.041031}
  {\bibfield  {journal} {\bibinfo  {journal} {Phys. Rev. X}\ }\textbf {\bibinfo
  {volume} {8}},\ \bibinfo {pages} {041031} (\bibinfo {year}
  {2018})}\BibitemShut {NoStop}%
\bibitem [{\citenamefont {Wanjura}\ \emph {et~al.}(2020)\citenamefont
  {Wanjura}, \citenamefont {Brunelli},\ and\ \citenamefont
  {Nunnenkamp}}]{Wanjura2020}%
  \BibitemOpen
  \bibfield  {author} {\bibinfo {author} {\bibfnamefont {C.~C.}\ \bibnamefont
  {Wanjura}}, \bibinfo {author} {\bibfnamefont {M.}~\bibnamefont {Brunelli}},\
  and\ \bibinfo {author} {\bibfnamefont {A.}~\bibnamefont {Nunnenkamp}},\
  }\bibfield  {title} {\bibinfo {title} {Topological framework for directional
  amplification in driven-dissipative cavity arrays},\ }\href
  {https://doi.org/10.1038/s41467-020-16863-9} {\bibfield  {journal} {\bibinfo
  {journal} {Nat. Commun.}\ }\textbf {\bibinfo {volume} {11}},\ \bibinfo
  {pages} {3149} (\bibinfo {year} {2020})}\BibitemShut {NoStop}%
\bibitem [{\citenamefont {Pollock}\ \emph {et~al.}(2018)\citenamefont
  {Pollock}, \citenamefont {Rodr{\'i}guez-Rosario}, \citenamefont {Frauenheim},
  \citenamefont {Paternostro},\ and\ \citenamefont {Modi}}]{Pollock2018}%
  \BibitemOpen
  \bibfield  {author} {\bibinfo {author} {\bibfnamefont {F.~A.}\ \bibnamefont
  {Pollock}}, \bibinfo {author} {\bibfnamefont {C.}~\bibnamefont
  {Rodr{\'i}guez-Rosario}}, \bibinfo {author} {\bibfnamefont {T.}~\bibnamefont
  {Frauenheim}}, \bibinfo {author} {\bibfnamefont {M.}~\bibnamefont
  {Paternostro}},\ and\ \bibinfo {author} {\bibfnamefont {K.}~\bibnamefont
  {Modi}},\ }\bibfield  {title} {\bibinfo {title} {Operational {Markov}
  condition for quantum processes},\ }\href
  {https://doi.org/10.1103/PhysRevLett.120.040405} {\bibfield  {journal}
  {\bibinfo  {journal} {Phys. Rev. Lett.}\ }\textbf {\bibinfo {volume} {120}},\
  \bibinfo {pages} {040405} (\bibinfo {year} {2018})}\BibitemShut {NoStop}%
\bibitem [{\citenamefont {Schnell}\ \emph {et~al.}(2020)\citenamefont
  {Schnell}, \citenamefont {Eckardt},\ and\ \citenamefont
  {Denisov}}]{Schnell2020}%
  \BibitemOpen
  \bibfield  {author} {\bibinfo {author} {\bibfnamefont {A.}~\bibnamefont
  {Schnell}}, \bibinfo {author} {\bibfnamefont {A.}~\bibnamefont {Eckardt}},\
  and\ \bibinfo {author} {\bibfnamefont {S.}~\bibnamefont {Denisov}},\
  }\bibfield  {title} {\bibinfo {title} {Is there a {Floquet Lindbladian}?},\
  }\href {https://doi.org/10.1103/PhysRevB.101.100301} {\bibfield  {journal}
  {\bibinfo  {journal} {Phys. Rev. B}\ }\textbf {\bibinfo {volume} {101}},\
  \bibinfo {pages} {100301(R)} (\bibinfo {year} {2020})}\BibitemShut {NoStop}%
\bibitem [{\citenamefont {Bingham}\ \emph {et~al.}(1987)\citenamefont
  {Bingham}, \citenamefont {Goldie},\ and\ \citenamefont
  {Teugels}}]{Bingham1987}%
  \BibitemOpen
  \bibfield  {author} {\bibinfo {author} {\bibfnamefont {N.~H.}\ \bibnamefont
  {Bingham}}, \bibinfo {author} {\bibfnamefont {C.~M.}\ \bibnamefont
  {Goldie}},\ and\ \bibinfo {author} {\bibfnamefont {J.~L.}\ \bibnamefont
  {Teugels}},\ }\href {https://doi.org/10.1017/CBO9780511721434} {\emph
  {\bibinfo {title} {Regular Variation}}},\ \bibinfo {series} {Encyclopedia of
  Mathematics and its Applications}, Vol.~\bibinfo {volume} {27}\ (\bibinfo
  {publisher} {Cambridge University Press},\ \bibinfo {address} {Cambridge},\
  \bibinfo {year} {1987})\BibitemShut {NoStop}%
\bibitem [{\citenamefont {Sz.-Nagy}\ \emph {et~al.}(2010)\citenamefont
  {Sz.-Nagy}, \citenamefont {Foias}, \citenamefont {Bercovici},\ and\
  \citenamefont {K{\'e}rchy}}]{SzNagy2010}%
  \BibitemOpen
  \bibfield  {author} {\bibinfo {author} {\bibfnamefont {B.}~\bibnamefont
  {Sz.-Nagy}}, \bibinfo {author} {\bibfnamefont {C.}~\bibnamefont {Foias}},
  \bibinfo {author} {\bibfnamefont {H.}~\bibnamefont {Bercovici}},\ and\
  \bibinfo {author} {\bibfnamefont {L.}~\bibnamefont {K{\'e}rchy}},\ }\href
  {https://doi.org/10.1007/978-1-4419-6094-8} {\emph {\bibinfo {title}
  {Harmonic Analysis of Operators on Hilbert Space}}},\ \bibinfo {edition}
  {2nd}\ ed.\ (\bibinfo  {publisher} {Springer},\ \bibinfo {address} {New
  York},\ \bibinfo {year} {2010})\BibitemShut {NoStop}%
\bibitem [{\citenamefont {Reichel}\ and\ \citenamefont
  {Trefethen}(1992)}]{ReichelTrefethen1992}%
  \BibitemOpen
  \bibfield  {author} {\bibinfo {author} {\bibfnamefont {L.}~\bibnamefont
  {Reichel}}\ and\ \bibinfo {author} {\bibfnamefont {L.~N.}\ \bibnamefont
  {Trefethen}},\ }\bibfield  {title} {\bibinfo {title} {Eigenvalues and
  pseudo-eigenvalues of {Toeplitz} matrices},\ }\href
  {https://doi.org/10.1016/0024-3795(92)90374-J} {\bibfield  {journal}
  {\bibinfo  {journal} {Linear Algebra Appl.}\ }\textbf {\bibinfo {volume}
  {162--164}},\ \bibinfo {pages} {153} (\bibinfo {year} {1992})}\BibitemShut
  {NoStop}%
\bibitem [{\citenamefont {Trefethen}\ \emph {et~al.}(2001)\citenamefont
  {Trefethen}, \citenamefont {Contedini},\ and\ \citenamefont
  {Embree}}]{TrefethenContediniEmbree2001}%
  \BibitemOpen
  \bibfield  {author} {\bibinfo {author} {\bibfnamefont {L.~N.}\ \bibnamefont
  {Trefethen}}, \bibinfo {author} {\bibfnamefont {M.}~\bibnamefont
  {Contedini}},\ and\ \bibinfo {author} {\bibfnamefont {M.}~\bibnamefont
  {Embree}},\ }\bibfield  {title} {\bibinfo {title} {Spectra, pseudospectra,
  and localization for random bidiagonal matrices},\ }\href
  {https://doi.org/10.1002/cpa.4} {\bibfield  {journal} {\bibinfo  {journal}
  {Commun. Pure Appl. Math.}\ }\textbf {\bibinfo {volume} {54}},\ \bibinfo
  {pages} {595} (\bibinfo {year} {2001})}\BibitemShut {NoStop}%
\bibitem [{\citenamefont {Macieszczak}\ \emph {et~al.}(2016)\citenamefont
  {Macieszczak}, \citenamefont {Gu\c{t}\u{a}}, \citenamefont {Lesanovsky},\
  and\ \citenamefont {Garrahan}}]{Macieszczak2016}%
  \BibitemOpen
  \bibfield  {author} {\bibinfo {author} {\bibfnamefont {K.}~\bibnamefont
  {Macieszczak}}, \bibinfo {author} {\bibfnamefont {M.}~\bibnamefont
  {Gu\c{t}\u{a}}}, \bibinfo {author} {\bibfnamefont {I.}~\bibnamefont
  {Lesanovsky}},\ and\ \bibinfo {author} {\bibfnamefont {J.~P.}\ \bibnamefont
  {Garrahan}},\ }\bibfield  {title} {\bibinfo {title} {Towards a theory of
  metastability in open quantum dynamics},\ }\href
  {https://doi.org/10.1103/PhysRevLett.116.240404} {\bibfield  {journal}
  {\bibinfo  {journal} {Phys. Rev. Lett.}\ }\textbf {\bibinfo {volume} {116}},\
  \bibinfo {pages} {240404} (\bibinfo {year} {2016})}\BibitemShut {NoStop}%
\bibitem [{\citenamefont {Kato}(1995)}]{Kato1995}%
  \BibitemOpen
  \bibfield  {author} {\bibinfo {author} {\bibfnamefont {T.}~\bibnamefont
  {Kato}},\ }\href {https://doi.org/10.1007/978-3-642-66282-9} {\emph {\bibinfo
  {title} {Perturbation Theory for Linear Operators}}},\ Classics in
  Mathematics\ (\bibinfo  {publisher} {Springer},\ \bibinfo {address}
  {Berlin},\ \bibinfo {year} {1995})\BibitemShut {NoStop}%
\bibitem [{\citenamefont {Stewart}(1973)}]{Stewart1973}%
  \BibitemOpen
  \bibfield  {author} {\bibinfo {author} {\bibfnamefont {G.~W.}\ \bibnamefont
  {Stewart}},\ }\bibfield  {title} {\bibinfo {title} {Error and perturbation
  bounds for subspaces associated with certain eigenvalue problems},\ }\href
  {https://doi.org/10.1137/1015095} {\bibfield  {journal} {\bibinfo  {journal}
  {SIAM Rev.}\ }\textbf {\bibinfo {volume} {15}},\ \bibinfo {pages} {727}
  (\bibinfo {year} {1973})}\BibitemShut {NoStop}%
\bibitem [{\citenamefont {Bhatia}\ and\ \citenamefont
  {Rosenthal}(1997)}]{BhatiaRosenthal1997}%
  \BibitemOpen
  \bibfield  {author} {\bibinfo {author} {\bibfnamefont {R.}~\bibnamefont
  {Bhatia}}\ and\ \bibinfo {author} {\bibfnamefont {P.}~\bibnamefont
  {Rosenthal}},\ }\bibfield  {title} {\bibinfo {title} {How and why to solve
  the operator equation $ax-xb=y$},\ }\href
  {https://doi.org/10.1112/S0024609396001828} {\bibfield  {journal} {\bibinfo
  {journal} {Bull. London Math. Soc.}\ }\textbf {\bibinfo {volume} {29}},\
  \bibinfo {pages} {1} (\bibinfo {year} {1997})}\BibitemShut {NoStop}%
\bibitem [{\citenamefont {Ciccarello}\ \emph {et~al.}(2022)\citenamefont
  {Ciccarello}, \citenamefont {Lorenzo}, \citenamefont {Giovannetti},\ and\
  \citenamefont {Palma}}]{Ciccarello2022}%
  \BibitemOpen
  \bibfield  {author} {\bibinfo {author} {\bibfnamefont {F.}~\bibnamefont
  {Ciccarello}}, \bibinfo {author} {\bibfnamefont {S.}~\bibnamefont {Lorenzo}},
  \bibinfo {author} {\bibfnamefont {V.}~\bibnamefont {Giovannetti}},\ and\
  \bibinfo {author} {\bibfnamefont {G.~M.}\ \bibnamefont {Palma}},\ }\bibfield
  {title} {\bibinfo {title} {Quantum collision models: Open system dynamics
  from repeated interactions},\ }\href
  {https://doi.org/10.1016/j.physrep.2022.01.001} {\bibfield  {journal}
  {\bibinfo  {journal} {Phys. Rep.}\ }\textbf {\bibinfo {volume} {954}},\
  \bibinfo {pages} {1} (\bibinfo {year} {2022})}\BibitemShut {NoStop}%
\bibitem [{\citenamefont {Herviou}\ \emph {et~al.}(2019)\citenamefont
  {Herviou}, \citenamefont {Bardarson},\ and\ \citenamefont
  {Regnault}}]{Herviou2019}%
  \BibitemOpen
  \bibfield  {author} {\bibinfo {author} {\bibfnamefont {L.}~\bibnamefont
  {Herviou}}, \bibinfo {author} {\bibfnamefont {J.~H.}\ \bibnamefont
  {Bardarson}},\ and\ \bibinfo {author} {\bibfnamefont {N.}~\bibnamefont
  {Regnault}},\ }\bibfield  {title} {\bibinfo {title} {Defining a bulk-edge
  correspondence for non-{Hermitian} {Hamiltonians} via singular-value
  decomposition},\ }\href {https://doi.org/10.1103/PhysRevA.99.052118}
  {\bibfield  {journal} {\bibinfo  {journal} {Phys. Rev. A}\ }\textbf {\bibinfo
  {volume} {99}},\ \bibinfo {pages} {052118} (\bibinfo {year}
  {2019})}\BibitemShut {NoStop}%
\bibitem [{\citenamefont {Su}\ \emph {et~al.}(1979)\citenamefont {Su},
  \citenamefont {Schrieffer},\ and\ \citenamefont {Heeger}}]{SSH1979}%
  \BibitemOpen
  \bibfield  {author} {\bibinfo {author} {\bibfnamefont {W.~P.}\ \bibnamefont
  {Su}}, \bibinfo {author} {\bibfnamefont {J.~R.}\ \bibnamefont {Schrieffer}},\
  and\ \bibinfo {author} {\bibfnamefont {A.~J.}\ \bibnamefont {Heeger}},\
  }\bibfield  {title} {\bibinfo {title} {Solitons in polyacetylene},\ }\href
  {https://doi.org/10.1103/PhysRevLett.42.1698} {\bibfield  {journal} {\bibinfo
   {journal} {Phys. Rev. Lett.}\ }\textbf {\bibinfo {volume} {42}},\ \bibinfo
  {pages} {1698} (\bibinfo {year} {1979})}\BibitemShut {NoStop}%
\bibitem [{\citenamefont {Gong}\ \emph {et~al.}(2018)\citenamefont {Gong},
  \citenamefont {Ashida}, \citenamefont {Kawabata}, \citenamefont {Takasan},
  \citenamefont {Higashikawa},\ and\ \citenamefont {Ueda}}]{Gong2018}%
  \BibitemOpen
  \bibfield  {author} {\bibinfo {author} {\bibfnamefont {Z.}~\bibnamefont
  {Gong}}, \bibinfo {author} {\bibfnamefont {Y.}~\bibnamefont {Ashida}},
  \bibinfo {author} {\bibfnamefont {K.}~\bibnamefont {Kawabata}}, \bibinfo
  {author} {\bibfnamefont {K.}~\bibnamefont {Takasan}}, \bibinfo {author}
  {\bibfnamefont {S.}~\bibnamefont {Higashikawa}},\ and\ \bibinfo {author}
  {\bibfnamefont {M.}~\bibnamefont {Ueda}},\ }\bibfield  {title} {\bibinfo
  {title} {Topological phases of non-{Hermitian} systems},\ }\href
  {https://doi.org/10.1103/PhysRevX.8.031079} {\bibfield  {journal} {\bibinfo
  {journal} {Phys. Rev. X}\ }\textbf {\bibinfo {volume} {8}},\ \bibinfo {pages}
  {031079} (\bibinfo {year} {2018})}\BibitemShut {NoStop}%
\bibitem [{\citenamefont {Yao}\ and\ \citenamefont {Wang}(2018)}]{YaoWang2018}%
  \BibitemOpen
  \bibfield  {author} {\bibinfo {author} {\bibfnamefont {S.}~\bibnamefont
  {Yao}}\ and\ \bibinfo {author} {\bibfnamefont {Z.}~\bibnamefont {Wang}},\
  }\bibfield  {title} {\bibinfo {title} {Edge states and topological invariants
  of non-{Hermitian} systems},\ }\href
  {https://doi.org/10.1103/PhysRevLett.121.086803} {\bibfield  {journal}
  {\bibinfo  {journal} {Phys. Rev. Lett.}\ }\textbf {\bibinfo {volume} {121}},\
  \bibinfo {pages} {086803} (\bibinfo {year} {2018})}\BibitemShut {NoStop}%
\bibitem [{\citenamefont {Okuma}\ \emph {et~al.}(2020)\citenamefont {Okuma},
  \citenamefont {Kawabata}, \citenamefont {Shiozaki},\ and\ \citenamefont
  {Sato}}]{Okuma2020}%
  \BibitemOpen
  \bibfield  {author} {\bibinfo {author} {\bibfnamefont {N.}~\bibnamefont
  {Okuma}}, \bibinfo {author} {\bibfnamefont {K.}~\bibnamefont {Kawabata}},
  \bibinfo {author} {\bibfnamefont {K.}~\bibnamefont {Shiozaki}},\ and\
  \bibinfo {author} {\bibfnamefont {M.}~\bibnamefont {Sato}},\ }\bibfield
  {title} {\bibinfo {title} {Topological origin of non-{Hermitian} skin
  effects},\ }\href {https://doi.org/10.1103/PhysRevLett.124.086801} {\bibfield
   {journal} {\bibinfo  {journal} {Phys. Rev. Lett.}\ }\textbf {\bibinfo
  {volume} {124}},\ \bibinfo {pages} {086801} (\bibinfo {year}
  {2020})}\BibitemShut {NoStop}%
\bibitem [{\citenamefont {Ammari}\ \emph {et~al.}(2025)\citenamefont {Ammari},
  \citenamefont {Barandun}, \citenamefont {De~Bruijn}, \citenamefont {Liu},\
  and\ \citenamefont {Thalhammer}}]{Ammari2025}%
  \BibitemOpen
  \bibfield  {author} {\bibinfo {author} {\bibfnamefont {H.}~\bibnamefont
  {Ammari}}, \bibinfo {author} {\bibfnamefont {S.}~\bibnamefont {Barandun}},
  \bibinfo {author} {\bibfnamefont {Y.}~\bibnamefont {De~Bruijn}}, \bibinfo
  {author} {\bibfnamefont {P.}~\bibnamefont {Liu}},\ and\ \bibinfo {author}
  {\bibfnamefont {C.}~\bibnamefont {Thalhammer}},\ }\bibfield  {title}
  {\bibinfo {title} {Spectra and pseudo-spectra of tridiagonal $k$-{Toeplitz}
  matrices and the topological origin of the non-{Hermitian} skin effect},\
  }\href {https://doi.org/10.1088/1751-8121/add5ab} {\bibfield  {journal}
  {\bibinfo  {journal} {J. Phys. A}\ }\textbf {\bibinfo {volume} {58}},\
  \bibinfo {pages} {205201} (\bibinfo {year} {2025})}\BibitemShut {NoStop}%
\bibitem [{\citenamefont {Kawabata}\ \emph {et~al.}(2019)\citenamefont
  {Kawabata}, \citenamefont {Shiozaki}, \citenamefont {Ueda},\ and\
  \citenamefont {Sato}}]{Kawabata2019}%
  \BibitemOpen
  \bibfield  {author} {\bibinfo {author} {\bibfnamefont {K.}~\bibnamefont
  {Kawabata}}, \bibinfo {author} {\bibfnamefont {K.}~\bibnamefont {Shiozaki}},
  \bibinfo {author} {\bibfnamefont {M.}~\bibnamefont {Ueda}},\ and\ \bibinfo
  {author} {\bibfnamefont {M.}~\bibnamefont {Sato}},\ }\bibfield  {title}
  {\bibinfo {title} {Symmetry and topology in non-{Hermitian} physics},\ }\href
  {https://doi.org/10.1103/PhysRevX.9.041015} {\bibfield  {journal} {\bibinfo
  {journal} {Phys. Rev. X}\ }\textbf {\bibinfo {volume} {9}},\ \bibinfo {pages}
  {041015} (\bibinfo {year} {2019})}\BibitemShut {NoStop}%
\bibitem [{\citenamefont {Ashida}\ \emph {et~al.}(2020)\citenamefont {Ashida},
  \citenamefont {Gong},\ and\ \citenamefont {Ueda}}]{Ashida2020}%
  \BibitemOpen
  \bibfield  {author} {\bibinfo {author} {\bibfnamefont {Y.}~\bibnamefont
  {Ashida}}, \bibinfo {author} {\bibfnamefont {Z.}~\bibnamefont {Gong}},\ and\
  \bibinfo {author} {\bibfnamefont {M.}~\bibnamefont {Ueda}},\ }\bibfield
  {title} {\bibinfo {title} {Non-{Hermitian} physics},\ }\href
  {https://doi.org/10.1080/00018732.2021.1876991} {\bibfield  {journal}
  {\bibinfo  {journal} {Adv. Phys.}\ }\textbf {\bibinfo {volume} {69}},\
  \bibinfo {pages} {249} (\bibinfo {year} {2020})}\BibitemShut {NoStop}%
\bibitem [{\citenamefont {Bergholtz}\ \emph {et~al.}(2021)\citenamefont
  {Bergholtz}, \citenamefont {Budich},\ and\ \citenamefont
  {Kunst}}]{Bergholtz2021}%
  \BibitemOpen
  \bibfield  {author} {\bibinfo {author} {\bibfnamefont {E.~J.}\ \bibnamefont
  {Bergholtz}}, \bibinfo {author} {\bibfnamefont {J.~C.}\ \bibnamefont
  {Budich}},\ and\ \bibinfo {author} {\bibfnamefont {F.~K.}\ \bibnamefont
  {Kunst}},\ }\bibfield  {title} {\bibinfo {title} {Exceptional topology of
  non-{Hermitian} systems},\ }\href
  {https://doi.org/10.1103/RevModPhys.93.015005} {\bibfield  {journal}
  {\bibinfo  {journal} {Rev. Mod. Phys.}\ }\textbf {\bibinfo {volume} {93}},\
  \bibinfo {pages} {015005} (\bibinfo {year} {2021})}\BibitemShut {NoStop}%
\bibitem [{\citenamefont {Okuma}\ and\ \citenamefont
  {Sato}(2023)}]{OkumaSato2023}%
  \BibitemOpen
  \bibfield  {author} {\bibinfo {author} {\bibfnamefont {N.}~\bibnamefont
  {Okuma}}\ and\ \bibinfo {author} {\bibfnamefont {M.}~\bibnamefont {Sato}},\
  }\bibfield  {title} {\bibinfo {title} {Non-{Hermitian} topological phenomena:
  A review},\ }\href {https://doi.org/10.1146/annurev-conmatphys-040521-033133}
  {\bibfield  {journal} {\bibinfo  {journal} {Annu. Rev. Condens. Matter
  Phys.}\ }\textbf {\bibinfo {volume} {14}},\ \bibinfo {pages} {83} (\bibinfo
  {year} {2023})}\BibitemShut {NoStop}%
\bibitem [{\citenamefont {Zhang}\ \emph {et~al.}(2020)\citenamefont {Zhang},
  \citenamefont {Yang},\ and\ \citenamefont {Fang}}]{ZhangYangFang2020}%
  \BibitemOpen
  \bibfield  {author} {\bibinfo {author} {\bibfnamefont {K.}~\bibnamefont
  {Zhang}}, \bibinfo {author} {\bibfnamefont {Z.}~\bibnamefont {Yang}},\ and\
  \bibinfo {author} {\bibfnamefont {C.}~\bibnamefont {Fang}},\ }\bibfield
  {title} {\bibinfo {title} {Correspondence between winding numbers and skin
  modes in non-{Hermitian} systems},\ }\href
  {https://doi.org/10.1103/PhysRevLett.125.126402} {\bibfield  {journal}
  {\bibinfo  {journal} {Phys. Rev. Lett.}\ }\textbf {\bibinfo {volume} {125}},\
  \bibinfo {pages} {126402} (\bibinfo {year} {2020})}\BibitemShut {NoStop}%
\bibitem [{\citenamefont {Borgnia}\ \emph {et~al.}(2020)\citenamefont
  {Borgnia}, \citenamefont {Kruchkov},\ and\ \citenamefont
  {Slager}}]{Borgnia2020}%
  \BibitemOpen
  \bibfield  {author} {\bibinfo {author} {\bibfnamefont {D.~S.}\ \bibnamefont
  {Borgnia}}, \bibinfo {author} {\bibfnamefont {A.~J.}\ \bibnamefont
  {Kruchkov}},\ and\ \bibinfo {author} {\bibfnamefont {R.-J.}\ \bibnamefont
  {Slager}},\ }\bibfield  {title} {\bibinfo {title} {Non-{Hermitian} boundary
  modes and topology},\ }\href {https://doi.org/10.1103/PhysRevLett.124.056802}
  {\bibfield  {journal} {\bibinfo  {journal} {Phys. Rev. Lett.}\ }\textbf
  {\bibinfo {volume} {124}},\ \bibinfo {pages} {056802} (\bibinfo {year}
  {2020})}\BibitemShut {NoStop}%
\bibitem [{\citenamefont {Kunjummen}\ \emph {et~al.}(2023)\citenamefont
  {Kunjummen}, \citenamefont {Tran}, \citenamefont {Carney},\ and\
  \citenamefont {Taylor}}]{Kunjummen2023}%
  \BibitemOpen
  \bibfield  {author} {\bibinfo {author} {\bibfnamefont {J.}~\bibnamefont
  {Kunjummen}}, \bibinfo {author} {\bibfnamefont {M.~C.}\ \bibnamefont {Tran}},
  \bibinfo {author} {\bibfnamefont {D.}~\bibnamefont {Carney}},\ and\ \bibinfo
  {author} {\bibfnamefont {J.~M.}\ \bibnamefont {Taylor}},\ }\bibfield  {title}
  {\bibinfo {title} {Shadow process tomography of quantum channels},\ }\href
  {https://doi.org/10.1103/PhysRevA.107.042403} {\bibfield  {journal} {\bibinfo
   {journal} {Phys. Rev. A}\ }\textbf {\bibinfo {volume} {107}},\ \bibinfo
  {pages} {042403} (\bibinfo {year} {2023})}\BibitemShut {NoStop}%
\bibitem [{\citenamefont {Wang}\ \emph {et~al.}(2024)\citenamefont {Wang},
  \citenamefont {Song},\ and\ \citenamefont {Wang}}]{Wang2024}%
  \BibitemOpen
  \bibfield  {author} {\bibinfo {author} {\bibfnamefont {H.-Y.}\ \bibnamefont
  {Wang}}, \bibinfo {author} {\bibfnamefont {F.}~\bibnamefont {Song}},\ and\
  \bibinfo {author} {\bibfnamefont {Z.}~\bibnamefont {Wang}},\ }\bibfield
  {title} {\bibinfo {title} {Amoeba formulation of non-bloch band theory in
  arbitrary dimensions},\ }\href {https://doi.org/10.1103/PhysRevX.14.021011}
  {\bibfield  {journal} {\bibinfo  {journal} {Phys. Rev. X}\ }\textbf {\bibinfo
  {volume} {14}},\ \bibinfo {pages} {021011} (\bibinfo {year}
  {2024})}\BibitemShut {NoStop}%
\bibitem [{\citenamefont {Kaneshiro}\ and\ \citenamefont
  {Peters}(2026)}]{Kaneshiro2026}%
  \BibitemOpen
  \bibfield  {author} {\bibinfo {author} {\bibfnamefont {S.}~\bibnamefont
  {Kaneshiro}}\ and\ \bibinfo {author} {\bibfnamefont {R.}~\bibnamefont
  {Peters}},\ }\bibfield  {title} {\bibinfo {title} {Wiener--hopf factorization
  and non-hermitian topology for {Amoeba} formulation in one-dimensional
  multiband systems},\ }\href {https://doi.org/10.1103/s43l-h6z6} {\bibfield
  {journal} {\bibinfo  {journal} {Phys. Rev. Research}\ }\textbf {\bibinfo
  {volume} {8}},\ \bibinfo {pages} {013292} (\bibinfo {year}
  {2026})}\BibitemShut {NoStop}%
\bibitem [{\citenamefont {Yokomizo}\ and\ \citenamefont
  {Murakami}(2019)}]{Yokomizo2019}%
  \BibitemOpen
  \bibfield  {author} {\bibinfo {author} {\bibfnamefont {K.}~\bibnamefont
  {Yokomizo}}\ and\ \bibinfo {author} {\bibfnamefont {S.}~\bibnamefont
  {Murakami}},\ }\bibfield  {title} {\bibinfo {title} {Non-{Bloch} band theory
  of non-{Hermitian} systems},\ }\href
  {https://doi.org/10.1103/PhysRevLett.123.066404} {\bibfield  {journal}
  {\bibinfo  {journal} {Phys. Rev. Lett.}\ }\textbf {\bibinfo {volume} {123}},\
  \bibinfo {pages} {066404} (\bibinfo {year} {2019})}\BibitemShut {NoStop}%
\bibitem [{\citenamefont {Weidemann}\ \emph {et~al.}(2020)\citenamefont
  {Weidemann}, \citenamefont {Kremer}, \citenamefont {Helbig}, \citenamefont
  {Hofmann}, \citenamefont {Stegmaier}, \citenamefont {Greiter}, \citenamefont
  {Thomale},\ and\ \citenamefont {Szameit}}]{Weidemann2020}%
  \BibitemOpen
  \bibfield  {author} {\bibinfo {author} {\bibfnamefont {S.}~\bibnamefont
  {Weidemann}}, \bibinfo {author} {\bibfnamefont {M.}~\bibnamefont {Kremer}},
  \bibinfo {author} {\bibfnamefont {T.}~\bibnamefont {Helbig}}, \bibinfo
  {author} {\bibfnamefont {T.}~\bibnamefont {Hofmann}}, \bibinfo {author}
  {\bibfnamefont {A.}~\bibnamefont {Stegmaier}}, \bibinfo {author}
  {\bibfnamefont {M.}~\bibnamefont {Greiter}}, \bibinfo {author} {\bibfnamefont
  {R.}~\bibnamefont {Thomale}},\ and\ \bibinfo {author} {\bibfnamefont
  {A.}~\bibnamefont {Szameit}},\ }\bibfield  {title} {\bibinfo {title}
  {Topological funneling of light},\ }\href
  {https://doi.org/10.1126/science.aaz8727} {\bibfield  {journal} {\bibinfo
  {journal} {Science}\ }\textbf {\bibinfo {volume} {368}},\ \bibinfo {pages}
  {311} (\bibinfo {year} {2020})}\BibitemShut {NoStop}%
\bibitem [{\citenamefont {Jordan}\ and\ \citenamefont
  {Wigner}(1928)}]{JordanWigner1928}%
  \BibitemOpen
  \bibfield  {author} {\bibinfo {author} {\bibfnamefont {P.}~\bibnamefont
  {Jordan}}\ and\ \bibinfo {author} {\bibfnamefont {E.}~\bibnamefont
  {Wigner}},\ }\bibfield  {title} {\bibinfo {title} {{\"U}ber das {Paulische}
  {\"a}quivalenzverbot},\ }\href {https://doi.org/10.1007/BF01331938}
  {\bibfield  {journal} {\bibinfo  {journal} {Z. Phys.}\ }\textbf {\bibinfo
  {volume} {47}},\ \bibinfo {pages} {631} (\bibinfo {year} {1928})}\BibitemShut
  {NoStop}%
\bibitem [{\citenamefont {Lieb}\ \emph {et~al.}(1961)\citenamefont {Lieb},
  \citenamefont {Schultz},\ and\ \citenamefont
  {Mattis}}]{LiebSchultzMattis1961}%
  \BibitemOpen
  \bibfield  {author} {\bibinfo {author} {\bibfnamefont {E.}~\bibnamefont
  {Lieb}}, \bibinfo {author} {\bibfnamefont {T.}~\bibnamefont {Schultz}},\ and\
  \bibinfo {author} {\bibfnamefont {D.}~\bibnamefont {Mattis}},\ }\bibfield
  {title} {\bibinfo {title} {Two soluble models of an antiferromagnetic
  chain},\ }\href {https://doi.org/10.1016/0003-4916(61)90115-4} {\bibfield
  {journal} {\bibinfo  {journal} {Ann. Phys. (N.Y.)}\ }\textbf {\bibinfo
  {volume} {16}},\ \bibinfo {pages} {407} (\bibinfo {year} {1961})}\BibitemShut
  {NoStop}%
\bibitem [{\citenamefont {Katsura}(1962)}]{Katsura1962}%
  \BibitemOpen
  \bibfield  {author} {\bibinfo {author} {\bibfnamefont {S.}~\bibnamefont
  {Katsura}},\ }\bibfield  {title} {\bibinfo {title} {Statistical mechanics of
  the anisotropic linear {Heisenberg} model},\ }\href
  {https://doi.org/10.1103/PhysRev.127.1508} {\bibfield  {journal} {\bibinfo
  {journal} {Phys. Rev.}\ }\textbf {\bibinfo {volume} {127}},\ \bibinfo {pages}
  {1508} (\bibinfo {year} {1962})}\BibitemShut {NoStop}%
\bibitem [{\citenamefont {Diehl}\ \emph {et~al.}(2008)\citenamefont {Diehl},
  \citenamefont {Micheli}, \citenamefont {Kantian}, \citenamefont {Kraus},
  \citenamefont {B{\"u}chler},\ and\ \citenamefont {Zoller}}]{Diehl2008}%
  \BibitemOpen
  \bibfield  {author} {\bibinfo {author} {\bibfnamefont {S.}~\bibnamefont
  {Diehl}}, \bibinfo {author} {\bibfnamefont {A.}~\bibnamefont {Micheli}},
  \bibinfo {author} {\bibfnamefont {A.}~\bibnamefont {Kantian}}, \bibinfo
  {author} {\bibfnamefont {B.}~\bibnamefont {Kraus}}, \bibinfo {author}
  {\bibfnamefont {H.~P.}\ \bibnamefont {B{\"u}chler}},\ and\ \bibinfo {author}
  {\bibfnamefont {P.}~\bibnamefont {Zoller}},\ }\bibfield  {title} {\bibinfo
  {title} {Quantum states and phases in driven open quantum systems with cold
  atoms},\ }\href {https://doi.org/10.1038/nphys1073} {\bibfield  {journal}
  {\bibinfo  {journal} {Nat. Phys.}\ }\textbf {\bibinfo {volume} {4}},\
  \bibinfo {pages} {878} (\bibinfo {year} {2008})}\BibitemShut {NoStop}%
\bibitem [{\citenamefont {Kraus}\ \emph {et~al.}(2008)\citenamefont {Kraus},
  \citenamefont {B{\"u}chler}, \citenamefont {Diehl}, \citenamefont {Kantian},
  \citenamefont {Micheli},\ and\ \citenamefont {Zoller}}]{Kraus2008}%
  \BibitemOpen
  \bibfield  {author} {\bibinfo {author} {\bibfnamefont {B.}~\bibnamefont
  {Kraus}}, \bibinfo {author} {\bibfnamefont {H.~P.}\ \bibnamefont
  {B{\"u}chler}}, \bibinfo {author} {\bibfnamefont {S.}~\bibnamefont {Diehl}},
  \bibinfo {author} {\bibfnamefont {A.}~\bibnamefont {Kantian}}, \bibinfo
  {author} {\bibfnamefont {A.}~\bibnamefont {Micheli}},\ and\ \bibinfo {author}
  {\bibfnamefont {P.}~\bibnamefont {Zoller}},\ }\bibfield  {title} {\bibinfo
  {title} {Preparation of entangled states by quantum {Markov} processes},\
  }\href {https://doi.org/10.1103/PhysRevA.78.042307} {\bibfield  {journal}
  {\bibinfo  {journal} {Phys. Rev. A}\ }\textbf {\bibinfo {volume} {78}},\
  \bibinfo {pages} {042307} (\bibinfo {year} {2008})}\BibitemShut {NoStop}%
\bibitem [{\citenamefont {Diehl}\ \emph {et~al.}(2011)\citenamefont {Diehl},
  \citenamefont {Rico}, \citenamefont {Baranov},\ and\ \citenamefont
  {Zoller}}]{Diehl2011}%
  \BibitemOpen
  \bibfield  {author} {\bibinfo {author} {\bibfnamefont {S.}~\bibnamefont
  {Diehl}}, \bibinfo {author} {\bibfnamefont {E.}~\bibnamefont {Rico}},
  \bibinfo {author} {\bibfnamefont {M.~A.}\ \bibnamefont {Baranov}},\ and\
  \bibinfo {author} {\bibfnamefont {P.}~\bibnamefont {Zoller}},\ }\bibfield
  {title} {\bibinfo {title} {Topology by dissipation in atomic quantum wires},\
  }\href {https://doi.org/10.1038/nphys2106} {\bibfield  {journal} {\bibinfo
  {journal} {Nat. Phys.}\ }\textbf {\bibinfo {volume} {7}},\ \bibinfo {pages}
  {971} (\bibinfo {year} {2011})}\BibitemShut {NoStop}%
\bibitem [{\citenamefont {Prosen}(2011)}]{Prosen2011XXZ}%
  \BibitemOpen
  \bibfield  {author} {\bibinfo {author} {\bibfnamefont {T.}~\bibnamefont
  {Prosen}},\ }\bibfield  {title} {\bibinfo {title} {Open {XXZ} spin chain:
  Nonequilibrium steady state and a strict bound on ballistic transport},\
  }\href {https://doi.org/10.1103/PhysRevLett.106.217206} {\bibfield  {journal}
  {\bibinfo  {journal} {Phys. Rev. Lett.}\ }\textbf {\bibinfo {volume} {106}},\
  \bibinfo {pages} {217206} (\bibinfo {year} {2011})}\BibitemShut {NoStop}%
\bibitem [{\citenamefont {Altland}\ \emph {et~al.}(2021)\citenamefont
  {Altland}, \citenamefont {Fleischhauer},\ and\ \citenamefont
  {Diehl}}]{Altland2021}%
  \BibitemOpen
  \bibfield  {author} {\bibinfo {author} {\bibfnamefont {A.}~\bibnamefont
  {Altland}}, \bibinfo {author} {\bibfnamefont {M.}~\bibnamefont
  {Fleischhauer}},\ and\ \bibinfo {author} {\bibfnamefont {S.}~\bibnamefont
  {Diehl}},\ }\bibfield  {title} {\bibinfo {title} {Symmetry classes of open
  fermionic quantum matter},\ }\href
  {https://doi.org/10.1103/PhysRevX.11.021037} {\bibfield  {journal} {\bibinfo
  {journal} {Phys. Rev. X}\ }\textbf {\bibinfo {volume} {11}},\ \bibinfo
  {pages} {021037} (\bibinfo {year} {2021})}\BibitemShut {NoStop}%
\bibitem [{\citenamefont {Lau}\ and\ \citenamefont
  {Clerk}(2018)}]{LauClerk2018}%
  \BibitemOpen
  \bibfield  {author} {\bibinfo {author} {\bibfnamefont {H.-K.}\ \bibnamefont
  {Lau}}\ and\ \bibinfo {author} {\bibfnamefont {A.~A.}\ \bibnamefont
  {Clerk}},\ }\bibfield  {title} {\bibinfo {title} {Fundamental limits and
  non-reciprocal approaches in non-{Hermitian} quantum sensing},\ }\href
  {https://doi.org/10.1038/s41467-018-06477-7} {\bibfield  {journal} {\bibinfo
  {journal} {Nat. Commun.}\ }\textbf {\bibinfo {volume} {9}},\ \bibinfo {pages}
  {4320} (\bibinfo {year} {2018})}\BibitemShut {NoStop}%
\bibitem [{\citenamefont {Prosen}\ and\ \citenamefont
  {Ilievski}(2011)}]{ProsenIlievski2011}%
  \BibitemOpen
  \bibfield  {author} {\bibinfo {author} {\bibfnamefont {T.}~\bibnamefont
  {Prosen}}\ and\ \bibinfo {author} {\bibfnamefont {E.}~\bibnamefont
  {Ilievski}},\ }\bibfield  {title} {\bibinfo {title} {Nonequilibrium phase
  transition in a periodically driven {XY} spin chain},\ }\href
  {https://doi.org/10.1103/PhysRevLett.107.060403} {\bibfield  {journal}
  {\bibinfo  {journal} {Phys. Rev. Lett.}\ }\textbf {\bibinfo {volume} {107}},\
  \bibinfo {pages} {060403} (\bibinfo {year} {2011})}\BibitemShut {NoStop}%
\bibitem [{\citenamefont {Barthel}\ and\ \citenamefont
  {Zhang}(2022)}]{BarthelZhang2022}%
  \BibitemOpen
  \bibfield  {author} {\bibinfo {author} {\bibfnamefont {T.}~\bibnamefont
  {Barthel}}\ and\ \bibinfo {author} {\bibfnamefont {Y.}~\bibnamefont
  {Zhang}},\ }\bibfield  {title} {\bibinfo {title} {Solving quasi-free and
  quadratic {Lindblad} master equations for open fermionic and bosonic
  systems},\ }\href {https://doi.org/10.1088/1742-5468/ac8e5c} {\bibfield
  {journal} {\bibinfo  {journal} {J. Stat. Mech.}\ }\textbf {\bibinfo {volume}
  {2022}},\ \bibinfo {pages} {113101} (\bibinfo {year} {2022})}\BibitemShut
  {NoStop}%
\bibitem [{\citenamefont {Song}\ \emph {et~al.}(2019)\citenamefont {Song},
  \citenamefont {Yao},\ and\ \citenamefont {Wang}}]{SongYaoWang2019}%
  \BibitemOpen
  \bibfield  {author} {\bibinfo {author} {\bibfnamefont {F.}~\bibnamefont
  {Song}}, \bibinfo {author} {\bibfnamefont {S.}~\bibnamefont {Yao}},\ and\
  \bibinfo {author} {\bibfnamefont {Z.}~\bibnamefont {Wang}},\ }\bibfield
  {title} {\bibinfo {title} {Non-{Hermitian} skin effect and chiral damping in
  open quantum systems},\ }\href
  {https://doi.org/10.1103/PhysRevLett.123.170401} {\bibfield  {journal}
  {\bibinfo  {journal} {Phys. Rev. Lett.}\ }\textbf {\bibinfo {volume} {123}},\
  \bibinfo {pages} {170401} (\bibinfo {year} {2019})}\BibitemShut {NoStop}%
\bibitem [{\citenamefont {Lee}\ \emph {et~al.}(2013)\citenamefont {Lee},
  \citenamefont {Gopalakrishnan},\ and\ \citenamefont {Lukin}}]{Lee2013}%
  \BibitemOpen
  \bibfield  {author} {\bibinfo {author} {\bibfnamefont {T.~E.}\ \bibnamefont
  {Lee}}, \bibinfo {author} {\bibfnamefont {S.}~\bibnamefont
  {Gopalakrishnan}},\ and\ \bibinfo {author} {\bibfnamefont {M.~D.}\
  \bibnamefont {Lukin}},\ }\bibfield  {title} {\bibinfo {title} {Unconventional
  magnetism via optical pumping of interacting spin systems},\ }\href
  {https://doi.org/10.1103/PhysRevLett.110.257204} {\bibfield  {journal}
  {\bibinfo  {journal} {Phys. Rev. Lett.}\ }\textbf {\bibinfo {volume} {110}},\
  \bibinfo {pages} {257204} (\bibinfo {year} {2013})}\BibitemShut {NoStop}%
\bibitem [{\citenamefont {Mori}\ \emph {et~al.}(2016)\citenamefont {Mori},
  \citenamefont {Kuwahara},\ and\ \citenamefont
  {Saito}}]{MoriKuwaharaSaito2016}%
  \BibitemOpen
  \bibfield  {author} {\bibinfo {author} {\bibfnamefont {T.}~\bibnamefont
  {Mori}}, \bibinfo {author} {\bibfnamefont {T.}~\bibnamefont {Kuwahara}},\
  and\ \bibinfo {author} {\bibfnamefont {K.}~\bibnamefont {Saito}},\ }\bibfield
   {title} {\bibinfo {title} {Rigorous bound on energy absorption and generic
  relaxation in periodically driven quantum systems},\ }\href
  {https://doi.org/10.1103/PhysRevLett.116.120401} {\bibfield  {journal}
  {\bibinfo  {journal} {Phys. Rev. Lett.}\ }\textbf {\bibinfo {volume} {116}},\
  \bibinfo {pages} {120401} (\bibinfo {year} {2016})}\BibitemShut {NoStop}%
\bibitem [{\citenamefont {Nachtergaele}\ and\ \citenamefont
  {Sims}(2006)}]{NachtergaeleSims2006}%
  \BibitemOpen
  \bibfield  {author} {\bibinfo {author} {\bibfnamefont {B.}~\bibnamefont
  {Nachtergaele}}\ and\ \bibinfo {author} {\bibfnamefont {R.}~\bibnamefont
  {Sims}},\ }\bibfield  {title} {\bibinfo {title} {Lieb--robinson bounds and
  the exponential clustering theorem},\ }\href
  {https://doi.org/10.1007/s00220-006-1556-1} {\bibfield  {journal} {\bibinfo
  {journal} {Commun. Math. Phys.}\ }\textbf {\bibinfo {volume} {265}},\
  \bibinfo {pages} {119} (\bibinfo {year} {2006})}\BibitemShut {NoStop}%
\bibitem [{\citenamefont {Nachtergaele}\ \emph {et~al.}(2019)\citenamefont
  {Nachtergaele}, \citenamefont {Sims},\ and\ \citenamefont
  {Young}}]{NachtergaeleSimsYoung2019}%
  \BibitemOpen
  \bibfield  {author} {\bibinfo {author} {\bibfnamefont {B.}~\bibnamefont
  {Nachtergaele}}, \bibinfo {author} {\bibfnamefont {R.}~\bibnamefont {Sims}},\
  and\ \bibinfo {author} {\bibfnamefont {A.}~\bibnamefont {Young}},\ }\bibfield
   {title} {\bibinfo {title} {Quasi-locality bounds for quantum lattice
  systems. {I}. {Lieb--Robinson} bounds, quasi-local maps, and spectral flow
  automorphisms},\ }\href {https://doi.org/10.1063/1.5095769} {\bibfield
  {journal} {\bibinfo  {journal} {J. Math. Phys.}\ }\textbf {\bibinfo {volume}
  {60}},\ \bibinfo {pages} {061101} (\bibinfo {year} {2019})}\BibitemShut
  {NoStop}%
\bibitem [{\citenamefont {Zwolak}\ and\ \citenamefont
  {Vidal}(2004)}]{ZwolakVidal2004}%
  \BibitemOpen
  \bibfield  {author} {\bibinfo {author} {\bibfnamefont {M.}~\bibnamefont
  {Zwolak}}\ and\ \bibinfo {author} {\bibfnamefont {G.}~\bibnamefont {Vidal}},\
  }\bibfield  {title} {\bibinfo {title} {Mixed-state dynamics in
  one-dimensional quantum lattice systems: A time-dependent superoperator
  renormalization algorithm},\ }\href
  {https://doi.org/10.1103/PhysRevLett.93.207205} {\bibfield  {journal}
  {\bibinfo  {journal} {Phys. Rev. Lett.}\ }\textbf {\bibinfo {volume} {93}},\
  \bibinfo {pages} {207205} (\bibinfo {year} {2004})}\BibitemShut {NoStop}%
\bibitem [{\citenamefont {Verstraete}\ \emph {et~al.}(2004)\citenamefont
  {Verstraete}, \citenamefont {Garc{\'i}a-Ripoll},\ and\ \citenamefont
  {Cirac}}]{Verstraete2004}%
  \BibitemOpen
  \bibfield  {author} {\bibinfo {author} {\bibfnamefont {F.}~\bibnamefont
  {Verstraete}}, \bibinfo {author} {\bibfnamefont {J.~J.}\ \bibnamefont
  {Garc{\'i}a-Ripoll}},\ and\ \bibinfo {author} {\bibfnamefont {J.~I.}\
  \bibnamefont {Cirac}},\ }\bibfield  {title} {\bibinfo {title} {Matrix product
  density operators: Simulation of finite-temperature and dissipative
  systems},\ }\href {https://doi.org/10.1103/PhysRevLett.93.207204} {\bibfield
  {journal} {\bibinfo  {journal} {Phys. Rev. Lett.}\ }\textbf {\bibinfo
  {volume} {93}},\ \bibinfo {pages} {207204} (\bibinfo {year}
  {2004})}\BibitemShut {NoStop}%
\bibitem [{\citenamefont {Werner}\ \emph {et~al.}(2016)\citenamefont {Werner},
  \citenamefont {Jaschke}, \citenamefont {Silvi}, \citenamefont {Kliesch},
  \citenamefont {Calarco}, \citenamefont {Eisert},\ and\ \citenamefont
  {Montangero}}]{Werner2016}%
  \BibitemOpen
  \bibfield  {author} {\bibinfo {author} {\bibfnamefont {A.~H.}\ \bibnamefont
  {Werner}}, \bibinfo {author} {\bibfnamefont {D.}~\bibnamefont {Jaschke}},
  \bibinfo {author} {\bibfnamefont {P.}~\bibnamefont {Silvi}}, \bibinfo
  {author} {\bibfnamefont {M.}~\bibnamefont {Kliesch}}, \bibinfo {author}
  {\bibfnamefont {T.}~\bibnamefont {Calarco}}, \bibinfo {author} {\bibfnamefont
  {J.}~\bibnamefont {Eisert}},\ and\ \bibinfo {author} {\bibfnamefont
  {S.}~\bibnamefont {Montangero}},\ }\bibfield  {title} {\bibinfo {title}
  {Positive tensor network approach for simulating open quantum many-body
  systems},\ }\href {https://doi.org/10.1103/PhysRevLett.116.237201} {\bibfield
   {journal} {\bibinfo  {journal} {Phys. Rev. Lett.}\ }\textbf {\bibinfo
  {volume} {116}},\ \bibinfo {pages} {237201} (\bibinfo {year}
  {2016})}\BibitemShut {NoStop}%
\bibitem [{\citenamefont {Cui}\ \emph {et~al.}(2015)\citenamefont {Cui},
  \citenamefont {Cirac},\ and\ \citenamefont {Ba{\~n}uls}}]{Cui2015}%
  \BibitemOpen
  \bibfield  {author} {\bibinfo {author} {\bibfnamefont {J.}~\bibnamefont
  {Cui}}, \bibinfo {author} {\bibfnamefont {J.~I.}\ \bibnamefont {Cirac}},\
  and\ \bibinfo {author} {\bibfnamefont {M.~C.}\ \bibnamefont {Ba{\~n}uls}},\
  }\bibfield  {title} {\bibinfo {title} {Variational matrix product operators
  for the steady state of dissipative quantum systems},\ }\href
  {https://doi.org/10.1103/PhysRevLett.114.220601} {\bibfield  {journal}
  {\bibinfo  {journal} {Phys. Rev. Lett.}\ }\textbf {\bibinfo {volume} {114}},\
  \bibinfo {pages} {220601} (\bibinfo {year} {2015})}\BibitemShut {NoStop}%
\bibitem [{\citenamefont {Wright}\ and\ \citenamefont
  {Trefethen}(2001)}]{WrightTrefethen2001}%
  \BibitemOpen
  \bibfield  {author} {\bibinfo {author} {\bibfnamefont {T.~G.}\ \bibnamefont
  {Wright}}\ and\ \bibinfo {author} {\bibfnamefont {L.~N.}\ \bibnamefont
  {Trefethen}},\ }\bibfield  {title} {\bibinfo {title} {Large-scale computation
  of pseudospectra using {ARPACK} and {Eigs}},\ }\href
  {https://doi.org/10.1137/S106482750037322X} {\bibfield  {journal} {\bibinfo
  {journal} {SIAM J. Sci. Comput.}\ }\textbf {\bibinfo {volume} {23}},\
  \bibinfo {pages} {591} (\bibinfo {year} {2001})}\BibitemShut {NoStop}%
\bibitem [{\citenamefont {Hochstenbach}(2001)}]{Hochstenbach2001}%
  \BibitemOpen
  \bibfield  {author} {\bibinfo {author} {\bibfnamefont {M.~E.}\ \bibnamefont
  {Hochstenbach}},\ }\bibfield  {title} {\bibinfo {title} {A {Jacobi--Davidson}
  type {SVD} method},\ }\href {https://doi.org/10.1137/S1064827500372973}
  {\bibfield  {journal} {\bibinfo  {journal} {SIAM J. Sci. Comput.}\ }\textbf
  {\bibinfo {volume} {23}},\ \bibinfo {pages} {606} (\bibinfo {year}
  {2001})}\BibitemShut {NoStop}%
\end{thebibliography}%


\begin{thebibliography}{23}%
\makeatletter
\providecommand \@ifxundefined [1]{%
 \@ifx{#1\undefined}
}%
\providecommand \@ifnum [1]{%
 \ifnum #1\expandafter \@firstoftwo
 \else \expandafter \@secondoftwo
 \fi
}%
\providecommand \@ifx [1]{%
 \ifx #1\expandafter \@firstoftwo
 \else \expandafter \@secondoftwo
 \fi
}%
\providecommand \natexlab [1]{#1}%
\providecommand \enquote  [1]{``#1''}%
\providecommand \bibnamefont  [1]{#1}%
\providecommand \bibfnamefont [1]{#1}%
\providecommand \citenamefont [1]{#1}%
\providecommand \href@noop [0]{\@secondoftwo}%
\providecommand \href [0]{\begingroup \@sanitize@url \@href}%
\providecommand \@href[1]{\@@startlink{#1}\@@href}%
\providecommand \@@href[1]{\endgroup#1\@@endlink}%
\providecommand \@sanitize@url [0]{\catcode `\\12\catcode `\$12\catcode
  `\&12\catcode `\#12\catcode `\^12\catcode `\_12\catcode `\%12\relax}%
\providecommand \@@startlink[1]{}%
\providecommand \@@endlink[0]{}%
\providecommand \url  [0]{\begingroup\@sanitize@url \@url }%
\providecommand \@url [1]{\endgroup\@href {#1}{\urlprefix }}%
\providecommand \urlprefix  [0]{URL }%
\providecommand \Eprint [0]{\href }%
\providecommand \doibase [0]{https://doi.org/}%
\providecommand \selectlanguage [0]{\@gobble}%
\providecommand \bibinfo  [0]{\@secondoftwo}%
\providecommand \bibfield  [0]{\@secondoftwo}%
\providecommand \translation [1]{[#1]}%
\providecommand \BibitemOpen [0]{}%
\providecommand \bibitemStop [0]{}%
\providecommand \bibitemNoStop [0]{.\EOS\space}%
\providecommand \EOS [0]{\spacefactor3000\relax}%
\providecommand \BibitemShut  [1]{\csname bibitem#1\endcsname}%
\let\auto@bib@innerbib\@empty
\bibitem [{\citenamefont {Trefethen}\ and\ \citenamefont
  {Embree}(2005)}]{TrefethenEmbree2005}%
  \BibitemOpen
  \bibfield  {author} {\bibinfo {author} {\bibfnamefont {L.~N.}\ \bibnamefont
  {Trefethen}}\ and\ \bibinfo {author} {\bibfnamefont {M.}~\bibnamefont
  {Embree}},\ }\href {https://doi.org/10.1515/9780691213101} {\emph {\bibinfo
  {title} {Spectra and Pseudospectra: The Behavior of Nonnormal Matrices and
  Operators}}}\ (\bibinfo  {publisher} {Princeton University Press},\ \bibinfo
  {address} {Princeton, NJ},\ \bibinfo {year} {2005})\BibitemShut {NoStop}%
\bibitem [{\citenamefont {Lindblad}(1976)}]{Lindblad1976}%
  \BibitemOpen
  \bibfield  {author} {\bibinfo {author} {\bibfnamefont {G.}~\bibnamefont
  {Lindblad}},\ }\bibfield  {title} {\bibinfo {title} {On the generators of
  quantum dynamical semigroups},\ }\href {https://doi.org/10.1007/BF01608499}
  {\bibfield  {journal} {\bibinfo  {journal} {Commun. Math. Phys.}\ }\textbf
  {\bibinfo {volume} {48}},\ \bibinfo {pages} {119} (\bibinfo {year}
  {1976})}\BibitemShut {NoStop}%
\bibitem [{\citenamefont {Gorini}\ \emph {et~al.}(1976)\citenamefont {Gorini},
  \citenamefont {Kossakowski},\ and\ \citenamefont {Sudarshan}}]{GKS1976}%
  \BibitemOpen
  \bibfield  {author} {\bibinfo {author} {\bibfnamefont {V.}~\bibnamefont
  {Gorini}}, \bibinfo {author} {\bibfnamefont {A.}~\bibnamefont
  {Kossakowski}},\ and\ \bibinfo {author} {\bibfnamefont {E.~C.~G.}\
  \bibnamefont {Sudarshan}},\ }\bibfield  {title} {\bibinfo {title} {Completely
  positive dynamical semigroups of {N}-level systems},\ }\href
  {https://doi.org/10.1063/1.522979} {\bibfield  {journal} {\bibinfo  {journal}
  {J. Math. Phys.}\ }\textbf {\bibinfo {volume} {17}},\ \bibinfo {pages} {821}
  (\bibinfo {year} {1976})}\BibitemShut {NoStop}%
\bibitem [{\citenamefont {Sz.-Nagy}\ \emph {et~al.}(2010)\citenamefont
  {Sz.-Nagy}, \citenamefont {Foias}, \citenamefont {Bercovici},\ and\
  \citenamefont {K{\'e}rchy}}]{SzNagy2010}%
  \BibitemOpen
  \bibfield  {author} {\bibinfo {author} {\bibfnamefont {B.}~\bibnamefont
  {Sz.-Nagy}}, \bibinfo {author} {\bibfnamefont {C.}~\bibnamefont {Foias}},
  \bibinfo {author} {\bibfnamefont {H.}~\bibnamefont {Bercovici}},\ and\
  \bibinfo {author} {\bibfnamefont {L.}~\bibnamefont {K{\'e}rchy}},\ }\href
  {https://doi.org/10.1007/978-1-4419-6094-8} {\emph {\bibinfo {title}
  {Harmonic Analysis of Operators on Hilbert Space}}},\ \bibinfo {edition}
  {2nd}\ ed.\ (\bibinfo  {publisher} {Springer},\ \bibinfo {address} {New
  York},\ \bibinfo {year} {2010})\BibitemShut {NoStop}%
\bibitem [{\citenamefont {Trefethen}(1997)}]{Trefethen1997}%
  \BibitemOpen
  \bibfield  {author} {\bibinfo {author} {\bibfnamefont {L.~N.}\ \bibnamefont
  {Trefethen}},\ }\bibfield  {title} {\bibinfo {title} {Pseudospectra of linear
  operators},\ }\href {https://doi.org/10.1137/S0036144595295284} {\bibfield
  {journal} {\bibinfo  {journal} {SIAM Rev.}\ }\textbf {\bibinfo {volume}
  {39}},\ \bibinfo {pages} {383} (\bibinfo {year} {1997})}\BibitemShut
  {NoStop}%
\bibitem [{\citenamefont {von Neumann}(1951)}]{vonNeumann1951}%
  \BibitemOpen
  \bibfield  {author} {\bibinfo {author} {\bibfnamefont {J.}~\bibnamefont {von
  Neumann}},\ }\bibfield  {title} {\bibinfo {title} {Eine {Spektraltheorie}
  f{\"u}r allgemeine {Operatoren} eines unit{\"a}ren {Raumes}},\ }\href
  {https://doi.org/10.1002/mana.3210040124} {\bibfield  {journal} {\bibinfo
  {journal} {Math. Nachr.}\ }\textbf {\bibinfo {volume} {4}},\ \bibinfo {pages}
  {258} (\bibinfo {year} {1951})}\BibitemShut {NoStop}%
\bibitem [{\citenamefont {Bingham}\ \emph {et~al.}(1987)\citenamefont
  {Bingham}, \citenamefont {Goldie},\ and\ \citenamefont
  {Teugels}}]{Bingham1987}%
  \BibitemOpen
  \bibfield  {author} {\bibinfo {author} {\bibfnamefont {N.~H.}\ \bibnamefont
  {Bingham}}, \bibinfo {author} {\bibfnamefont {C.~M.}\ \bibnamefont
  {Goldie}},\ and\ \bibinfo {author} {\bibfnamefont {J.~L.}\ \bibnamefont
  {Teugels}},\ }\href {https://doi.org/10.1017/CBO9780511721434} {\emph
  {\bibinfo {title} {Regular Variation}}},\ \bibinfo {series} {Encyclopedia of
  Mathematics and its Applications}, Vol.~\bibinfo {volume} {27}\ (\bibinfo
  {publisher} {Cambridge University Press},\ \bibinfo {address} {Cambridge},\
  \bibinfo {year} {1987})\BibitemShut {NoStop}%
\bibitem [{\citenamefont {Kato}(1995)}]{Kato1995}%
  \BibitemOpen
  \bibfield  {author} {\bibinfo {author} {\bibfnamefont {T.}~\bibnamefont
  {Kato}},\ }\href {https://doi.org/10.1007/978-3-642-66282-9} {\emph {\bibinfo
  {title} {Perturbation Theory for Linear Operators}}},\ Classics in
  Mathematics\ (\bibinfo  {publisher} {Springer},\ \bibinfo {address}
  {Berlin},\ \bibinfo {year} {1995})\BibitemShut {NoStop}%
\bibitem [{\citenamefont {Stewart}(1973)}]{Stewart1973}%
  \BibitemOpen
  \bibfield  {author} {\bibinfo {author} {\bibfnamefont {G.~W.}\ \bibnamefont
  {Stewart}},\ }\bibfield  {title} {\bibinfo {title} {Error and perturbation
  bounds for subspaces associated with certain eigenvalue problems},\ }\href
  {https://doi.org/10.1137/1015095} {\bibfield  {journal} {\bibinfo  {journal}
  {SIAM Rev.}\ }\textbf {\bibinfo {volume} {15}},\ \bibinfo {pages} {727}
  (\bibinfo {year} {1973})}\BibitemShut {NoStop}%
\bibitem [{\citenamefont {Macieszczak}\ \emph {et~al.}(2016)\citenamefont
  {Macieszczak}, \citenamefont {Gu\c{t}\u{a}}, \citenamefont {Lesanovsky},\
  and\ \citenamefont {Garrahan}}]{Macieszczak2016}%
  \BibitemOpen
  \bibfield  {author} {\bibinfo {author} {\bibfnamefont {K.}~\bibnamefont
  {Macieszczak}}, \bibinfo {author} {\bibfnamefont {M.}~\bibnamefont
  {Gu\c{t}\u{a}}}, \bibinfo {author} {\bibfnamefont {I.}~\bibnamefont
  {Lesanovsky}},\ and\ \bibinfo {author} {\bibfnamefont {J.~P.}\ \bibnamefont
  {Garrahan}},\ }\bibfield  {title} {\bibinfo {title} {Towards a theory of
  metastability in open quantum dynamics},\ }\href
  {https://doi.org/10.1103/PhysRevLett.116.240404} {\bibfield  {journal}
  {\bibinfo  {journal} {Phys. Rev. Lett.}\ }\textbf {\bibinfo {volume} {116}},\
  \bibinfo {pages} {240404} (\bibinfo {year} {2016})}\BibitemShut {NoStop}%
\bibitem [{\citenamefont {Ciccarello}\ \emph {et~al.}(2022)\citenamefont
  {Ciccarello}, \citenamefont {Lorenzo}, \citenamefont {Giovannetti},\ and\
  \citenamefont {Palma}}]{Ciccarello2022}%
  \BibitemOpen
  \bibfield  {author} {\bibinfo {author} {\bibfnamefont {F.}~\bibnamefont
  {Ciccarello}}, \bibinfo {author} {\bibfnamefont {S.}~\bibnamefont {Lorenzo}},
  \bibinfo {author} {\bibfnamefont {V.}~\bibnamefont {Giovannetti}},\ and\
  \bibinfo {author} {\bibfnamefont {G.~M.}\ \bibnamefont {Palma}},\ }\bibfield
  {title} {\bibinfo {title} {Quantum collision models: Open system dynamics
  from repeated interactions},\ }\href
  {https://doi.org/10.1016/j.physrep.2022.01.001} {\bibfield  {journal}
  {\bibinfo  {journal} {Phys. Rep.}\ }\textbf {\bibinfo {volume} {954}},\
  \bibinfo {pages} {1} (\bibinfo {year} {2022})}\BibitemShut {NoStop}%
\bibitem [{\citenamefont {Porras}\ and\ \citenamefont
  {Fern{\'a}ndez-Lorenzo}(2019)}]{Porras2019}%
  \BibitemOpen
  \bibfield  {author} {\bibinfo {author} {\bibfnamefont {D.}~\bibnamefont
  {Porras}}\ and\ \bibinfo {author} {\bibfnamefont {S.}~\bibnamefont
  {Fern{\'a}ndez-Lorenzo}},\ }\bibfield  {title} {\bibinfo {title} {Topological
  amplification in photonic lattices},\ }\href
  {https://doi.org/10.1103/PhysRevLett.122.143901} {\bibfield  {journal}
  {\bibinfo  {journal} {Phys. Rev. Lett.}\ }\textbf {\bibinfo {volume} {122}},\
  \bibinfo {pages} {143901} (\bibinfo {year} {2019})}\BibitemShut {NoStop}%
\bibitem [{\citenamefont {Gong}\ \emph {et~al.}(2018)\citenamefont {Gong},
  \citenamefont {Ashida}, \citenamefont {Kawabata}, \citenamefont {Takasan},
  \citenamefont {Higashikawa},\ and\ \citenamefont {Ueda}}]{Gong2018}%
  \BibitemOpen
  \bibfield  {author} {\bibinfo {author} {\bibfnamefont {Z.}~\bibnamefont
  {Gong}}, \bibinfo {author} {\bibfnamefont {Y.}~\bibnamefont {Ashida}},
  \bibinfo {author} {\bibfnamefont {K.}~\bibnamefont {Kawabata}}, \bibinfo
  {author} {\bibfnamefont {K.}~\bibnamefont {Takasan}}, \bibinfo {author}
  {\bibfnamefont {S.}~\bibnamefont {Higashikawa}},\ and\ \bibinfo {author}
  {\bibfnamefont {M.}~\bibnamefont {Ueda}},\ }\bibfield  {title} {\bibinfo
  {title} {Topological phases of non-{Hermitian} systems},\ }\href
  {https://doi.org/10.1103/PhysRevX.8.031079} {\bibfield  {journal} {\bibinfo
  {journal} {Phys. Rev. X}\ }\textbf {\bibinfo {volume} {8}},\ \bibinfo {pages}
  {031079} (\bibinfo {year} {2018})}\BibitemShut {NoStop}%
\bibitem [{\citenamefont {Su}\ \emph {et~al.}(1979)\citenamefont {Su},
  \citenamefont {Schrieffer},\ and\ \citenamefont {Heeger}}]{SSH1979}%
  \BibitemOpen
  \bibfield  {author} {\bibinfo {author} {\bibfnamefont {W.~P.}\ \bibnamefont
  {Su}}, \bibinfo {author} {\bibfnamefont {J.~R.}\ \bibnamefont {Schrieffer}},\
  and\ \bibinfo {author} {\bibfnamefont {A.~J.}\ \bibnamefont {Heeger}},\
  }\bibfield  {title} {\bibinfo {title} {Solitons in polyacetylene},\ }\href
  {https://doi.org/10.1103/PhysRevLett.42.1698} {\bibfield  {journal} {\bibinfo
   {journal} {Phys. Rev. Lett.}\ }\textbf {\bibinfo {volume} {42}},\ \bibinfo
  {pages} {1698} (\bibinfo {year} {1979})}\BibitemShut {NoStop}%
\bibitem [{\citenamefont {Kunst}\ \emph {et~al.}(2018)\citenamefont {Kunst},
  \citenamefont {Edvardsson}, \citenamefont {Budich},\ and\ \citenamefont
  {Bergholtz}}]{Kunst2018}%
  \BibitemOpen
  \bibfield  {author} {\bibinfo {author} {\bibfnamefont {F.~K.}\ \bibnamefont
  {Kunst}}, \bibinfo {author} {\bibfnamefont {E.}~\bibnamefont {Edvardsson}},
  \bibinfo {author} {\bibfnamefont {J.~C.}\ \bibnamefont {Budich}},\ and\
  \bibinfo {author} {\bibfnamefont {E.~J.}\ \bibnamefont {Bergholtz}},\
  }\bibfield  {title} {\bibinfo {title} {Biorthogonal bulk-boundary
  correspondence in non-{Hermitian} systems},\ }\href
  {https://doi.org/10.1103/PhysRevLett.121.026808} {\bibfield  {journal}
  {\bibinfo  {journal} {Phys. Rev. Lett.}\ }\textbf {\bibinfo {volume} {121}},\
  \bibinfo {pages} {026808} (\bibinfo {year} {2018})}\BibitemShut {NoStop}%
\bibitem [{\citenamefont {Wang}\ \emph {et~al.}(2024)\citenamefont {Wang},
  \citenamefont {Song},\ and\ \citenamefont {Wang}}]{Wang2024}%
  \BibitemOpen
  \bibfield  {author} {\bibinfo {author} {\bibfnamefont {H.-Y.}\ \bibnamefont
  {Wang}}, \bibinfo {author} {\bibfnamefont {F.}~\bibnamefont {Song}},\ and\
  \bibinfo {author} {\bibfnamefont {Z.}~\bibnamefont {Wang}},\ }\bibfield
  {title} {\bibinfo {title} {Amoeba formulation of non-bloch band theory in
  arbitrary dimensions},\ }\href {https://doi.org/10.1103/PhysRevX.14.021011}
  {\bibfield  {journal} {\bibinfo  {journal} {Phys. Rev. X}\ }\textbf {\bibinfo
  {volume} {14}},\ \bibinfo {pages} {021011} (\bibinfo {year}
  {2024})}\BibitemShut {NoStop}%
\bibitem [{\citenamefont {Kaneshiro}\ and\ \citenamefont
  {Peters}(2026)}]{Kaneshiro2026}%
  \BibitemOpen
  \bibfield  {author} {\bibinfo {author} {\bibfnamefont {S.}~\bibnamefont
  {Kaneshiro}}\ and\ \bibinfo {author} {\bibfnamefont {R.}~\bibnamefont
  {Peters}},\ }\bibfield  {title} {\bibinfo {title} {Wiener--hopf factorization
  and non-hermitian topology for {Amoeba} formulation in one-dimensional
  multiband systems},\ }\href {https://doi.org/10.1103/s43l-h6z6} {\bibfield
  {journal} {\bibinfo  {journal} {Phys. Rev. Research}\ }\textbf {\bibinfo
  {volume} {8}},\ \bibinfo {pages} {013292} (\bibinfo {year}
  {2026})}\BibitemShut {NoStop}%
\bibitem [{\citenamefont {Bhatia}\ and\ \citenamefont
  {Rosenthal}(1997)}]{BhatiaRosenthal1997}%
  \BibitemOpen
  \bibfield  {author} {\bibinfo {author} {\bibfnamefont {R.}~\bibnamefont
  {Bhatia}}\ and\ \bibinfo {author} {\bibfnamefont {P.}~\bibnamefont
  {Rosenthal}},\ }\bibfield  {title} {\bibinfo {title} {How and why to solve
  the operator equation $ax-xb=y$},\ }\href
  {https://doi.org/10.1112/S0024609396001828} {\bibfield  {journal} {\bibinfo
  {journal} {Bull. London Math. Soc.}\ }\textbf {\bibinfo {volume} {29}},\
  \bibinfo {pages} {1} (\bibinfo {year} {1997})}\BibitemShut {NoStop}%
\bibitem [{\citenamefont {Jordan}\ and\ \citenamefont
  {Wigner}(1928)}]{JordanWigner1928}%
  \BibitemOpen
  \bibfield  {author} {\bibinfo {author} {\bibfnamefont {P.}~\bibnamefont
  {Jordan}}\ and\ \bibinfo {author} {\bibfnamefont {E.}~\bibnamefont
  {Wigner}},\ }\bibfield  {title} {\bibinfo {title} {{\"U}ber das {Paulische}
  {\"a}quivalenzverbot},\ }\href {https://doi.org/10.1007/BF01331938}
  {\bibfield  {journal} {\bibinfo  {journal} {Z. Phys.}\ }\textbf {\bibinfo
  {volume} {47}},\ \bibinfo {pages} {631} (\bibinfo {year} {1928})}\BibitemShut
  {NoStop}%
\bibitem [{\citenamefont {Lieb}\ \emph {et~al.}(1961)\citenamefont {Lieb},
  \citenamefont {Schultz},\ and\ \citenamefont
  {Mattis}}]{LiebSchultzMattis1961}%
  \BibitemOpen
  \bibfield  {author} {\bibinfo {author} {\bibfnamefont {E.}~\bibnamefont
  {Lieb}}, \bibinfo {author} {\bibfnamefont {T.}~\bibnamefont {Schultz}},\ and\
  \bibinfo {author} {\bibfnamefont {D.}~\bibnamefont {Mattis}},\ }\bibfield
  {title} {\bibinfo {title} {Two soluble models of an antiferromagnetic
  chain},\ }\href {https://doi.org/10.1016/0003-4916(61)90115-4} {\bibfield
  {journal} {\bibinfo  {journal} {Ann. Phys. (N.Y.)}\ }\textbf {\bibinfo
  {volume} {16}},\ \bibinfo {pages} {407} (\bibinfo {year} {1961})}\BibitemShut
  {NoStop}%
\bibitem [{\citenamefont {Katsura}(1962)}]{Katsura1962}%
  \BibitemOpen
  \bibfield  {author} {\bibinfo {author} {\bibfnamefont {S.}~\bibnamefont
  {Katsura}},\ }\bibfield  {title} {\bibinfo {title} {Statistical mechanics of
  the anisotropic linear {Heisenberg} model},\ }\href
  {https://doi.org/10.1103/PhysRev.127.1508} {\bibfield  {journal} {\bibinfo
  {journal} {Phys. Rev.}\ }\textbf {\bibinfo {volume} {127}},\ \bibinfo {pages}
  {1508} (\bibinfo {year} {1962})}\BibitemShut {NoStop}%
\bibitem [{\citenamefont {Prosen}(2008)}]{Prosen2008}%
  \BibitemOpen
  \bibfield  {author} {\bibinfo {author} {\bibfnamefont {T.}~\bibnamefont
  {Prosen}},\ }\bibfield  {title} {\bibinfo {title} {Third quantization: a
  general method to solve master equations for quadratic open {Fermi}
  systems},\ }\href {https://doi.org/10.1088/1367-2630/10/4/043026} {\bibfield
  {journal} {\bibinfo  {journal} {New J. Phys.}\ }\textbf {\bibinfo {volume}
  {10}},\ \bibinfo {pages} {043026} (\bibinfo {year} {2008})}\BibitemShut
  {NoStop}%
\bibitem [{\citenamefont {Barthel}\ and\ \citenamefont
  {Zhang}(2022)}]{BarthelZhang2022}%
  \BibitemOpen
  \bibfield  {author} {\bibinfo {author} {\bibfnamefont {T.}~\bibnamefont
  {Barthel}}\ and\ \bibinfo {author} {\bibfnamefont {Y.}~\bibnamefont
  {Zhang}},\ }\bibfield  {title} {\bibinfo {title} {Solving quasi-free and
  quadratic {Lindblad} master equations for open fermionic and bosonic
  systems},\ }\href {https://doi.org/10.1088/1742-5468/ac8e5c} {\bibfield
  {journal} {\bibinfo  {journal} {J. Stat. Mech.}\ }\textbf {\bibinfo {volume}
  {2022}},\ \bibinfo {pages} {113101} (\bibinfo {year} {2022})}\BibitemShut
  {NoStop}%
\end{thebibliography}%


\begin{thebibliography}{0}%
\makeatletter
\providecommand \@ifxundefined [1]{%
 \@ifx{#1\undefined}
}%
\providecommand \@ifnum [1]{%
 \ifnum #1\expandafter \@firstoftwo
 \else \expandafter \@secondoftwo
 \fi
}%
\providecommand \@ifx [1]{%
 \ifx #1\expandafter \@firstoftwo
 \else \expandafter \@secondoftwo
 \fi
}%
\providecommand \natexlab [1]{#1}%
\providecommand \enquote  [1]{``#1''}%
\providecommand \bibnamefont  [1]{#1}%
\providecommand \bibfnamefont [1]{#1}%
\providecommand \citenamefont [1]{#1}%
\providecommand \href@noop [0]{\@secondoftwo}%
\providecommand \href [0]{\begingroup \@sanitize@url \@href}%
\providecommand \@href[1]{\@@startlink{#1}\@@href}%
\providecommand \@@href[1]{\endgroup#1\@@endlink}%
\providecommand \@sanitize@url [0]{\catcode `\\12\catcode `\$12\catcode
  `\&12\catcode `\#12\catcode `\^12\catcode `\_12\catcode `\%12\relax}%
\providecommand \@@startlink[1]{}%
\providecommand \@@endlink[0]{}%
\providecommand \url  [0]{\begingroup\@sanitize@url \@url }%
\providecommand \@url [1]{\endgroup\@href {#1}{\urlprefix }}%
\providecommand \urlprefix  [0]{URL }%
\providecommand \Eprint [0]{\href }%
\providecommand \doibase [0]{https://doi.org/}%
\providecommand \selectlanguage [0]{\@gobble}%
\providecommand \bibinfo  [0]{\@secondoftwo}%
\providecommand \bibfield  [0]{\@secondoftwo}%
\providecommand \translation [1]{[#1]}%
\providecommand \BibitemOpen [0]{}%
\providecommand \bibitemStop [0]{}%
\providecommand \bibitemNoStop [0]{.\EOS\space}%
\providecommand \EOS [0]{\spacefactor3000\relax}%
\providecommand \BibitemShut  [1]{\csname bibitem#1\endcsname}%
\let\auto@bib@innerbib\@empty
\end{thebibliography}%
\end{document}